\documentclass[12pt]{article}

\usepackage{amsmath}
\usepackage{amsthm}
\usepackage{a4wide}
\usepackage[all,2cell]{xy}
\CompileMatrices 
\UseAllTwocells
\SilentMatrices
\usepackage{stmaryrd}
\usepackage{cite}
\usepackage{float}
\usepackage{fancybox}
\usepackage{tikz-cd}
\usepackage{hyperref}

\newcommand{\cgr}[2][scale=0.45]{\raisebox{0.1em}{\begingroup
\setbox0=\hbox{\includegraphics[#1]{graffles/#2}}%
\parbox{\wd0}{\box0}\endgroup}}

\newcommand{\dup}{\Delta}
\newcommand{\dis}{!}
\newcommand{\codup}{\nabla}
\newcommand{\codis}{?}

\newcommand{\Frob}[1]{\mathcal{F}_{#1} }
\newcommand{\Law}[1]{\mathcal{L}_{#1} }
\newcommand{\freePROP}[1]{\mathcal{P}_{#1} }
\newcommand{\freeordPROP}[1]{\freePROP{#1} }
\newcommand{\freeLPPROP}[1]{\mathcal{LP}_{#1} }

\newcommand{\Theory}[1]{\mathbb{#1}}
\newcommand\T{\Theory{T}} % arbitrary Theory

\newcommand{\FTheory}[1]{\mathbb{#1}^+} %Frobenius Theory corresponding to a Cartesian one

\newcommand{\CMtheory}{\Theory{CM}} %Theory of monoid
\newcommand{\CCtheory}{\Theory{CC}} %Theory of comonoid
\newcommand{\Bialgtheory}{\Theory{B}} %Theory of bialgebra
\newcommand{\Laxbialgtheory}{\Theory{LB}} %Theory of Lax bialgebra
\newcommand{\Oplaxbialgtheory}{\Theory{OLB}} %Theory of Op Lax bialgebra
\newcommand{\Frobtheory}{\Theory{F}} %Theory of FROBENIUS
\newcommand{\Hopftheory}{\Theory{H}} %Theory of HOPF
\newcommand{\AGtheory}{\Theory{AG}}% Theory of Abelian Group

\newcommand{\CMtheoryF}{\FTheory{CM}} %Theory of monoid FROBENIUS
\newcommand{\AGtheoryF}{\FTheory{AG}}% Theory of Abelian Group

\def \Mon {\freePROP{\CMtheory}}  % PROP of monoids
\def \Com {\freePROP{\CCtheory}} % PROP of comonoids
 \newcommand{\B}{\freePROP{\Bialgtheory}}  % prop of bialgebras
\newcommand{\ModtheoryF}[1]{\Theory{MOD}_#1^+}
\newcommand{\ModRtheoryF}{\ModtheoryF{\mathsf{R}}}

\newcommand{\cat}[1]{\mathbf{#1}}
\def \catC {\cat{C}}%{\mf{C}}
\newcommand\funF{F}
\newcommand\funG{G}

\newcommand{\Map}[1]{\mathsf{Map(#1)}} % Category of Maps

\newcommand{\model}[1]{\mathsf{Mod_{SMT}}(#1)} % models of a PRO
\newcommand{\carmodel}[1]{\mathsf{Mod_{CAR}}(#1)} % models of a PRO
\newcommand{\frobmodel}[1]{\mathsf{Mod_{FROB}}(#1)} % models of a PRO

\DeclareMathOperator{\id}{id}

\newcommand{\Defeq}% definitional equality
 {\stackrel{\mathrm{def}}{=}}

\newcommand{\stran}{\raise1pt\hbox{$\centerdot$}}

\newcommand{\rring}[1]{\ensuremath{\mathsf{#1}}}

\newcommand{\N}{\rring{N}}

\newcommand{\Ra}{\Rightarrow}

\newcommand{\Set}{\cat{Set}}
\newcommand{\Rel}{\cat{Rel}}

\renewcommand{\emptyset}{\varnothing}

\newcommand{\ladj}[2]{\ar@/^/[#1]^-{#2} \ar@{}[#1]|-%

{\ifthenelse{\equal{#1}{r}}{\bot}{%

{\ifthenelse{\equal{#1}{rr}}{\bot}{%

{\ifthenelse{\equal{#1}{l}}{\top}{%

{\ifthenelse{\equal{#1}{u}}{\dashv}{%

{\vdash}}}}}}}}}}

\newcommand{\radj}[2]{\ar@/^/[#1]^-{#2}}

\newcommand{\radjff}[2]{\ar@{_{(}->}[#1]^{#2}}

\newcommand{\pullbacktop}[4]{%

{#1} \ar@/_/[ddr]_{#4} \ar@/^/[drr]^{#2}% 

\ar@{.>}[dr]|-{#3} \\}

\newtheorem{clm}{Claim}[section]

\newtheorem{lem}[clm]{Lemma}
\newtheorem{thm}[clm]{Theorem}

\newcounter{ncomm}

\newcommand{\ltsred}[1]
{ \setbox0=\hbox{$\ {}^{#1}\ $}
  \setbox1=\hbox{$\longrightarrow$}
  \loop\setbox1=\hbox{$-$\kern-0.3em\unhbox1}\ifdim\wd1<\wd0\repeat
  \hbox{$\ \ \mathop{\box1}\limits^{#1}\ \ $}
}

\newcommand{\arx}[2]{\!\xymatrix@=15pt{\ar[r]^{{#1}}_{{#2}}&}\!}

\newlength{\mylength}

{\setlength{\fboxsep}{15pt} 
\setlength{\mylength}{\linewidth}% 
\addtolength{\mylength}{-2\fboxsep}% 
\addtolength{\mylength}{-2\fboxrule}% 
\Sbox 
\minipage{\mylength}% 
\setlength{\abovedisplayskip}{0pt}% 
\setlength{\belowdisplayskip}{0pt}% 
}% 
{\endminipage\endSbox 
\[\fbox{\TheSbox}\]}

\usepackage[english]{babel}
\usepackage{enumerate}
\usepackage[shortlabels]{enumitem}
\usepackage{amsmath}
\usepackage{amssymb}
\usepackage{graphicx,epsfig,color}
\usepackage[color]{xypic}
\xyoption{rotate}
\input xy
\xyoption{all}
\xyoption{2cell}
\UseAllTwocells
\usepackage{comment}
\usepackage{etoolbox}
\usepackage{multicol}
\newtheorem{theorem}{Theorem}[section]
\newtheorem{proposition}[theorem]{Proposition}
\newtheorem{corollary}[theorem]{Corollary}
\newtheorem{lemma}[theorem]{Lemma}

\newtheorem{definition}[theorem]{Definition}

\newtheorem{example}[theorem]{Example}

\newtheorem{remark}[theorem]{Remark}

\usepackage[layout=margin]{fixme}

\FXRegisterAuthor{fz}{afz}{Fab}%Fabio
\FXRegisterAuthor{ps}{aps}{Chavez}%Pawel
\FXRegisterAuthor{fb}{afb}{Trotski}%Filippo

\def \tns {\oplus}
\def \: {\colon}
\def \poi {\,\ensuremath{;}\,}
\def \N {\mathbb{N}}
\def \Set {\mathbf{Set}}% Set
\usepackage{bbm}
\newcommand\OpDiag{\lower4pt\hbox{$\includegraphics[width=24pt]{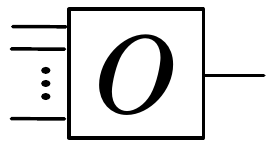}$}}
\newcommand\diagD{\lower4pt\hbox{$\includegraphics[width=24pt]{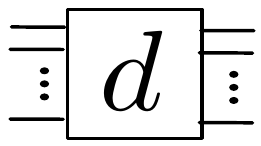}$}}
\newcommand\Idnet{\lower3pt\hbox{$\includegraphics[width=20pt]{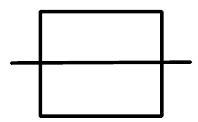}$}}
\newcommand\symNet{\lower3pt\hbox{$\includegraphics[width=20pt]{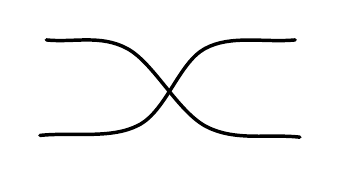}$}}
\newcommand{\ZeronetT}{\lower4pt\hbox{$\includegraphics[width=14pt]{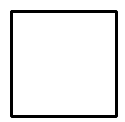}$}}

\newcommand\Bmult{\lower5pt\hbox{$\includegraphics[width=27pt]{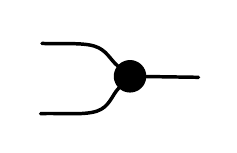}$}}
\newcommand\Bcomult{\lower5pt\hbox{$\includegraphics[width=27pt]{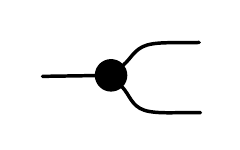}$}}
\newcommand\Bunit{\lower3pt\hbox{$\includegraphics[width=18pt]{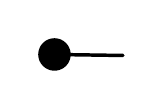}$}}
\newcommand\Bcounit{\lower5pt\hbox{$\includegraphics[width=23pt]{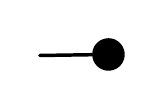}$}}

\newcommand\Bmultn{\lower6pt\hbox{$\includegraphics[width=23pt]{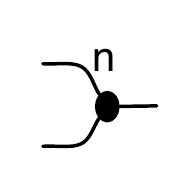}$}}
\newcommand\Bcomultn{\lower8pt\hbox{$\includegraphics[width=26pt]{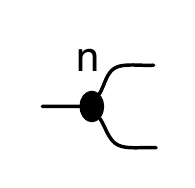}$}}
\newcommand\Bunitn{\lower5pt\hbox{$\includegraphics[width=22pt]{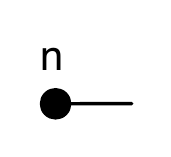}$}}
\newcommand\Bcounitn{\lower3pt\hbox{$\includegraphics[width=18pt]{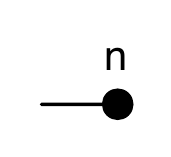}$}}

\newcommand\Wmult{\lower5pt\hbox{$\includegraphics[width=27pt]{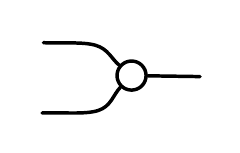}$}}
\newcommand\Wcomult{\lower5pt\hbox{$\includegraphics[width=27pt]{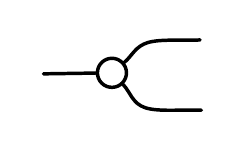}$}}
\newcommand\Wunit{\lower5pt\hbox{$\includegraphics[width=23pt]{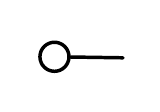}$}}
\newcommand\Wcounit{\lower5pt\hbox{$\includegraphics[width=16pt]{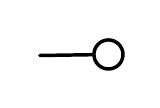}$}}

\newcommand\twoBcounit{\lower5pt\hbox{$\includegraphics[width=15pt]{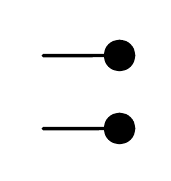}$}}

\title{Functorial Semantics for Relational Theories}
\author{Filippo Bonchi, Dusko Pavlovic and Pawe{\l} Soboci{\'n}ski}

\begin{document}
\maketitle
\abstract{We introduce the concept of \emph{Frobenius theory} as a generalisation of Lawvere's functorial semantics approach to categorical universal algebra. Whereas the universe for models of Lawvere theories is the category of sets and functions, or more generally cartesian categories,  Frobenius theories take their models in the category of sets and relations, or more generally in cartesian bicategories.}

\section{Introduction and roadmap}

There has been a recent explosion of interest in algebraic structures borne by objects of symmetric monoidal categories, with applications in quantum foundations, concurrency theory, control theory, linguistics and database theory, amongst others. In several cases these ``resource-sensitive'' algebraic theories are presented using generators and equations. Moreover, many contain Frobenius algebra as a sub-theory, which yields a self-dual compact closed structure and gives the theories a relational flavour (e.g. a dagger operation, which one can often think semantically as giving the \emph{opposite relation}). In this paper we propose a \emph{categorical universal algebra} for such monoidal theories, generalising functorial semantics, the classical approach due to Lawvere. A canonical notion of model clarifies the conceptual landscape, and is a useful tool for the study of the algebraic theories themselves. For example: how can one show that a particular equation does \emph{not} hold in a theory? One way is to find a model where the equation does not hold.

\subsection{Functorial semantics}

Lawvere categories (a.k.a. finite product theories) are a standard setting for classical universal algebra. Syntactically speaking, terms are trees where some leaves are labelled with variables, and these can be copied and discarded arbitrarily. Categorically, this means a finite product structure, and (classical) models are \emph{product preserving functors}.
In particular, the classical notion of model as a set, together with appropriate $n$-ary functions, satisfying the requisite equations, is captured by a product preserving functor from the corresponding Lawvere category to the category of sets and functions $\Law{}\to\Set$. This methodology is well-known as \emph{functorial semantics}. 
\begin{center}
\begin{tabular}{ | c | c | }
\hline
\textbf{Specification} & algebraic theory \\
\textbf{Syntax} & trees \\
\textbf{Category} & Lawvere category (finite product category) \\
\textbf{Models} & product preserving functors \\
\textbf{Homomorphisms} & natural transformations \\ \hline
\end{tabular}
\end{center}
Thus commutative monoids are exactly the product preserving functors from the Lawvere category of commutative monoids, abelian groups the product preserving functors from the Lawvere category of abelian groups, etc.
Moreover, the usual notion of homomorphism between models is given by natural transformations between models-as-functors.

In applications, classical algebraic theories are often not the right fit. Sometimes this is because an underlying data type is not classical, e.g. qubits, that cannot be copied. Other times it's because one needs to be explicit about the actual copying and discarding being carried out \emph{as (co)algebraic operations}, instead of relying on an implicit cartesian structure. That is, we require a \emph{resource sensitive syntax}. In practice, this means replacing algebraic theories with symmetric monoidal theories (SMTs), trees with string diagrams, cartesian product with symmetric monoidal product (Lawvere categories with props), and product preserving functors with monoidal functors. This suggests an updated table:
\begin{center}
\begin{tabular}{ | c | c |}
\hline
\textbf{Specification} & symmetric monoidal theory (SMT) \\
\textbf{Syntax} & string diagrams \\
\textbf{Category} & prop \\
\textbf{Models} & symmetric monoidal functors \\
\textbf{Homomorphisms} & monoidal natural transformations \\ \hline
\end{tabular}
\end{center}

Props are symmetric strict monoidal categories with objects the natural numbers, such that $m \oplus n=m+n$. Of course, any Lawvere category is an example of a prop, since the cartesian structure induces a canonical symmetry. Arrows of (freely generated) props seem, therefore, to offer an attractive solution to the quest for resource sensitive syntax. Given that the underlying monoidal product is not assumed to be cartesian, props give the possibility of considering bona fide operations with co-arities other than one, e.g. the structure (comultiplication and counit) of a comonoid. In fact, comonoids are the bridge between the classical and the resource-sensitive.

Indeed, given an algebraic theory, we can consider it as a symmetric monoidal theory by encoding the cartesian structure. This amounts to introducing a commutative comonoid (copying) and equations making all other operations comonoid homomorphisms. 
\begin{equation}\label{eq:comonoidhom}
\lower25pt\hbox{$\includegraphics[height=2cm]{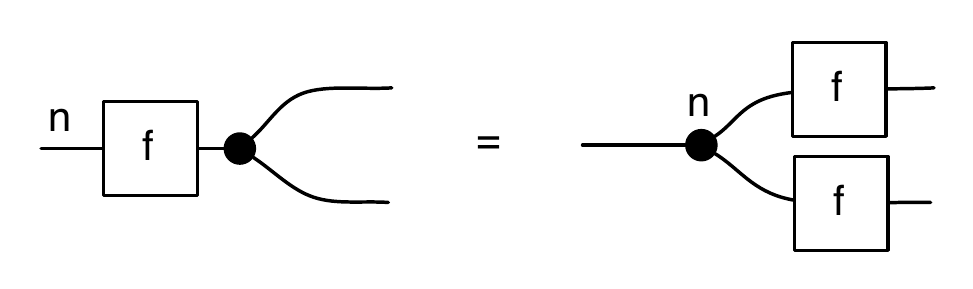}$}
\qquad
\lower15pt\hbox{$\includegraphics[height=1.2cm]{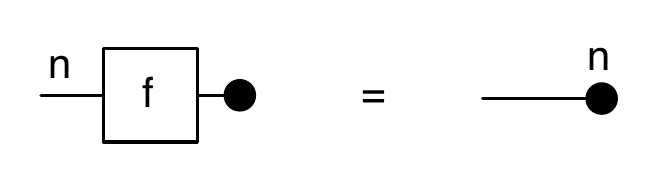}$}
\end{equation}

This means that, as props, the following are actually isomorphic:
\begin{center}
Lawvere category $\Law{\CMtheory}$ of commutative monoids 
$\CMtheory$ \\
  $\cong$    \\
prop $\freePROP{\Bialgtheory}$ of (co/commutative) bialgebras $\Bialgtheory$ 
\end{center}
\begin{center}
Lawvere category $\Law{\AGtheory}$  of abelian groups $\AGtheory$ \\
  $\cong$ \\  
 prop $\freePROP{\Hopftheory}$ of (co/commutative) Hopf algebras
\end{center}
Thus, in effect, bialgebras are what one gets if by considering classical commutative monoids and taking resource sensitivity seriously. Similarly, Hopf algebras can be seen as abelian groups in a ``resource sensitive'' universe.

The structure of props suggests that, for models, we ought to look at symmetric monoidal functors. Indeed, considering the category of sets and functions $\Set$ with cartesian product as monoidal product as codomain, the symmetric monoidal functors from the prop SMT of commutative monoids are in bijective correspondence with ordinary commutative monoids. Here it is the cartesianity of $\Set$ that means that, although the theory is non-cartesian, the models are classic. Similarly, commutative monoids are captured by symmetric monoidal functors $\mathbf{B}\to\Set$ -- it is not difficult to show that the only comonoid action on a set is given by the diagonal, so the ``copying'' comonoid structure is uniquely determined in any $\Set$-model of $\mathbf{B}$.

\subsection{Relations as a universe for models}
Our goal is to study algebras of relations (e.g. relational algebra, allegories, Kleene algebra, automata, labelled transition systems, ...), important in computer science.
Thus, we are interested in developing a theory of functorial semantics that takes its classical models in $\Rel_\times$ (object = sets, arrows = relations, and the subscript indicates that we take cartesian product as monoidal product). 

Here mere SMTs and monoidal functors are not enough to characterise commutative monoids. Considering the SMT of commutative monoids, monoidal functors to $\Rel_\times$ are not guaranteed to give a functional monoid action: e.g., one could map the monoid action to the opposite of the diagonal relation.
Considering the SMT of bialgebras fails also: there is no guarantee that the comultiplication maps to the diagonal. For a concrete example, consider 
\begin{multline}\label{eq:problematicmodel}
1 \mapsto \N \quad
\nabla \mapsto + = \lambda (x,y):X\times X.\, x+y \\
\top \mapsto \lambda x:\{\star\}: 0 \quad \Delta \mapsto \nabla^{op} \quad \bot \mapsto \top^{op}
\end{multline}
here the structure of addition (i.e. the fact that $\N$ is a rig--a ring without negatives) on $\N$ ensures that the bialgebra equations are satisfied, in particular,
if $x+y=0$ then $x=y=0$.

We could require -- for props that have a commutative comonoid structure that defines a product, that the product structure ought to be preserved, so that $\Delta$ is mapped to the diagonal. Unfortunately, this would preclude considering $\Rel_\times$ as a universe of models, since the monoidal product in $\Rel_\times$ is \emph{not} a cartesian product (recall that $\Rel$ actually has $+$ as biproduct). Indeed, when interpreted in $\Rel$, the general form of equations~\eqref{eq:comonoidhom}
\[
\lower20pt\hbox{$\includegraphics[height=1.75cm]{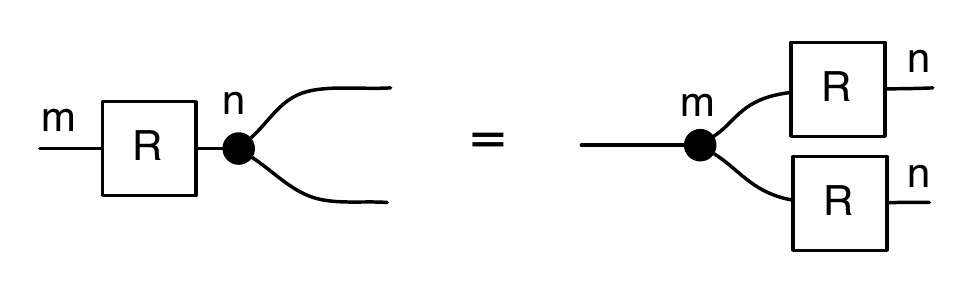}$}
\qquad
\lower15pt\hbox{$\includegraphics[height=1.2cm]{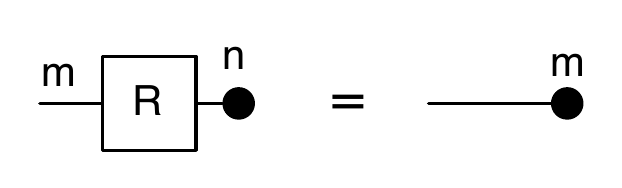}$}
\]
tells us, respectively, that $R$ is single-valued and total, which is true only of those relations that are (the graphs of total) functions. 

A crucial observation to make at this point is that products play \emph{two} different roles in functorial semantics \`{a} la Lawvere. The first is resource insensitivity---i.e.\ an assumption about the classical nature of the underlying data--- which we would like to discard. The second is \emph{preservation of arities} -- the idea that one should be able to specify algebraic operations and have that structure borne by a set (or more generally, an object in some category). We would like to keep this second role, and here we take advantage of the the notion of \emph{lax product}, which 
Carboni and Walters~\cite{Carboni1987}  identified as important for the algebra of categories of relations. Indeed, the monoidal product of $\Rel_\times$ satisfies a lax universal property. In practice this turns out to be, bureaucratically speaking, quite a tame notion of laxness: the 2-dimensional structure of $\Rel_\times$ is posetal (set inclusion), and indeed we will concentrate on the poset-enriched case. At specification level, this means that it's natural to introduce \emph{inequations} between terms.

In fact, the most we can say about relations $R$, in general, is that they are lax comonoid homomorphisms:
\begin{equation}\label{eq:cartesianbicats}
\lower20pt\hbox{$\includegraphics[height=1.7cm]{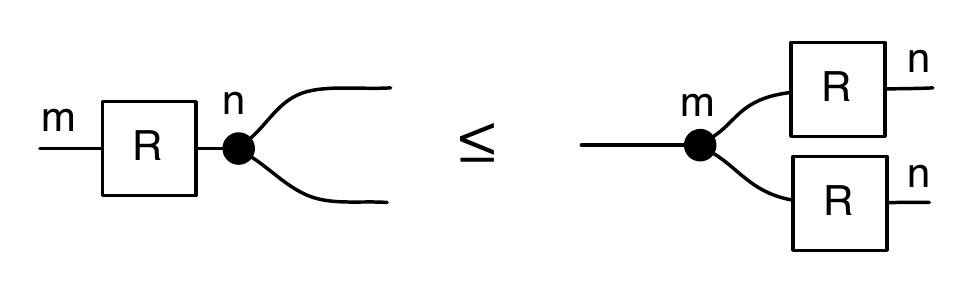}$}
\quad
\lower13pt\hbox{$\includegraphics[height=1.2cm]{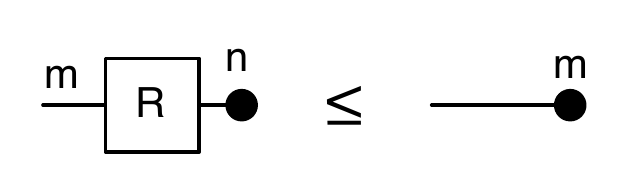}$}
\end{equation}
If \emph{all} arrows are lax homomorphisms in this sense then the monoidal product is a lax product. Crucially, \eqref{eq:cartesianbicats} holds in $\Rel_\times$, where $X$-comultiplication is the diagonal relation 
$\{\, (x,(x,x)) \,|\, x\in X \,\}$, that is, the graph of the diagonal function. Notice that, in $\Rel_\times$, the inequations~\eqref{eq:cartesianbicats} are strict exactly when, respectively, $R$ is not single valued and not total.

\subsection{Lax product theories}
By a \emph{lax product theory} we mean a generalisation of SMT that replaces equations with inequations,
and includes a chosen commutative comonoid structure. Moreover, we require inequations~\eqref{eq:cartesianbicats} that say that all other data is a lax comonoid homomorphism. We call the comonoid structure together with the aforementioned inequations a \emph{lax product structure}.
Every lax product theory leads to a free ordered prop (a prop enriched in the category or posets and monotone maps), where~\eqref{eq:cartesianbicats} ensure that the monoidal product is a lax product, in the bicategorical sense.
\begin{center}
\begin{tabular}{ | c | c | }
\hline
\textbf{Specification} & lax product theory \\
\textbf{Syntax} & string diagrams \\
\textbf{Category} & lax product 2-prop \\
\textbf{Models} & lax product structure preserving functors \\
\textbf{Homomorphisms} & monoidal lax natural transformations \\
\hline
\end{tabular}
\end{center}

Note that both the SMT of bialgebras and Hopf algebras are lax product theories (each equation is replaced by two inequations). And we now obtain a satisfactory ``resource sensitive'' 
generalisation of Lawvere's functorial semantics to $\Rel$-models. For example, the models of the SMT of bialgebras, given by lax product structure preserving functors to $\Rel_\times$, are exactly commutative monoids. This may appear surprising, since we are mapping to $\Rel_\times$, one could expect that $\nabla$ may map to an arbitrary relation. Instead, the fact that we need to preserve lax products means, since $\nabla$ is functional in the specification, it maps to a function in the model. Moreover, the \emph{categories} of models (where morphisms between models are given by monoidal natural transformations) coincide: both are the category of commutative monoids and homomorphisms. Thus the mismatch of~\eqref{eq:problematicmodel} is avoided. 

\medskip
Yet lax product theories are not quite expressive enough for our purposes. We have seen that, using the lax product structure, we can express equationally when a relation is a function. But we cannot, for instance, say when a relation is a ``co-function'', that is, the opposite relation of a function. This capability is very useful in examples, for example for the calculus of fractions in the SMT of Interacting Hopf Algebras~\cite{interactinghopf}.

\subsection{Frobenius theories}
A Frobenius theory---a concept introduced in this paper---is a lax product theory that includes additionally a ``black'' commutative monoid right adjoint to the comonoid.
\[
\includegraphics[height=1.2cm]{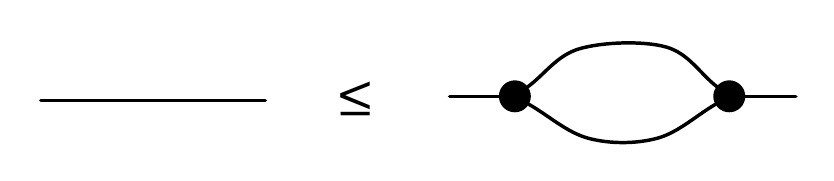}
\quad
\includegraphics[height=1.2cm]{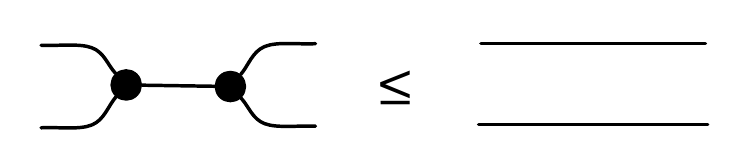}
\]
\[
\includegraphics[height=1cm]{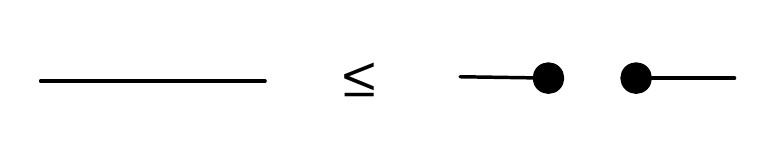}
\quad
\includegraphics[height=1cm]{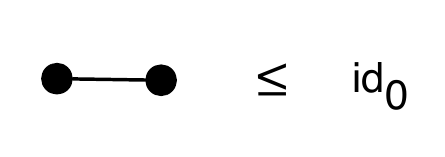}
\]
Moreover, the monoid-comonoid pair satisfies the Frobenius equations and the special law -- i.e. an additional
inequation relating the multiplication and comultiplication.
\[
\includegraphics[height=1.5cm]{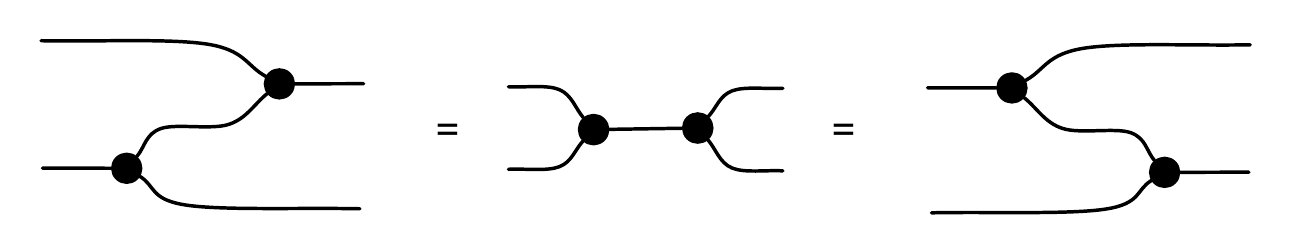}
\]
\[
\includegraphics[height=1.3cm]{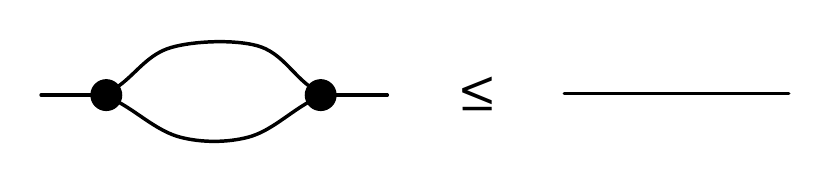}
\]
Notice that the presence of the Frobenius equations induces a self-dual compact closed structure. Together,
the Frobenius equations and the special law give us the spider theorem, which implies that the mirror images of~\eqref{eq:cartesianbicats} hold, i.e.
\[
\lower20pt\hbox{$\includegraphics[height=1.7cm]{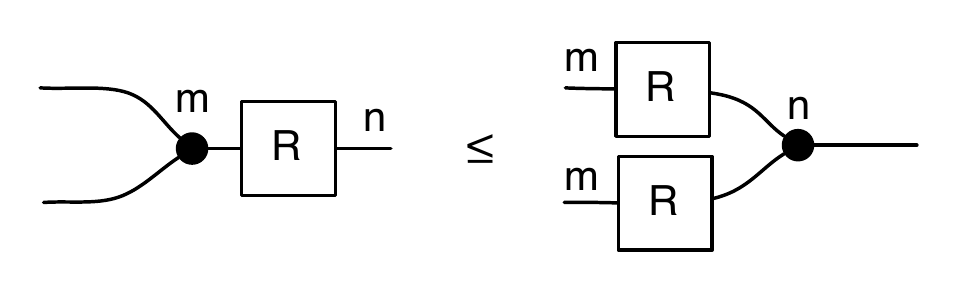}$}
\quad
\lower13pt\hbox{$\includegraphics[height=1.2cm]{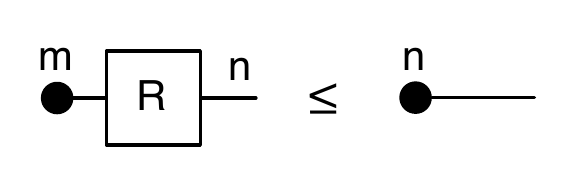}$}
\]

The resulting free ordered prop is what we refer to as a Frobenius prop (frop) or a Carboni-Walters category, after A Carboni and RFC Walters. Indeed, it is an example of Carboni and Walters' bicategory of relations~\cite[Sec.~2]{Carboni1987}: a cartesian bicategory where the ``black'' structure satisfies the Frobenius equations. Moreover, discarding the 2-structure, one obtains a hypergraph category (a.k.a.\ well-supported compact closed category). 
\begin{center}
\begin{tabular}{ | c | c | }
\hline
\textbf{Specification} & Frobenius theory \\
\textbf{Syntax} & string diagrams \\
\textbf{Category} & Carboni-Walters category \\
\textbf{Models}  & lax product preserving functors \\
\textbf{Homomorphisms} & monoidal lax natural transformations \\
\hline
\end{tabular}
\end{center}

The frop of commutative monoids can be thought of as the prop of bialgebras, together with an additional commutative ``black'' monoid. Again, as in the case of lax product theories, the models in $\Rel_\times$ are ordinary commutative monoids, and the model transformations are monoid homomorphisms.
The example of commutative monoids generalises to arbitrary algebraic theories: there is a procedure, analogous to that of producing an SMT from a classical algebraic theory, that results in a Frobenius theory, so that the models of the relevant Lawvere category in $\Set$ are in bijective correspondence with the models of the Frobenius theory in $\Rel_\times$. More than that, the categories of models are equivalent.

But Frobenius theories give us much more that a way of doing ``resource-honest'' algebraic theories in $\Rel_\times$: they are much more expressive and allow us to bring many new examples into the fold. This report introduces the basic theory together with a wide range of examples.

\subsection*{Structure of the paper.}
We recall the concepts of symmetric monoidal theory and props in Section~\ref{sec:smt} and explain how classical algebraic theories can be considered as symmetric monoidal theories, and the corresponding Lawvere theories as certain props. In Section~\ref{sec:laxproducttheories} we extend the picture to inequational theories, resulting in poset-enriched props. We also identify the crucial concept of lax product structure, which allows us to keep the ``arity-preservation'' property of classical models. In Section~\ref{sec:frobenius} we introduce the central concept of Frobenius theory, describe models and focus on some general properties. In Section~\ref{sec:ex}, we highlight interesting examples of Frobenius theories, showcasing the expressivity of the framework. In Section~\ref{sec:frobcartesian} we explain how cartesian theories can be considered as Frobenius theories, without altering the category of models. In the last three sections are an in-depth look at three ubiquitous mathematical theories, considered as Frobenius theories: commutative monoids (Section~\ref{sec:monoid}), abelian groups (Section~\ref{sec:ag}) and modules (Section~\ref{sec:modules}).

\section{Symmetric Monoidal Theories and Props}\label{sec:smt}

Our exposition is founded on \emph{symmetric monoidal theories}: presentations of algebraic structures borne by objects in a symmetric monoidal category.

\begin{definition} A (presentation of a) \emph{symmetric monoidal theory} (SMT) is a pair $\Theory{T} = (\Sigma, E)$ consisting of a \emph{signature} $\Sigma$ and a set of \emph{equations} $E$. The signature $\Sigma$ is a set of \emph{generators} $o \: n\to m$ with \emph{arity} $n$ and \emph{coarity} $m$. The set of \emph{$\Sigma$-terms} is obtained by composing generators in $\Sigma$, the unit $\id \: 1\to 1$ and the symmetry $\sigma_{1,1} \: 2\to 2$ with $;$ and $\tns$. This is a purely formal process: given $\Sigma$-terms $t \: k\to l$, $u \: l\to m$, $v \: m\to n$, one constructs new $\Sigma$-terms $t \mathrel{;} u \: k\to m$ and $t \tns v \: k+n \to l+n$. The set $E$ of \emph{equations} contains pairs $(t,t' \: n\to m)$ of $\Sigma$-terms with the same arity and coarity.
\end{definition}

The categorical concept associated with symmetric monoidal theories is the notion of prop (\textbf{pro}duct and \textbf{p}ermutation category~\cite{MacLane1965}).

\begin{definition} A \emph{prop} is a symmetric strict monoidal category with objects the natural numbers, where $\tns$ on objects is addition. The prop  \emph{freely generated} by a theory $\Theory{T}=(\Sigma,E)$, denoted by $\freePROP{\Theory{T}}$, has as its set of arrows $n\to m$ the set of $\Sigma$-terms $n\to m$ taken modulo the laws of symmetric strict monoidal categories --- Fig.~\ref{fig:axSMC} --- and the smallest congruence (with respect to $\poi$ and $\tns$) containing  equations $t=t'$ for any $(t,t')\in E$.
\end{definition}

\begin{figure}[t]
$$
\begin{array}{c}
(t_1 \poi t_3) \tns (t_2 \poi t_4) = (t_1 \tns t_2) \poi (t_3 \tns t_4)\end{array} $$ $$
\begin{array}{rcl}
(t_1 \poi t_2) \poi t_3= t_1 \poi (t_2 \poi t_3) & &
id_n \poi c = c = c\poi id_m\\
(t_1 \tns t_2) \tns t_3 = t_1 \tns (t_2\tns t_3) & &
id_0 \tns t = t  = t \tns id_0\\
\sigma_{1,1}\poi\sigma_{1,1}=id_2 & &
(t \tns id_z) \poi \sigma_{m,z} = \sigma_{n,z} \poi (id_z \tns t)
\end{array} $$
\caption{Axioms of symmetric strict monoidal categories for a prop $\T$.}\label{fig:axSMC}
\end{figure}

 There is a natural graphical representation for arrows of a prop as string diagrams, which we now sketch, referring to~\cite{Selinger2009} for the details. A $\Sigma$-term $n \to m$ is pictured as a box with $n$ ports on the left and $m$ ports on the right. Composition via $\poi$ and $\tns$ are rendered graphically by horizontal and vertical juxtaposition of boxes, respectively.
    \begin{eqnarray}\label{eq:graphlanguage}
t \poi s \text{ is drawn }
\lower7pt\hbox{$\includegraphics[height=.6cm]{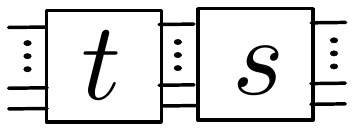}$}
\quad
t \tns s \text{ is drawn }
\lower13pt\hbox{$\includegraphics[height=1.1cm]{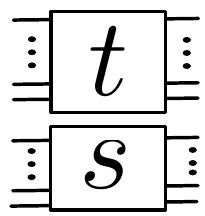}$}.
\end{eqnarray}
    In any SMT there are specific $\Sigma$-terms generating the underlying symmetric monoidal structure: these are $\id_1 \: 1 \to 1$, represented as $\Idnet$, the symmetry $\sigma_{1,1} \: 1+1 \to 1+1$, represented as $\symNet$, and the unit object for $\tns$, that is, $\id_0 \: 0 \to 0$, whose representation is an empty diagram. Graphical representation for arbitrary identities $\id_n$ and symmetries $\sigma_{n,m}$ are generated according to the pasting rules in~\eqref{eq:graphlanguage}.

  \begin{example}~\label{ex:equationalprops}
  \begin{enumerate}[(a)]
  \item We write $\CMtheory = (\Sigma_M,E_M)$ for the SMT of \emph{commutative monoids}. The signature $\Sigma_M$ contains a \emph{multiplication} $\Wmult \: 2 \to 1$ and a \emph{unit} $\Wunit \: 0 \to 1$.
 Equations $E_M$ assert associativity \eqref{eq:wmonassoc}, commutativity~\eqref{eq:wmoncomm} and unitality~\eqref{eq:wmonunitlaw}.
 \begin{multicols}{3}\noindent
 \begin{equation}
\label{eq:wmonassoc}
\lower10pt\hbox{$\includegraphics[height=1cm]{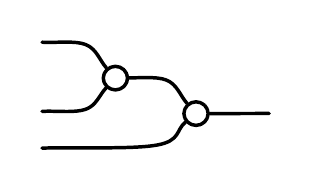}$}
\!\!\!
=
\!\!\!
\lower10pt\hbox{$\includegraphics[height=1cm]{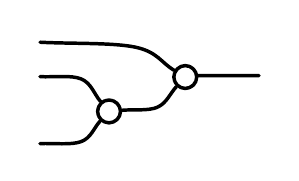}$}
\end{equation}
\begin{equation}
\label{eq:wmoncomm}
\lower7pt\hbox{$\includegraphics[height=.7cm]{graffles/Wmult.pdf}$}
\!
=
\!\!
\lower7pt\hbox{$\includegraphics[height=.7cm]{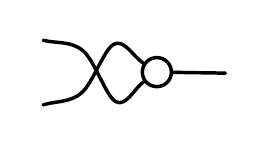}$}
\end{equation}
\begin{equation}
\label{eq:wmonunitlaw}
\lower8pt\hbox{$\includegraphics[height=.8cm]{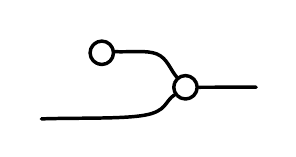}$}
\!
=\!
\lower5pt\hbox{$\includegraphics[height=.5cm]{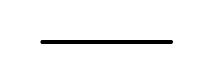}$}
\end{equation}
\end{multicols}
\item\label{it:comonoids} Next, the SMT $\CCtheory = (\Sigma_C, E_C)$ of \emph{commutative comonoids}. The signature $\Sigma_C$ consists of a \emph{comultiplication} $\Bcomult \: 1 \to 2$ and a \emph{counit} $\Bcounit \: 1\to 0$. $E_C$ consists of the following equations.
\begin{multicols}{3}\noindent
\begin{equation}
\label{eq:bcomonassoc}
\lower10pt\hbox{$\includegraphics[height=1cm]{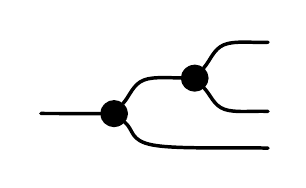}$}
\!\!
=
\!\!
\lower10pt\hbox{$\includegraphics[height=1cm]{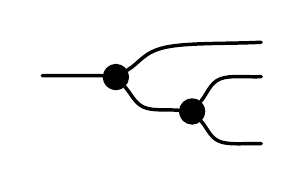}$}
\end{equation}
\begin{equation}
\label{eq:bcomoncomm}
\lower7pt\hbox{$\includegraphics[height=.9cm]{graffles/Bcomult.pdf}$}
\!
=
\!\!\!
\lower9pt\hbox{$\includegraphics[height=.9cm]{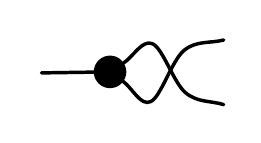}$}
\end{equation}
\begin{equation}
\label{eq:bcomonunitlaw}
\lower8pt\hbox{$\includegraphics[height=.8cm]{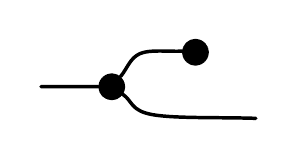}$}
\!\!\!
=
\!
\lower5pt\hbox{$\includegraphics[height=.5cm]{graffles/idcircuit.pdf}$}
\end{equation}
\end{multicols}
\item \label{it:frobenius} Monoids and comonoids can be combined into a theory that plays an important role in our exposition: the theory of \emph{special Frobenius algebras}~\cite{Carboni1987}. This is given by $\Frobtheory = (\Sigma_M \uplus \Sigma_C, E_M \uplus E_C \uplus F)$, where $F$ is the following set of equations.
    \begin{multicols}{2}\noindent
    \begin{equation}\label{eq:BWFrob}
\lower8pt\hbox{$\includegraphics[height=.8cm]{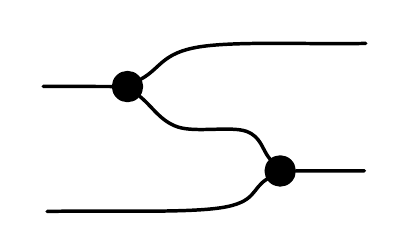}$}
=
\lower6pt\hbox{$\includegraphics[height=.6cm]{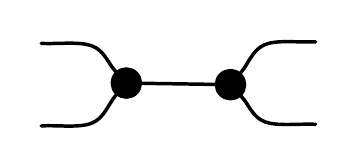}$}
=
\lower8pt\hbox{$\includegraphics[height=.8cm]{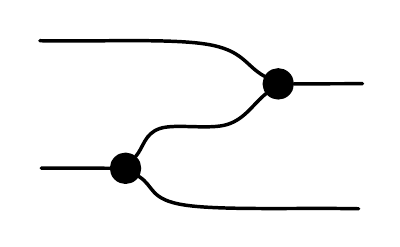}$}
\end{equation}
\begin{equation}\label{eq:BWSep}
\lower8pt\hbox{$\includegraphics[height=.8cm]{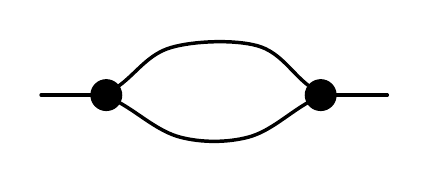}$}
=
\lower5pt\hbox{$\includegraphics[height=.5cm]{graffles/idcircuit.pdf}$}
\end{equation}
\end{multicols}

\item Another fundamental way to combine monoids and comonoids is the theory of (commutative/cocommutative) \emph{bialgebras} $\Bialgtheory = (\Sigma_M \uplus \Sigma_C, E_M \uplus E_C \uplus B)$, where $B$ is the following set of equations.
    \begin{multicols}{2}
\noindent
\begin{equation}
\label{eq:bialgunitsl}
\lower2pt\hbox{$
\lower7pt\hbox{$\includegraphics[height=.8cm]{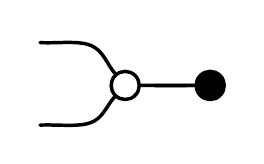}$}
=
\lower7pt\hbox{$\includegraphics[height=.8cm]{graffles/lunitsr.pdf}$}
$}
\end{equation}
\begin{equation}
\label{eq:bialgunitsr}
\lower7pt\hbox{$\includegraphics[height=.8cm]{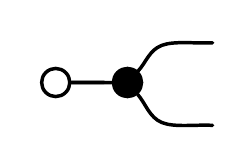}$}
=
\lower7pt\hbox{$\includegraphics[height=.8cm]{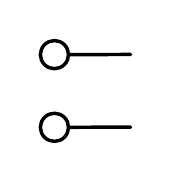}$}
\end{equation}
\begin{equation}
\label{eq:bialg}
\lower8pt\hbox{$\includegraphics[height=.8cm]{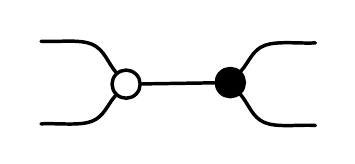}$}
=
\lower13pt\hbox{$\includegraphics[height=1.2cm]{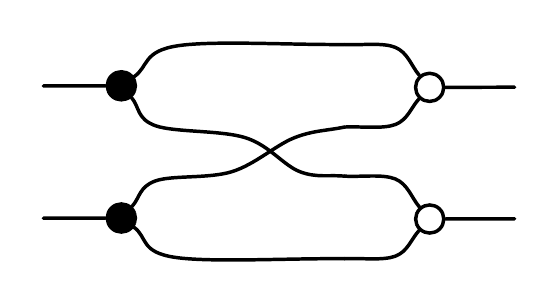}$}
\end{equation}
\begin{equation}
\label{eq:bwbone}
\lower4pt\hbox{$\includegraphics[height=.5cm]{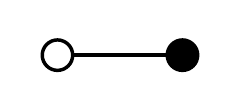}$}
=
\id_0
\end{equation}
\end{multicols}
One can read \eqref{eq:bialgunitsl}-\eqref{eq:bwbone} as stating that the monoid structure (multiplication,unit) is a comonoid homomorphism, and vice versa, the comonoid structure is a monoid homomorphism. 

Bialgebras and special Frobenius algebras play an important role in recent research threads in quantum~\cite{CoeckeDuncanZX2011,BialgAreFrob14}, concurrency~\cite{Bruni2006,Sobocinski2013a} and control theory~\cite{Bonchi2014b,BaezErbele-CategoriesInControl,Bonchi2015}.

\item Another theory that play a crucial role in the aforementioned works is the theory $\Hopftheory$ of \emph{Hopf algebras}. This is obtained from the theory of bialgebra by extending the $\Sigma_M \uplus \Sigma_C$, with the \emph{antipode} $\cgr[height=12pt]{antipode.pdf} \: 1 \to 1$ and the set of equations $E_M \uplus E_C \uplus B $ with the following three.
\begin{multicols}{3}
\noindent
\begin{equation} \label{eq:antipodehom1}
%\lower2pt\hbox{$
\lower7pt\hbox{$\includegraphics[height=.7cm]{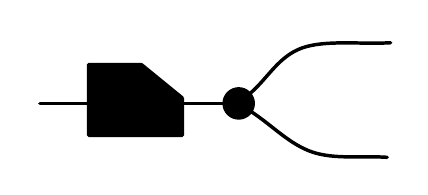}$}
\!=\!
\lower10pt\hbox{$\includegraphics[height=1cm]{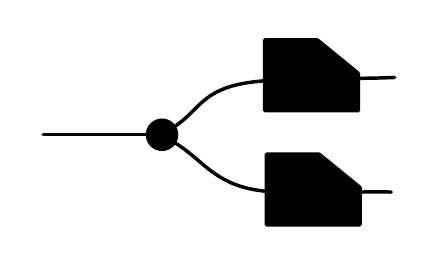}$}
%$}
\end{equation}
\begin{equation}
\label{eq:antipodehom2}
\lower7pt\hbox{$\includegraphics[height=.7cm]{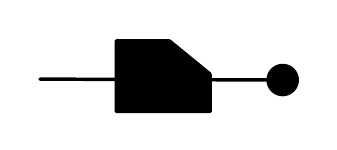}$}
=
\lower6pt\hbox{$\includegraphics[height=.6cm]{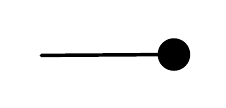}$}
\end{equation}
  \begin{equation} \label{eq:hopf}
  \!\lower8pt\hbox{$\includegraphics[height=.8cm]{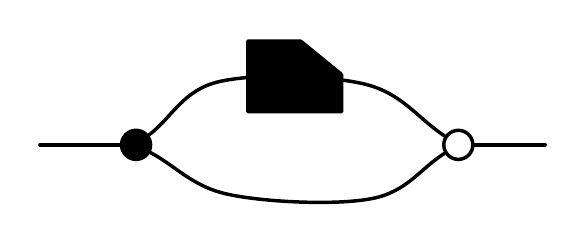}$}
  \!\!=\!\! 
  \lower3pt\hbox{$\includegraphics[height=.4cm]{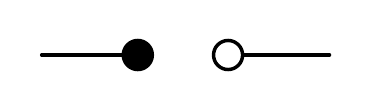}$}
   \end{equation}
   \end{multicols}
\end{enumerate}
  \end{example}

The assertion that $\CMtheory$ \emph{is the SMT of commutative monoids}---and similarly for other SMTs in our exposition---can be made precise through the notion of \emph{model} of an SMT.
\begin{definition}
Given a symmetric monoidal category $\catC$, 
a model of an SMT $\Theory{T}$ in $\catC$ is a symmetric monoidal functor $\funF \: \freePROP{\Theory{T}} \to \catC$. Then $\model{\T,\catC}$ is the category of models of $\T$ in $\catC$ and monoidal natural transformations between them. 
\end{definition}

Turning to commutative monoids, there is a category $\mathsf{Monoid}(\catC)$ whose objects are the commutative monoids in $\catC$, i.e., objects $x \in \catC$ equipped with arrows $x \tns x \to x$ and $I \to x$, satisfying the usual equations.
Given any model $\funF \: \Mon \to \catC$, it follows that $\funF 1$ is a commutative monoid in $\catC$: this yields a functor $\model{\Mon, \catC} \to \mathsf{Monoid}(\catC)$. Saying that $(\Sigma_M, E_M)$ is the SMT of commutative monoids means that this functor is an equivalence natural in $\catC$. 

We can  recover classical models by considering symmetric monoidal functors to $\Set_\times$, the symmetric monoidal category of sets, where the monoidal product is the cartesian product $\times$.
Indeed, the functor is determined, up-to natural isomorphism, by where it sends $1$. 
Concretely, we can consider the image of a symmetric monoidal functor of this type to consist of the
sets of $n$-tuples $n \mapsto X^n \Defeq \{\,(x_1,\dots,x_n) \,|\, x_i\in X\,\}$. Then $\mathsf{Monoid}(\Set_\times)$ is equivalent to the category of ordinary commutative monoids and monoid homomorphisms.

\subsection{Cartesian theories and Lawvere Categories}\label{sec:cartesian}

A \emph{cartesian category} (or finite product category) is a symmetric monoidal category where the monoidal product $\tns$ satisfies the universal property of the categorical product; a cartesian functor is a product preserving functor. It is well-known  that a symmetric monoidal category $\catC$ is cartesian \emph{iff} for every object $n$ in $\catC$, there are arrows $\dup_n \colon n \to n \tns n$ and $\dis_n \colon n\to I$ forming a cocommutative comonoid, graphically denoted by $\Bcomultn$ and $\Bcounitn$, and every arrow $f\colon m \to n$ in $\catC$ is a comonoid homomorphism. 
\begin{multicols}{2}
\noindent
\begin{equation}
\lower2pt\hbox{$
\lower10pt\hbox{$\includegraphics[height=1cm]{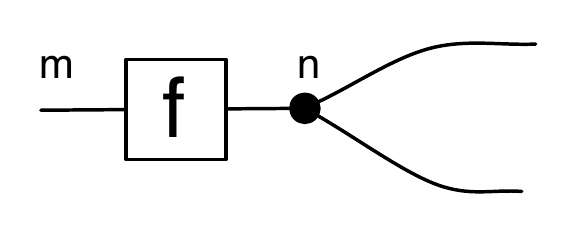}$}
=
\lower12pt\hbox{$\includegraphics[height=1.3cm]{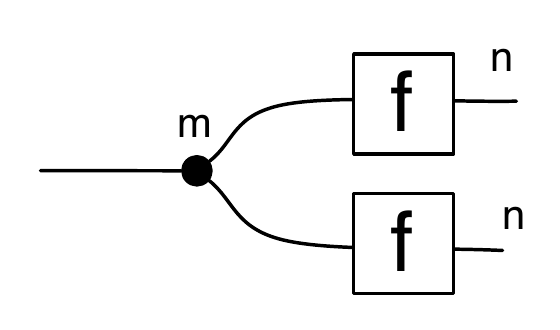}$}
$}
\end{equation}
\begin{equation}
\lower8pt\hbox{$\includegraphics[height=1cm]{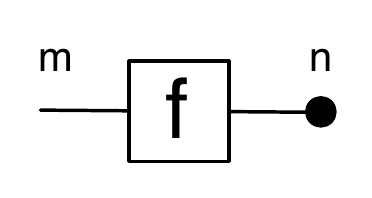}$}
=
\lower5pt\hbox{$\includegraphics[height=1cm]{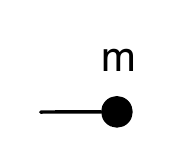}$}
\end{equation}
\end{multicols}

A \emph{Lawvere category}~\cite{LawvereOriginalPaper,hyland2007category} is then a symmetric monoidal category that is both cartesian and a prop.

\begin{example}~\label{ex:lawvere}
\begin{enumerate}[(a)]
\item
Recall the theory of commutative comonoids $\CCtheory = (\Sigma_C , E_C)$ from Example \ref{ex:equationalprops}(b). The resulting prop $\Com$ is the initial Lawvere category, the free category with products on one object. The 
comultiplication $\Bcomult \: 1 \to 2$ and the counit $\Bcounit \: 1\to 0$ are the comonoid on $1$.
For $n\in\N$, $\dup_n \colon n \to n \tns n$ and $\dis_n\colon n\to 0$ are defined recursively: $\dup_0 = id_0$ and $\dup_{n+1} = (\dup_1 \tns \dup_n) \poi (\id_1 \tns \sigma_{1,n} \tns id_n)$,  $\dis_0 = id_0$ and $\dis_{n+1} = \dis_1\tns \dis_n$.
\item The prop of \emph{bialgebras} $\B$   (Example \ref{ex:equationalprops}(c)) is also a Lawvere category. For every natural number, the comonoid structure is defined as above. Moreover all arrows in $\B$ are comonoid homomorphisms, since \eqref{eq:bialgunitsl}, \eqref{eq:bialgunitsr}, \eqref{eq:bialg}, \eqref{eq:bwbone} say exactly that $\Wmult$ and $\Wunit$ are comonoid homomorphisms.
\item Amongst the other SMTs in Example \ref{ex:equationalprops}, only $\Hopftheory$ freely generates a Lawvere category: indeed equations \eqref{eq:antipodehom1} and \eqref{eq:antipodehom2} state that the antipode is a comonoid homomorphism.
\end{enumerate}
\end{example}

\begin{definition}\label{def:Lawvere} 
A (presentation of a) \emph{cartesian theory} is a pair $\Theory{T}=(\Sigma, E)$ consisting of a signature $\Sigma$ and equations $E$.
$\Sigma$ is a set of \emph{generators} $o \: n\to 1$ with \emph{arity} $n$ and \emph{coarity} $1$.  The set of equations $E$ contains pairs $(t,t' \: n\to 1)$ of \emph{Cartesian $\Sigma$-terms}, namely arrows of the prop freely generated by the SMT $(\Sigma \uplus \Sigma_C, E_C \uplus E_{CH})$ where $E_{CH}$  contains equations 
%\begin{multicols}{2}
%\noindent
\begin{equation}
%\lower5pt\hbox{$
\lower10pt\hbox{$\includegraphics[height=1cm]{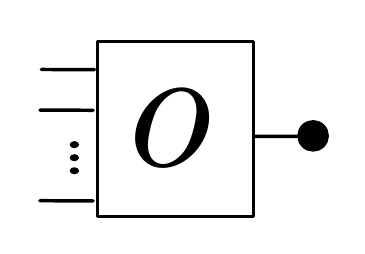}$}
=
\lower10pt\hbox{$\includegraphics[height=1cm]{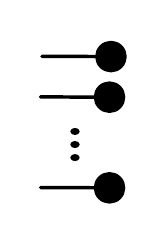}$}
%$}
\end{equation}
\begin{equation}
\lower10pt\hbox{$\includegraphics[height=1cm]{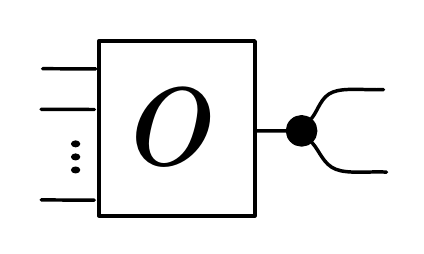}$}
=
\lower15pt\hbox{$\includegraphics[height=1.5cm]{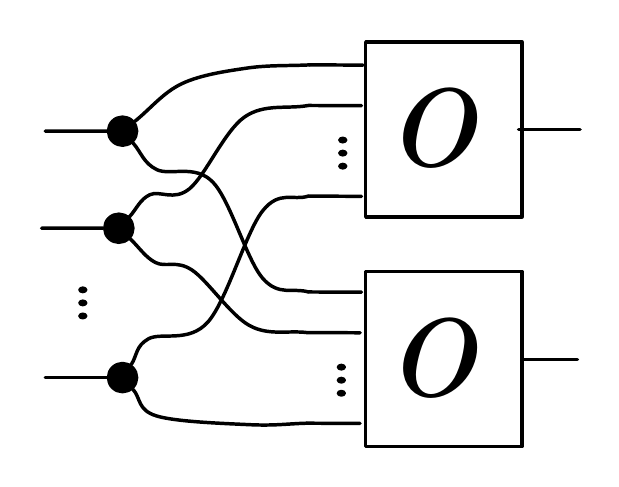}$}
\end{equation}
%\end{multicols}
for each generator $\cgr[height=18pt]{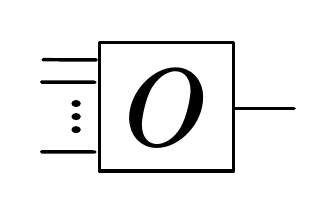} \: n \to 1$ of the signature $\Sigma$.

The Lawvere category \emph{freely generated} by a cartesian theory $\Theory{T}=(\Sigma,E)$, denoted by $\Law{\Theory{T}}$,  is the prop freely generated by the SMT $(\Sigma \uplus \Sigma_C, E\uplus E_C \uplus E_{CH})$. The latter will be often referred to as the SMT corresponding to the cartesian theory $(\Sigma,E)$.
\end{definition}

Cartesian terms can be thought as the familiar notion of standard syntactic term: trees with leaves labeled by variables. The ability to copy and discard variables is given by  $\Bcomult$ and $\Bcounit$, respectively. Since these structures are implicit in any cartesian theory, one can therefore think of cartesian terms as \emph{resource-insensitive} syntax. On the other hand, the string diagrams of SMTs provide a resource-aware syntax since the ability to add and copy variables, if available, is made explicit. 

\begin{example}\label{ex:cartesiantheory}~
\begin{enumerate}[(a)]
\item In $\Sigma$-terms of the SMT of commutative monoids $\CMtheory=(\Sigma_M,E_M)$ (Example \ref{ex:equationalprops}(a)), variables cannot be copied or discarded. The cartesian theory of commutative monoids has the same signature and equations, but terms have the implicit capability of being copied and discharged. Indeed, the Lawvere category $\Law{\CMtheory}$ is isomorphic to the prop $\freePROP{\Bialgtheory}$  generated by the SMT of bialgebras (Example \ref{ex:equationalprops}(d)). 
\item By adding to  $\CMtheory$ the antipode  $\cgr[height=13pt]{antipode.pdf} \: 1 \to 1$ and equation \eqref{eq:hopf}, one obtains the cartesian theory $\AGtheory$ of Abelian groups. The corresponding SMT is the theory $\Hopftheory$ of Hopf algebras (Example \ref{ex:equationalprops}(e)): i.e. $\freePROP{\Bialgtheory} \cong \Law{\AGtheory}$.
\end{enumerate}
\end{example}

As for SMTs, the assertion that $\CMtheory$ \emph{is the cartesian theory of commutative monoids} can be made precise using the notion of \emph{model} of an cartesian theory.
\begin{definition}
Given a cartesian category $\catC$, 
a model of a cartesian theory $\Theory{T}$ in $\catC$ is a cartesian functor $\funF \: \Law{\Theory{T}} \to \catC$. Then $\carmodel{\T,\catC}$ is the category of models of $\T$ in $\catC$ and monoidal natural transformations between them. 
\end{definition}

For an example take the cartesian category $\Set$. Every model $\funF \: \Law{\Theory{T}} \to \Set$ maps $1$ to some set $X$ and thus every natural number $n$ to  $X^n$. Requiring $F$ to be cartesian forces the counit $\Bcounit \colon 1\to 0$ to be mapped into the unique morphism from $X$ to the final object $1=X^0$ and the comultiplication $\Bcomult\colon 1\to 2$ to the diagonal $\dup_X = \langle id_X,id_X\rangle \colon X\to X\times X$. So, a model $F$ is uniquely determined by the set $F1$ and the functions $Fo \colon (F1)^n \to F1$ for each generator $o\colon n \to 1$ of the signature. In a nutshell, the notion of Cartesian model for $\Theory{T}$ coincides with the standard notion of algebra. By spelling out the definition of natural transformation, one can readily check that morphism of models are homomorphisms.

\section{Lax product theories}\label{sec:laxproducttheories}
A first step toward Frobenius theories and their models consists in relaxing products into lax products. In this section, we introduce the categorical machinery to deals with theories of inequations and lax products theories.

\medskip

Suppose that $\Sigma$ is a set of generators and $I$ is a set of \emph{inequations}: similarly to an equation, the underlying data of an inequation is simply a pair $(t_1,t_2)$ of equal-typed $\Sigma$-terms. Unlike equations, however, we will understand this data as being directed:
\[
t_1 \leq t_2
\]
We call the pair $(\Sigma,I)$ a (presentation of a) symmetric monoidal inequation theory (SMIT).

Throughout the paper we use \emph{ordered} as a synonym for ``enriched in $\mathbf{Pos}$'' - the category of posets and monotonic functions. Indeed, just as SMTs lead to props, SMITs lead to ordered props, as defined below.
\begin{definition}[Ordered prop]
An \emph{ordered prop} is a prop enriched over the category of posets: that is, it is a strict symmetric 2-category $\catC$ with objects the natural numbers, monoidal product on objects defined as $m\oplus n \Defeq m+n$, where each set of arrows $\catC[m,n]$ is a poset, with composition and monoidal product monotonic.
Similarly, a \emph{pre-ordered prop} is a prop enriched over the category of pre-orders.
\end{definition}

Analogously to how one---given and SMT $(\Sigma,E)$---constructs a free prop, we can use a SMIT $(\Sigma, I)$ to generate a free ordered prop. First, we constructs the free pre-ordered prop: arrows are $\Sigma$-terms. The homset orders are determined by whiskering $I$ and closing it under $\oplus$, then applying reflexive and transitive closure: this is the smallest preorder containing $I$ that makes $\catC$ into a pre-ordered prop (i.e. composition is monotonic and $\oplus$ is a 2-functor). Then, we obtain the free ordered prop by quotienting the free pre-ordered prop by the equivalence induced by the pre-order.

\medskip

Any SMT $(\Sigma,E)$ gives rise to a canonical SMIT $(\Sigma,I)$ where each equation is replaced with two inequalities $I=E\uplus E^{op}$, in the obvious way. The free prop for $(\Sigma,E)$ can then be obtained from the free ordered prop for $(\Sigma,I)$ by forgetting the underlying 2-structure. For this reason, we can safely abuse the notation $\freeordPROP{\Theory{T}}$ to denote the ordered prop freely generated by an SMIT $\Theory{T}$.

\begin{example}~\label{ex:SMIT}
\begin{enumerate}[(a)]
\item The SMT of commutative monoids $\CMtheory = (\Sigma_{M}, E_M)$ (Example \ref{ex:equationalprops} (a)) can be regarded as the SMIT $(\Sigma_{M}, E_M \uplus E_M^{op})$.
\item The SMT of cocommutative comonoids $\CCtheory = (\Sigma_{C}, E_C)$ (Example \ref{ex:equationalprops} (b)) can be regarded as the SMIT $(\Sigma_{C}, E_C \uplus E_C^{op})$.
\item The SMT of bialgebra $\Bialgtheory = (\Sigma_M \uplus \Sigma_C, E_M \uplus E_C \uplus B)$ (Example \ref{ex:equationalprops} (d)) can be regarded as the SMIT $(\Sigma_M \uplus \Sigma_C, E_M \uplus E_M^{op} \uplus E_C \uplus E_C^{op} \uplus B \uplus B^{op}) $.
\item From the SMIT of bialgebra, one can drop the inequations $B^{op}$ and obtain the SMIT of \emph{lax bialgebras} $\Laxbialgtheory = (\Sigma_M \uplus \Sigma_C, E_M \uplus E_M^{op} \uplus E_C \uplus E_C^{op} \uplus B)$. In this theory we have a monoid, a comonoid and the inequations of $B$ -- that we depict  below for the convenience of the reader -- force the monoid to be a lax comonoid homomorphism.
    \begin{multicols}{2}
\noindent
\begin{equation}
\label{eq:lunitsl}
\lower2pt\hbox{$
\lower7pt\hbox{$\includegraphics[height=.8cm]{graffles/lunitsl.pdf}$}
\leq
\lower7pt\hbox{$\includegraphics[height=.8cm]{graffles/lunitsr.pdf}$}
$}
\end{equation}
\begin{equation}
\label{eq:lunitsr}
\lower7pt\hbox{$\includegraphics[height=.8cm]{graffles/runitsl.pdf}$}
\leq
\lower7pt\hbox{$\includegraphics[height=.8cm]{graffles/runitsr.pdf}$}
\end{equation}
\begin{equation}
\label{eq:lbialg}
\lower7pt\hbox{$\includegraphics[height=.8cm]{graffles/bialgl.pdf}$}
\leq
\lower11pt\hbox{$\includegraphics[height=1.1cm]{graffles/bialgr.pdf}$}
\end{equation}
\begin{equation}
\label{eq:lbwbone}
\lower4pt\hbox{$\includegraphics[height=.5cm]{graffles/unitsl.pdf}$}
\leq
id_0
\end{equation}
\end{multicols}
\item Otherwise one can drop the inequations in $B$ and obtains the SMIT of \emph{oplax bialgebras} $\Oplaxbialgtheory = (\Sigma_M \uplus \Sigma_C, E_M \uplus E_M^{op} \uplus E_C \uplus E_C^{op} \uplus B^{op})$. The inequations  of $B^{op}$ are depicted below.
    \begin{multicols}{2}
\noindent
\begin{equation}
\label{eq:olunitsl}
\lower2pt\hbox{$
\lower7pt\hbox{$\includegraphics[height=.8cm]{graffles/lunitsl.pdf}$}
\geq
\lower7pt\hbox{$\includegraphics[height=.8cm]{graffles/lunitsr.pdf}$}$}
\end{equation}
\begin{equation}
\label{eq:olunitsr}
\lower7pt\hbox{$\includegraphics[height=.8cm]{graffles/runitsl.pdf}$}
\geq
\lower7pt\hbox{$\includegraphics[height=.8cm]{graffles/runitsr.pdf}$}
\end{equation}
\begin{equation}
\label{eq:olbialg}
\lower7pt\hbox{$\includegraphics[height=.8cm]{graffles/bialgl.pdf}$}
\geq
\lower11pt\hbox{$\includegraphics[height=1.1cm]{graffles/bialgr.pdf}$}
\end{equation}
\begin{equation}
\label{eq:olbwbone}
\lower4pt\hbox{$\includegraphics[height=.5cm]{graffles/unitsl.pdf}$}
\geq
id_0
\end{equation}
\end{multicols}
\end{enumerate}
\end{example}

Particularly relevant for our exposition is the SMIT of commutative comonoids: cartesian theories include an implicit comonoid structure and force the generators in the signature to be comonoid homomorphisms. The theories that we are going to introduce next -- lax product theories -- are analogous, but they require the generators to be \emph{lax} comonoid homomorphisms. 

%The SMIT of commutative comonoids, obtained from the SMT of commutative comonoids, is of central interest. Recall that the generators are 
%$\Bcomult$, $\Bcounit$ and equations (pairs of inequations)
%\[
%\includegraphics[height=3cm]{graffles2/Bcomonoid}
%\]

\begin{definition}[Lax Product Theory]\label{def:LPT} 
A (presentation of a) \emph{lax product theory} (LPT)  is a pair $\Theory{T}=(\Sigma, I)$ consisting of a signature $\Sigma$ and a set of \emph{inequations} $I$.
The signature $\Sigma$ is a set of \emph{generators} $o \: n\to m$ with \emph{arity} $n$ and \emph{coarity} $m$.  The set of inequations $I$ contains pairs $(t,t' \: n\to m)$ of \emph{L-$\Sigma$-terms}, namely  arrows of the ordered prop freely generated by the SMIT $(\Sigma \uplus \Sigma_C, E_C \uplus E_C^{op} \uplus I_{LCH})$ where $I_{LCH}$ is the set containing 
%\begin{multicols}{2}
\begin{equation}\label{eq:lcomhom1}
\lower15pt\hbox{$\includegraphics[height=1.5cm]{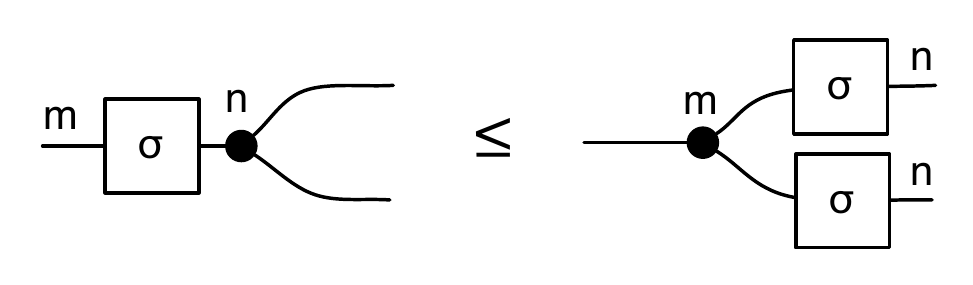}$}
\end{equation}
\begin{equation}\label{eq:lcomhom2}
\lower7pt\hbox{$\includegraphics[height=1cm]{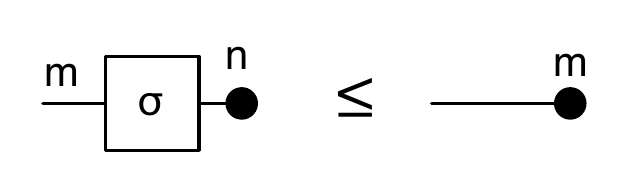}$}
\end{equation}
%\end{multicols}
for each generator $\sigma$ in $\Sigma$.

We refer to $(\Sigma \uplus \Sigma_C, I \uplus E_C \uplus E_C^{op} \uplus I_{LCH})$ as the SMIT corresponding to a LPT $\Theory{T} = (\Sigma, I)$. The ordered prop freely generated by the SMIT corresponding to $\Theory{T}$ is called the \emph{lax product prop freely generated} by $\Theory{T}$ and denoted by $\freeLPPROP{\Theory{T}}$.
\end{definition}

The mismatch between and SMITs and LPTs is analogous to the one of SMTs and cartesian theories: the theory of comonoids (Example \ref{ex:SMIT} (b)) is the SMIT corresponding to the empty LPT $(\emptyset, \emptyset)$; the theory of lax bialgebra (Example \ref{ex:SMIT} (d)) is the SMIT corresponding to the LPT of commutative monoids (Example \ref{ex:SMIT} (a)).

%\marginpar{PLEASE PAWEL CHECK THIS FACT ABOUT COARITY, SINCE I DO NOT REALLY UNDERSTAND IT.}
An important difference between lax product theories and cartesian theories is that generators in $\Sigma$ can have arbitrary coarity, not necessarily $1$ as is the case in any cartesian theory.
Indeed, the presence of finite products eliminates the need for coarities other than 1, since to give an arrow $X^m \to X^n$, in a cartesian category is to give an $n$-tuple of arrows $X^m\to X$, obtained by composing with the projections. In a lax product theory, instead, this is not the case. As we shall see below, the category of relations can be considered as a source of models for a lax product theory and it is clearly not true, in general, that relations $X^m\to X^n$ are determined by their projections.

 The notion of lax product prop will be formalised in the next subsection. For the moment, the reader can think of these structures as ordered props where objects are equipped with a comonoid structure and arrows are lax comonoid homomorphism. This is the case in $\freeLPPROP{\Theory{T}}$ as shown below.

\begin{theorem}\label{thm:lptlaxhom}
Let $\Theory{T}= (\Sigma,I)$ be an LPT and $\freeLPPROP{\Theory{T}}$ be lax product prop freely generated by it.
Then every $t\colon m\to n$ in $\freeLPPROP{\Theory{T}}$ is a lax comonoid homomorphism.
\end{theorem}
\begin{proof}
By inequations \eqref{eq:lcomhom1} and \eqref{eq:lcomhom2}, every generator in $\Sigma$ is a lax comonoid homomorphism. A simple structural induction confirms it for compound terms. 
\[
\includegraphics[height=1.5cm]{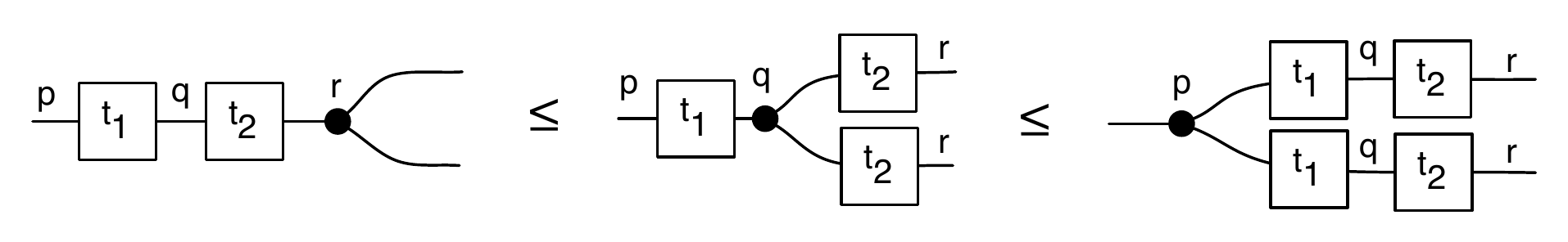}
\]
\[
\includegraphics[height=2.5cm]{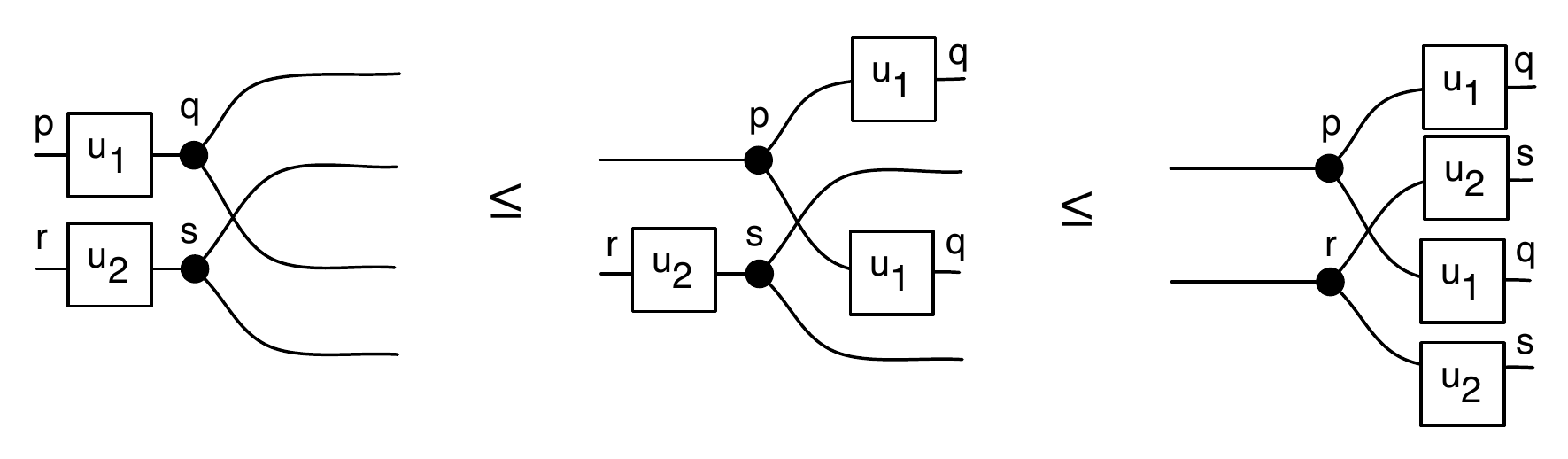}
\]
\end{proof}

\subsection{Lax product structures and lax products} 
%\marginpar{PAWEL PLEASE MAKE SURE THAT THERE IS NO MISMATCH WITH THE BI-CATEGORY AND 2-CATEGORY}
In Section \ref{sec:cartesian} we recalled that the monoidal product is a categorical product precisely when
all arrows are comonoid homomorphisms. Here we show that the property of arrows being lax-comonoid homomorphisms force the monoidal product to be a lax product, a bicategorical limit. We begin by noting that the commutative comonoid structure in any lax product theory is an instance of something we call a \emph{lax product structure}.

\begin{definition}[Lax product structure]\label{defn:laxproductstructure}
Given an ordered monoidal category $\mathbf{C}$, a \emph{lax product structure} is a choice, for each object $C\in\mathbf{C}$, of commutative comonoid $(\Delta_C,\bot_C)$, compatible with the monoidal product in the obvious way, i.e.:
\[
\Delta_{C\oplus D} = (\Delta_C \oplus \Delta_D);(C\oplus \sigma_{C,D} \oplus D)
\qquad
\bot_{C\oplus D} = \bot_C \oplus \bot_D
\]
 such that for every arrow $\alpha:B\to C$ we have 
\[
\alpha;\Delta_C \leq \Delta_B \poi (\alpha\oplus\alpha)\quad\text{and}\quad \alpha;\bot_C \leq \bot_C.
\]
\end{definition}
\begin{lem}
In an ordered monoidal category $\mathbf{C}$, a lax product structure, if it exists, is unique.
\end{lem}
\begin{proof}
Suppose that for some $C\in \mathbf{C}$ we have lax product structures 
$\{\,(\Delta_C, \bot_C)\,|\,C\in\mathbf{C}\,\}$ and $\{\,(\Delta'_C, \bot'_C)\,|\,C\in\mathbf{C}\,\}$
which we shall draw $(\Bcomult,\Bcounit)$ and 
$(\lower7pt\hbox{$\includegraphics[height=.7cm]{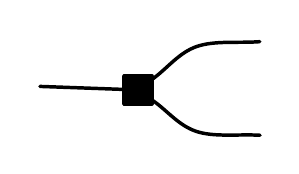}$},
\lower4pt\hbox{$\includegraphics[height=.5cm]{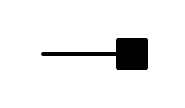}$}
)$, respectively. It follows that, for all $C$, $\bot_C = \bot'_C$ since $\bot'_C\leq \bot_C$:
\[
\includegraphics[height=1.2cm]{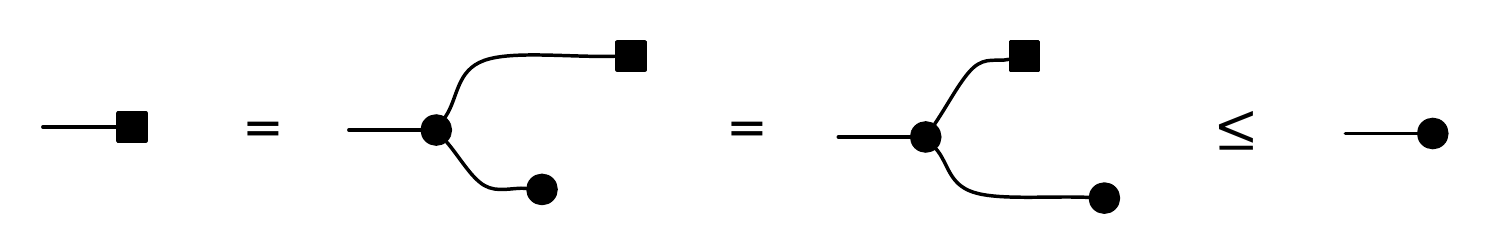}
\]
and, using a symmetric argument, $\bot_C \leq \bot'_C$. Using the fact that the lax product structure is, by definition, assumed to be compatible with monoidal product and the fact that $\bot_C=\bot'_C$ it follows that
$\Delta'_C \leq \Delta_C$:
\[
\includegraphics[height=1.7cm]{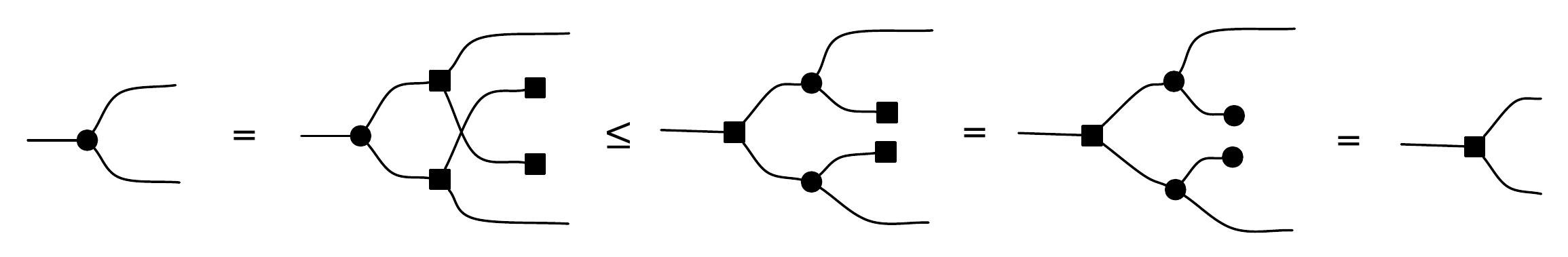}
\]
and, again by a symmetric argument, that $\Delta_C\leq \Delta'_C$.
\end{proof}

\medskip
We will now show that, if an ordered monoidal category has a lax product structure then the monoidal product is a lax product, by which we mean the following bicategorical limit.
Given objects $m$ and $n$, a \emph{lax product} is an object $m\times n$
with \emph{projections}, that is arrows 
$\pi_1: m\times n \to m$, $\pi_2:m\times n \to n$
s.\ t.\ for any $f: k\to m$, $g:k\to n$ there exists $h: k\to m\times n$ 
and 2-cells $\rho_1$, $\rho_2$ as illustrated below:
\begin{equation}\label{eq:laxprod1}
\raise15pt\hbox{$
\xymatrix{
{m} & {m\times n} 
\ar@{=>}[]!<-3ex,-3ex>;[dl]!<5ex,3ex>_{\rho_1}
\ar@{=>}[]!<3ex,-3ex>;[dr]!<-5ex,3ex>^{\rho_2}
\ar[l]_{\pi_1} \ar[r]^{\pi_2} &  {n} \\ 
& {k} \ar@/^1pc/[ul]^{f} \ar@/_1pc/[ur]_{g} \ar@{.>}[u]|{h} &
}$}
\end{equation}
such that, given any other $h'$, $\sigma_1$ and $\sigma_2$ as below
\[
\raise15pt\hbox{$
\xymatrix{
{m} & {m\times n} 
\ar@{=>}[]!<-3ex,-3ex>;[dl]!<5ex,3ex>_{\sigma_1}
\ar@{=>}[]!<3ex,-3ex>;[dr]!<-5ex,3ex>^{\sigma_2}
\ar[l]_{\pi_1} \ar[r]^{\pi_2} &  {n} \\ 
& {k} \ar@/^1pc/[ul]^{f} \ar@/_1pc/[ur]_{g} \ar@{.>}[u]|{h'} &
}$}
\]
there exists unique $\xi: h' \Ra h$ such that 
\[
\raise20pt\hbox{$
\xymatrix{
m & {m\times n} \ar[l]_{\pi_1} 
\ar@{=>}[]!<-4ex,-3ex>;[dl]!<4ex,3ex>_{\rho_1}
& 
\ar@{=>}[]!<-8ex,-4ex>;[ll]!<9ex,-4ex>_{\xi}
\\
& k \ar@/_1pc/[u]_{h'} \ar@/^1pc/[u]|h \ar@/^1pc/[ul]^f
}$}
\hspace{-1.5pc}
=
\raise20pt\hbox{$
\xymatrix{
m & {m\times n} \ar[l]_{\pi_1} 
\ar@{=>}[]!<-3ex,-3ex>;[dl]!<4ex,3ex>_{\sigma_1}
\\
& k \ar[u]_{h'} \ar@/^1pc/[ul]^f
}$}
\quad \text{and} \quad
\raise20pt\hbox{$
\xymatrix{
\ar@{=>}[]!<8ex,-4ex>;[rr]!<-8ex,-4ex>^{\xi}
& {m\times n} \ar[r]^{\pi_2} 
\ar@{=>}[]!<4ex,-3ex>;[dr]!<-4ex,3ex>^{\rho_2}
 & n  
\\
& k \ar@/_1pc/[u]|{h} \ar@/^1pc/[u]^{h'} \ar@/_1pc/[ur]_g & 
}$}
=
\raise20pt\hbox{$
\xymatrix{
{m\times n} \ar[r]^{\pi_2} \ar@{=>}[]!<3ex,-3ex>;[dr]!<-4ex,3ex>^{\sigma_2} & n \\
k \ar[u]^{h'} \ar@/_1pc/[ur]_g &}$}
\]

%\marginpar{THIS THEOREM SHOULD BE AN IFF}
%\marginpar{WHAT DOES IT MEAN "compatible with the monoidal 
%product in the obvious way"?}
\begin{thm}[Carboni and Walters]\label{thm:laxproduct}
Suppose that $\cat{C}$ has a lax product structure. Then the monoidal product
of $\cat{C}$ is a lax product.
\end{thm}
\begin{proof}
The projections are 
\[
\pi_1 \quad \Defeq \quad \lower15pt\hbox{$\includegraphics[height=1.5cm]{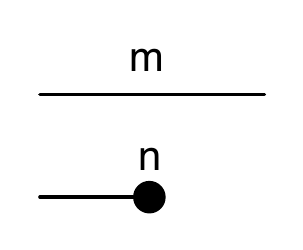}$}
\qquad
\pi_2 \quad \Defeq \quad \lower15pt\hbox{$\includegraphics[height=1.5cm]{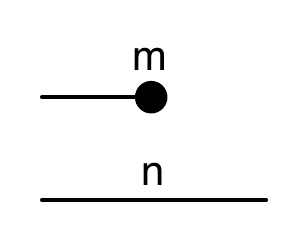}$}.
\]
It is easy to show that the universal property holds.
Indeed,
given $f:k\to m$, $g:k\to n$, we see that
\[
\includegraphics[height=1.75cm]{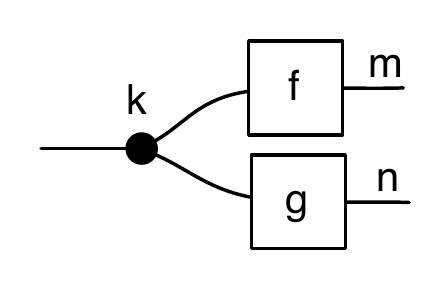}
\]
gives us $\rho_1$ and $\rho_2$ (\eqref{eq:laxprod1}) since
\[
\includegraphics[height=1.75cm]{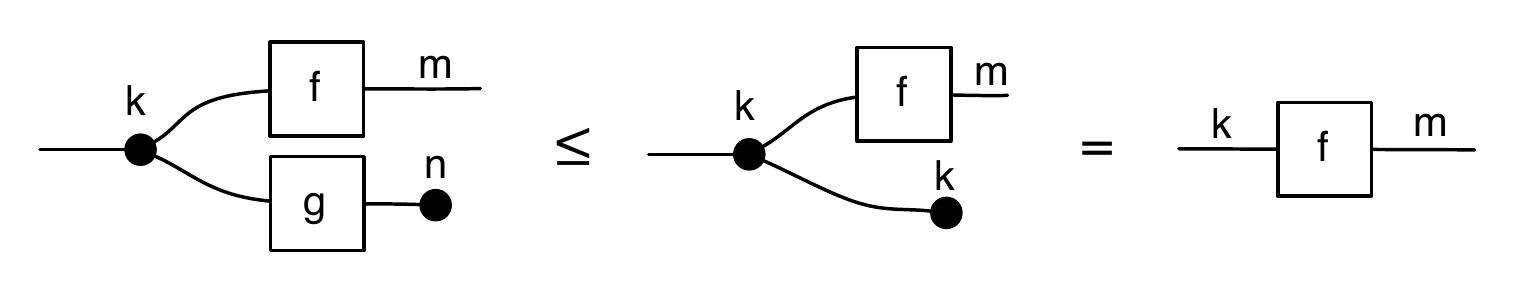}
\]
and similarly
\[
\includegraphics[height=1.75cm]{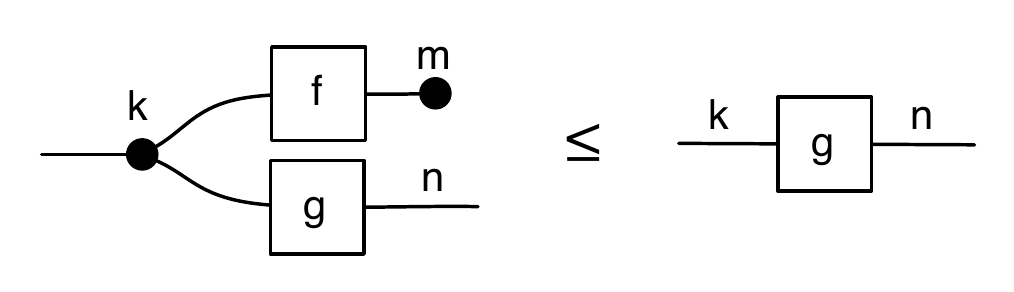}.
\]
Given any $h':k\to m\oplus n$ with
\[
\lower13pt\hbox{$\includegraphics[height=1.3cm]{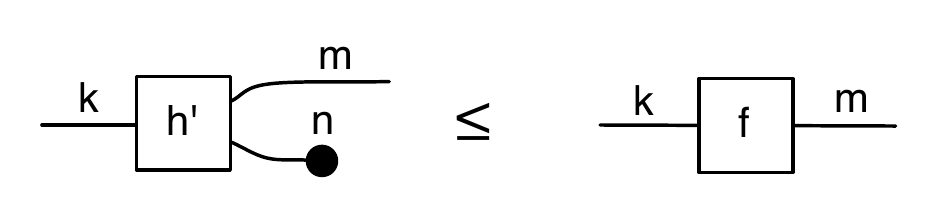}$}
\quad \text{and} \quad
\lower13pt\hbox{$\includegraphics[height=1.3cm]{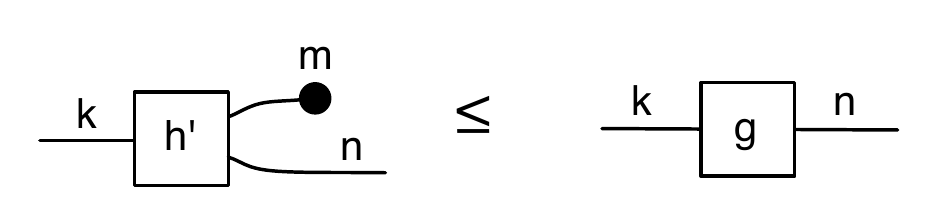}$}
\]
We have
\[
\includegraphics[height=2cm]{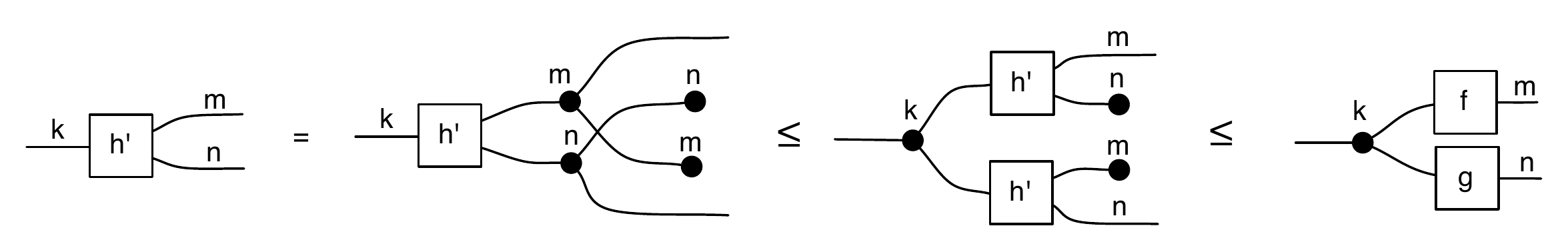}.
\]
\end{proof}

\begin{corollary}
Let $\Theory{T}$ be an SMIT. Then in $\freeLPPROP{\Theory{T}}$, the
monoidal product is a lax product.
\end{corollary}
\begin{proof}
By Theorem~\ref{thm:lptlaxhom} we know that every arrow $\freeLPPROP{\Theory{T}}$ is a lax comonoid homomorphism,
thus the monoidal product is a lax product by Theorem~\ref{thm:laxproduct}.
\end{proof}

\begin{corollary}\label{cor:laxproductRel}
In $\Rel$, considered as a 2-category (the 2-cells are set inclusions), the cartesian product is a lax product.
\end{corollary}
\begin{proof}
It suffices to show that every relation is a lax comonoid homomorphism. The lax product structure is then obtained as in the proof of Theorem~\ref{thm:laxproduct}.
\end{proof}

The appropriate notion of model for a lax product theory is a monoidal functor that preserves lax product structure, in the sense of Definition~\ref{defn:laxproductstructure}.  The notion of homomorphism of models is then a lax monoidal natural transformation between such functors.

More concretely, we define a \emph{lax product category} to be a symmetric monoidal category, enriched over the category of posets, that contains a lax product structure in the sense of Definition~\ref{defn:laxproductstructure}.
%where the monoidal product $\tns$ is the lax product. 
A lax product prop is both an ordered prop and lax product category. A lax product functor is a poset-enriched functor that preserves the lax product structure. Fixing a lax product category $\catC$, models for an LPT $\Theory{T}$ are lax product functors $F\colon \freeLPPROP{\Theory{T}} \to \catC$. 

In order to avoid duplication, we postpone a more comprehensive discussion to Section~\ref{sec:frobenius}, and first introduce \emph{Frobenius theories}, which are a particularly interesting and expressive variant of lax product theories.

\section{Frobenius Theories and their models}\label{sec:frobenius}
In this section we introduce the main contribution of this work: Frobenius theories and their models. To this aim we recall the concept of \emph{cartesian bicategory of relations} from~\cite{Carboni1987} and flesh out some of its properties.

Recall that an arrow $f:B\to C$ in a 2-category has a right adjoint $g:C\to B$ when there exist 2-cells 
$\eta: \id_B \to f;g$ and $\epsilon: g;f \to \id_C$ satisfying the well-known triangle equations. In 
the poset-enriched case, this simplifies to requiring merely
\begin{equation}\label{eq:adjoints}
\id_B \leq f\poi g \qquad g\poi f \leq \id_C.
\end{equation}

In Section~\ref{sec:laxproducttheories} we saw a particular emphasis on the the SMT of commutative comonoids (Example~\ref{ex:equationalprops}\ref{it:comonoids}).
Given~\eqref{eq:adjoints}, to define adjoints to the generators $\{\,\Bcomult,\,\Bcounit\,\}$, in a lax product theory, it suffices to add generators $\{\,\Bmult,\,\Bunit\,\}$ and inequations
%\begin{equation}\label{eq:multcomultadj}
\[
\lower10pt\hbox{$\includegraphics[height=1.2cm]{graffles2/relnabladelta}
\quad
\includegraphics[height=1.2cm]{graffles2/reldeltanabla}$}
\]
%\end{equation}
%\begin{equation}\label{eq:unitcounitadj}
\[
\lower8pt\hbox{$
\includegraphics[height=1cm]{graffles2/relunitsadjunit}
\quad
\includegraphics[height=1cm]{graffles2/relunits}$}
\]
%\end{equation}

Next, recall from Example~\ref{ex:equationalprops}\ref{it:frobenius} that the SMT $\Frobtheory$ of special Frobenius monoids has as its set of generators
\{\,\Bmult,\,\Bunit,\,\Bcomult,\,\Bcounit\,\} 
and as equations, those that guarantee that the generators above form, respectively, a commutative monoid 
and commutative comonoid, together with the Frobenius equation and the special equation:
\begin{equation}\label{eq:specialfrobenius}
\lower10pt\hbox{$
\lower7pt\hbox{$\includegraphics[height=1.5cm]{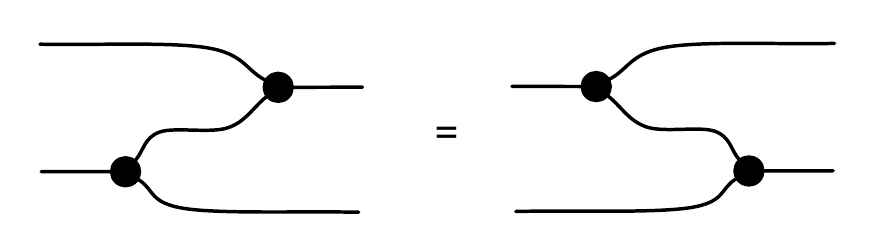}$}
\qquad
\includegraphics[height=1.2cm]{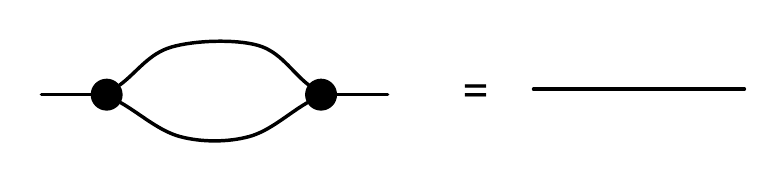}$}
\end{equation}

Succinctly, cartesian bicategories of relations are lax product categories where the lax product structure has right adjoints (a commutative monoid), satisfying the equations of special Frobenius monoids. We spell out the details below.

%\subsection{Cartesian Bicategories of Relations}\label{ssec:CartBic}
%\marginpar{PAWEL: PLEASE CHECK THIS DEFINITION: I COPIED IT FROM CARBONI WALTER}

\begin{definition}
A \emph{cartesian bicategory of relations} is a poset enriched category that is symmetric monoidal and additionally
\begin{enumerate}
\item for every object $n$, there are 
arrows $\dup_n \colon n \to n \tns n$ and $\dis_n \colon n\to I$, graphically denoted by $\Bcomultn$ and $\Bcounitn$ forming a cocommutative comonoid; 
\item for every object $n$, there are arrows $\codup_n \colon n \tns n \to n $ and $\codis_n \colon I\to n$, graphically denoted by $\Bmultn$ and $\Bunitn$, forming a commutative monoid;
\item such that the monoids and the comonoids satisfy the following inequations
\begin{multicols}{2}
\noindent
\begin{equation}\label{eq:adj1}
\lower12pt\hbox{$\includegraphics[height=1.2cm]{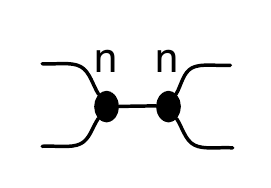}$}
\leq
\lower10pt\hbox{$\includegraphics[height=1.2cm]{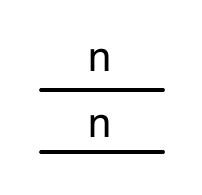}$}
\end{equation}
\begin{equation}\label{eq:adj2}
\lower9pt\hbox{$\includegraphics[height=1.2cm]{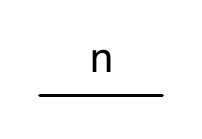}$}
\leq
\lower12pt\hbox{$\includegraphics[height=1.2cm]{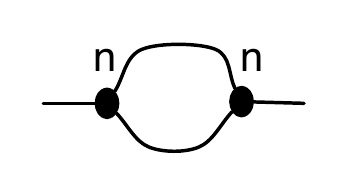}$}\end{equation}
\end{multicols}
\begin{multicols}{2}
\noindent
\begin{equation}\label{eq:adj3}
\lower5pt\hbox{$\includegraphics[height=1cm]{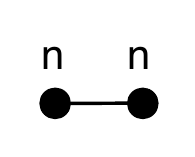}$}
\leq
%\lower5pt\hbox{$\includegraphics[height=.6cm]{graffles/injectiveAG2.pdf}$}
id_I
\end{equation}
\begin{equation}
\label{eq:adj4}
\lower8pt\hbox{$\includegraphics[height=1cm]{graffles/idn.pdf}$}
\leq
\lower5pt\hbox{$\includegraphics[height=1cm]{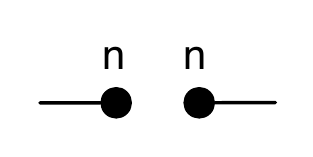}$}
\end{equation}
\end{multicols}
\begin{equation}\label{eq:frobn}
\lower17pt\hbox{$\includegraphics[height=1.5cm]{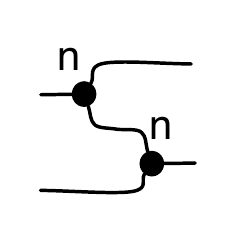}$}
=
\lower12pt\hbox{$\includegraphics[height=1.2cm]{graffles/Xn.pdf}$}
=
\lower17pt\hbox{$\includegraphics[height=1.5cm]{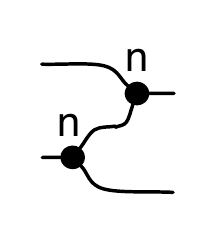}$}
\end{equation}

\item and, every arrow $R\colon m \to n$ is a lax comonoid homomorphism. 
%\begin{multicols}{2}
%\noindent
%\begin{equation}\label{eq:CBCH1}
%\lower5pt\hbox{$\includegraphics[height=.6cm]{graffles/singlevaluedAG2.pdf}$}
%\leq
%\lower5pt\hbox{$\includegraphics[height=.6cm]{graffles/singlevaluedAG1.pdf}$}
%\end{equation}
%\begin{equation}\label{eq:CBCH2}
%\lower5pt\hbox{$\includegraphics[height=.6cm]{graffles/totalAG2.pdf}$}
%\leq
%\lower5pt\hbox{$\includegraphics[height=.6cm]{graffles/totalAG1.pdf}$}
%\end{equation}
%\end{multicols}
\begin{multicols}{2}
\noindent
\begin{equation}\label{eq:CBCH1}
\lower10pt\hbox{$\includegraphics[height=1cm]{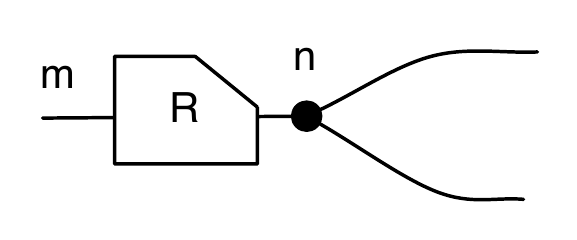}$}
\leq
\lower13pt\hbox{$\includegraphics[height=1.2cm]{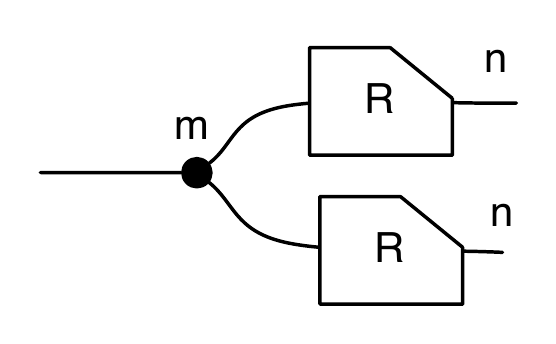}$}
\end{equation}
\begin{equation}
\label{eq:CBCH2}
\lower8pt\hbox{$\includegraphics[height=1cm]{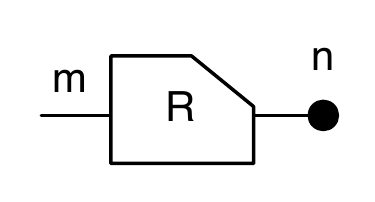}$}
\leq
\lower5pt\hbox{$\includegraphics[height=1cm]{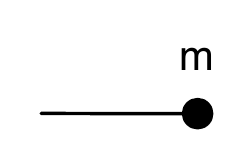}$}
\end{equation}
\end{multicols}
\end{enumerate}
\end{definition}

Inequations \eqref{eq:adj1} and \eqref{eq:adj2} state that $\Bcomultn$ is the left adjoint of $\Bmultn$, while inequations \eqref{eq:adj3} and \eqref{eq:adj4} state that $\Bcounitn$ is the left adjoint of $\Bunitn$. Note that \eqref{eq:frobn} is %the same 
%-- modulo the white-black colouring of the multiplication -- of axiom
\eqref{eq:BWFrob} of the SMT of special Frobenius algebra (Example \ref{ex:equationalprops} (c)). The other equation \eqref{eq:BWSep} holds in any Cartesian bicategory of relations: one direction is given by \eqref{eq:adj2} and the other is proved as follows. 
\begin{equation*}
\lower10pt\hbox{$\includegraphics[height=1cm]{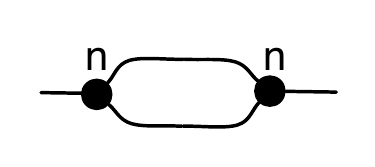}$}
\leq
\lower10pt\hbox{$\includegraphics[height=1.3cm]{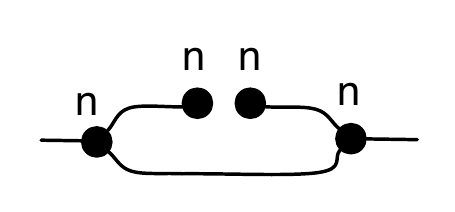}$}
=
\lower5pt\hbox{$\includegraphics[height=1cm]{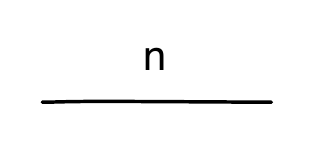}$}
\end{equation*}
Therefore, we will often refer to the monoid and the comonoid of a Cartesian bicategory of relations $\catC$ as to the \emph{(special) Frobenius structure} of $\catC$.

Since, point 4 requires every arrow to be a lax comonoid homomorphism, we know by Theorem \ref{thm:laxproduct}, that the monoidal product $\tns$ is a lax-product. This helps us in showing that $\Rel$ is a Cartesian bicategory of relations: indeed, from Corollary \ref{cor:laxproductRel} we know that the lax product in $\Rel$ is simply the Cartesian one. Now for every set $X$, $\dup_X \colon X \to X\times X$ is the diagonal relation $\{( \,x, \,(x,x) \,) \text{ s.t. } x\in X\}$, $\dis_X \colon X\to 1=\{\bullet\}$ is the relation  $\{(x,\bullet) \text{ s.t. } x\in X\}$, $\codup_X \colon X \times X \to X$ and $\dis_X \colon 1\to X$ are, respectively, their opposite relations. One can easily check that the inequations in point 3 hold.

\medskip

A \emph{cartesian bifunctor} is the notion of structure-preserving homomorphism between cartesian bicategories of relations. In fact, it suffices to require that the lax product structure, in the sense of Definition~\ref{defn:laxproductstructure} is preserved, that is, the notion of cartesian bifunctor is the same as lax product functor. Indeed: 
the fact that adjoints are preserved follows from (2-)functoriality, since adjoints, if they exist, are uniquely defined in ordered categories. 

%\marginpar{PAWEL PLEASE INTRODUCE Cartesian bifunctor. Is it enough to require that they are 2-functors? See Corollary 1.7 in Carboni Walter}
%Introduce Cartesian bifunctors as bifunctors preserving the Frobenius structures (This should also imply that they preserve lax product)

\medskip

\paragraph{Compact closed structure.}
In any Cartesian bicategory of relations we have a self-dual compact closed structure. To describe it, we adopt the following 
graphical notation.
\[
\lower10pt\hbox{$\includegraphics[height=1cm]{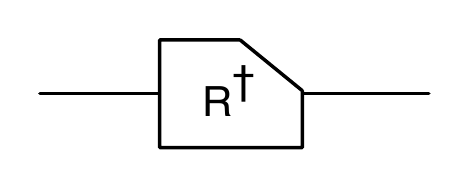}$}
\Defeq
\lower10pt\hbox{$\includegraphics[height=1cm]{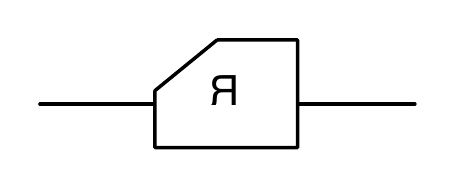}$}
\Defeq 
\lower15pt\hbox{$\includegraphics[height=1.5cm]{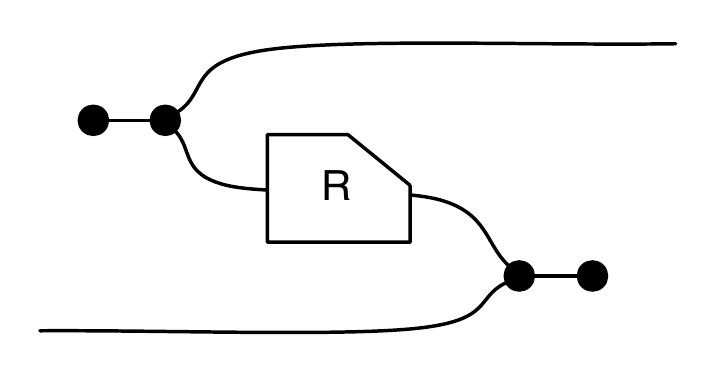}$}
\]
For an intuition consider $\Rel$: it is easy to check that $R^\dagger$ is just the opposite relation of $R$.

\begin{lemma}%[$\dagger$ is a monoidal 2-functor]
\label{lemma:dagger}
If $\cat{C}$ is a Cartesian bicategory of relations, then
$\dagger \colon \cat{C}^{op}\to\cat{C}$ is a 2-functor.
\begin{enumerate}[(i)]
\item $\id^\dagger = \id$
\item $(R;S)^\dagger = S^\dagger;R^\dagger$
\item $(R\oplus S)^\dagger = R^\dagger \oplus S^\dagger$
\item if $\lower10pt\hbox{$\includegraphics[height=1cm]{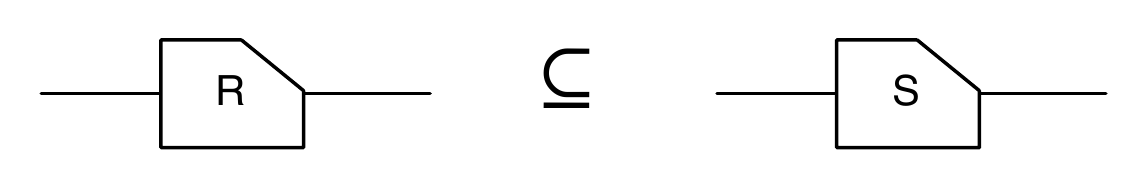}$}$
then  $\lower10pt\hbox{$\includegraphics[height=1cm]{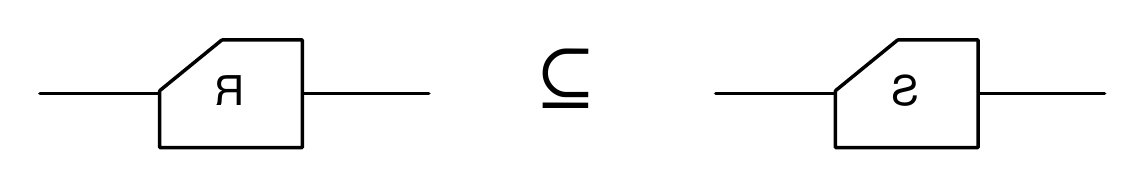}$}$.
\end{enumerate}
\end{lemma}
\begin{proof}
(i) is easy to check. It is the so called snake lemma, for (ii)
\[
\includegraphics[height=2cm]{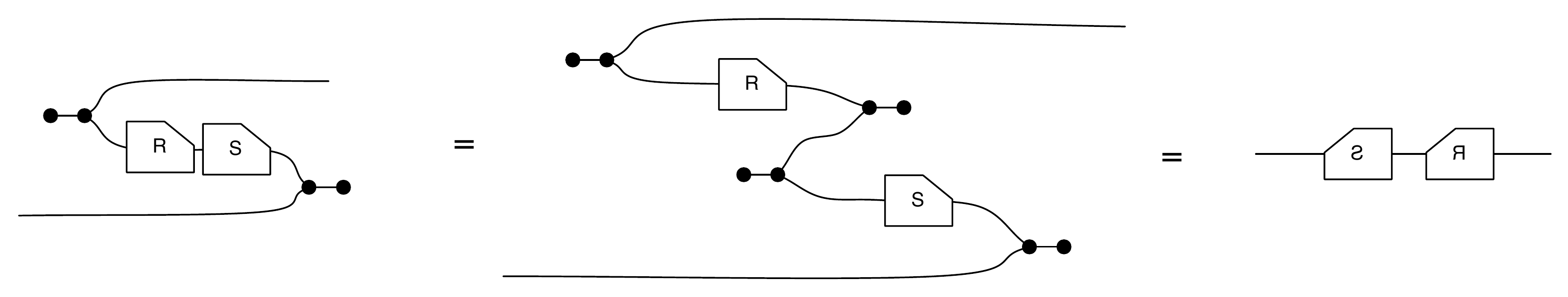}
\]
for (ii)
\[
\includegraphics[height=2.5cm]{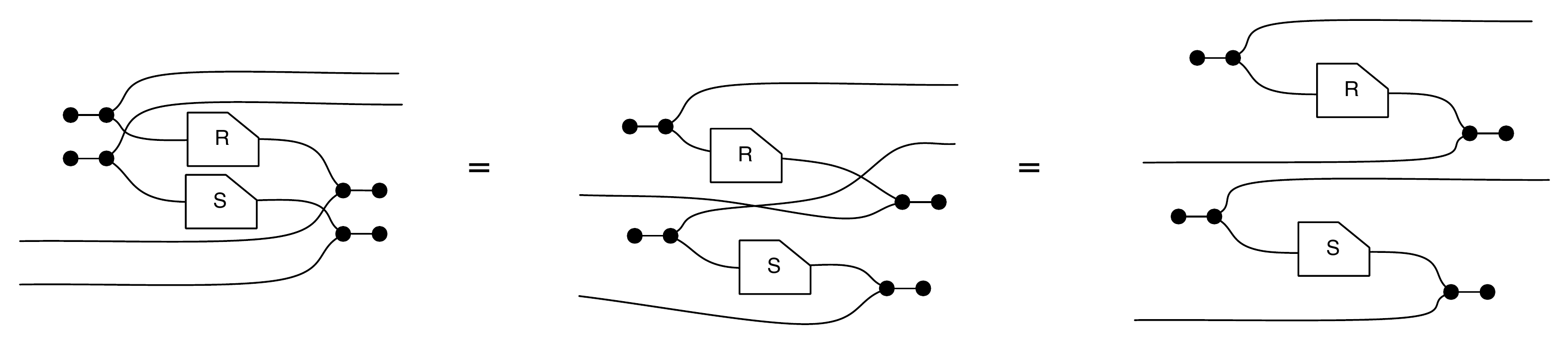}
\]
for (iv)
\[
\includegraphics[height=1.5cm]{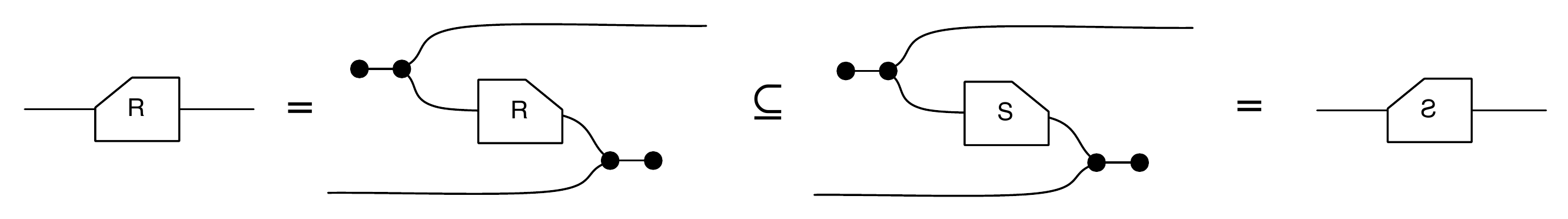}
\]
\end{proof}

\paragraph{Maps and comaps.}
By applying Lemma \ref{lemma:dagger} (ii) and (iv) to inequations \eqref{eq:CBCH1} and \eqref{eq:CBCH2}, one obtains that the following two hold for any arrow $R\colon m\to n$ of a Cartesian bicategory of relations.
\begin{multicols}{2}
\noindent
\begin{equation}\label{eq:inverselaxhom1}
\lower5pt\hbox{$\includegraphics[height=.8cm]{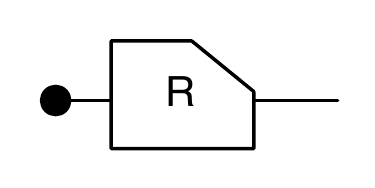}$}
\leq
\lower1pt\hbox{$\includegraphics[height=.5cm]{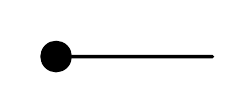}$}
\end{equation}
\begin{equation}\label{eq:inverselaxhom2}
\lower10pt\hbox{$\includegraphics[height=1cm]{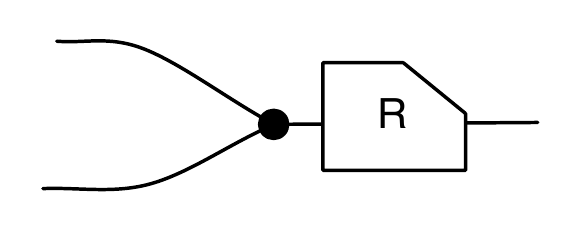}$}
\leq
\lower11pt\hbox{$\includegraphics[height=1.2cm]{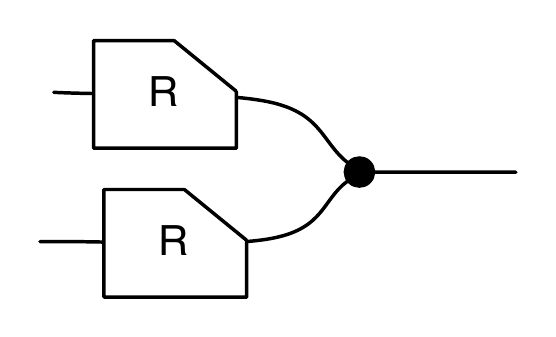}$}
\end{equation}
\end{multicols}

The following inequations do \emph{not} in general hold.

\begin{multicols}{2}
\noindent
\begin{equation}\label{eq:sv}\tag{SV}
\lower9pt\hbox{$\includegraphics[height=.9cm]{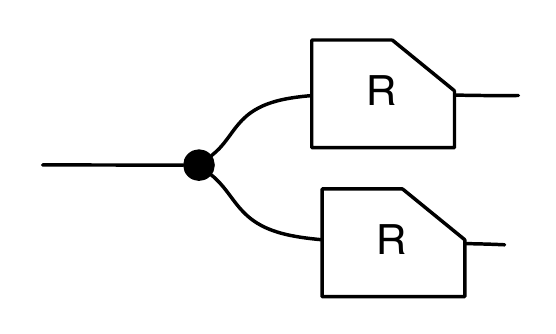}$}
\leq
\lower7pt\hbox{$\includegraphics[height=.7cm]{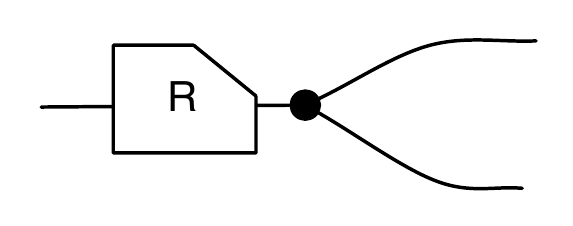}$}
\end{equation}
\begin{equation}\label{eq:tot}\tag{TOT}
\lower4pt\hbox{$\includegraphics[height=.4cm]{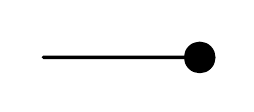}$}
\leq
\lower7pt\hbox{$\includegraphics[height=.7cm]{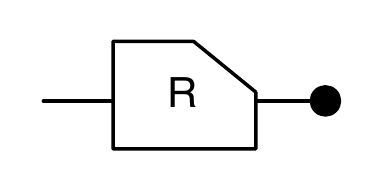}$}\end{equation}
\end{multicols}
\begin{multicols}{2}
\noindent
\begin{equation}\label{eq:inj}\tag{INJ}
\lower9pt\hbox{$\includegraphics[height=.9cm]{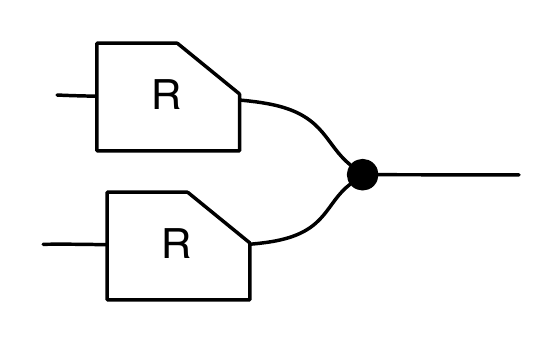}$}
\leq
\lower7pt\hbox{$\includegraphics[height=.7cm]{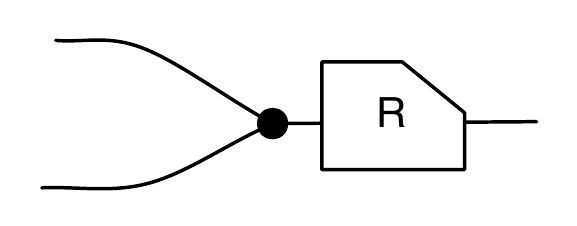}$}
\end{equation}
\begin{equation}
\label{eq:sur}\tag{SUR}
\lower4pt\hbox{$\includegraphics[height=.4cm]{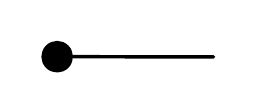}$}
\leq
\lower7pt\hbox{$\includegraphics[height=.7cm]{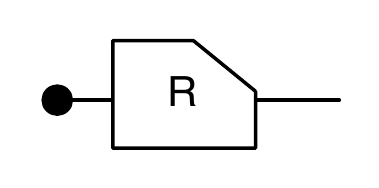}$}
\end{equation}
\end{multicols}
An arrow $R\colon m\to n$ in a cartesian bicategory is said to be \emph{single valued} iff satisfies \eqref{eq:sv}, \emph{total} iff satisfies \eqref{eq:tot}, \emph{injective} iff satisfies \eqref{eq:inj} and \emph{surjective} iff satisfies \eqref{eq:sur}. A \emph{map} is an arrow that is both single valued and total, namely a comonoid homomorphism. A \emph{comap} is an arrow that is both injective and surjective, namely a monoid homomorphism. It is easy to see that in $\Rel$, these coincide with the familiar notions.

\medskip

The following simple lemma will be useful in subsequent proofs.
\begin{lemma}[Wrong way]
\[
\includegraphics[height=1.4cm]{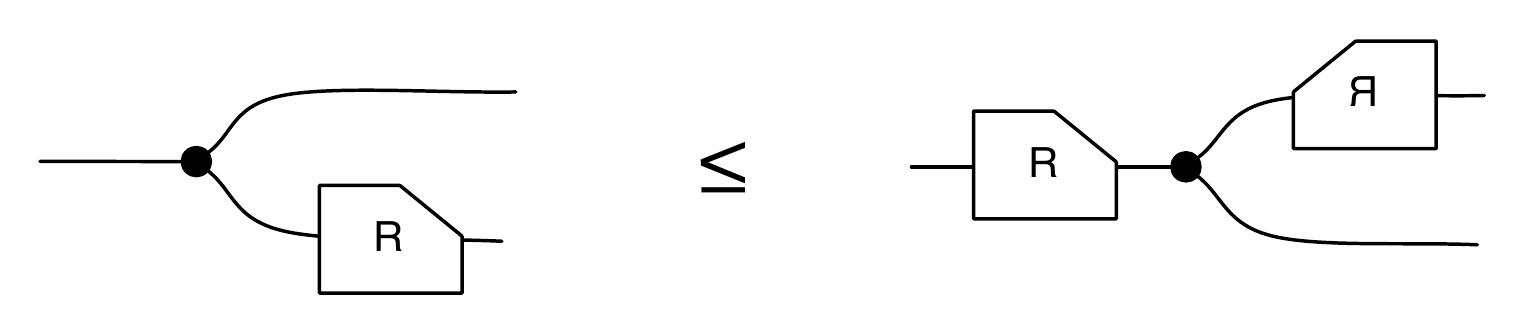}
\]
\end{lemma}
\begin{proof}
\[
\includegraphics[height=2cm]{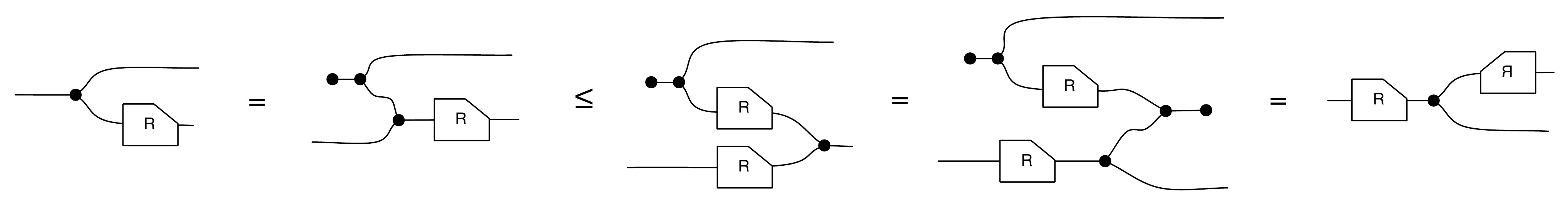}
\]

\end{proof}

\begin{lemma}\label{lem:characterizationmap}
Consider the following inequalities.
\begin{multicols}{2}
\noindent
\begin{equation}\label{eq:adjcounit} 
\lower10pt\hbox{$\includegraphics[height=1cm]{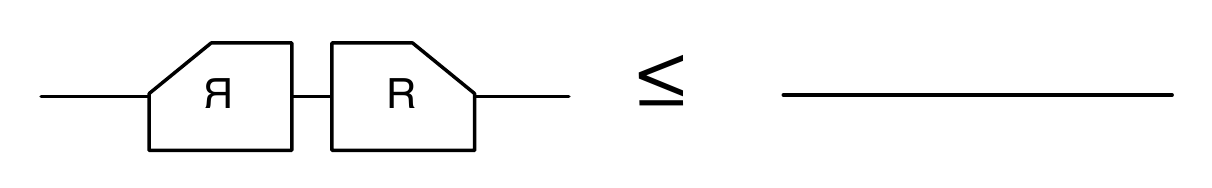}$}
\end{equation}
\begin{equation}\label{eq:adjunit}
\lower10pt\hbox{$\includegraphics[height=1cm]{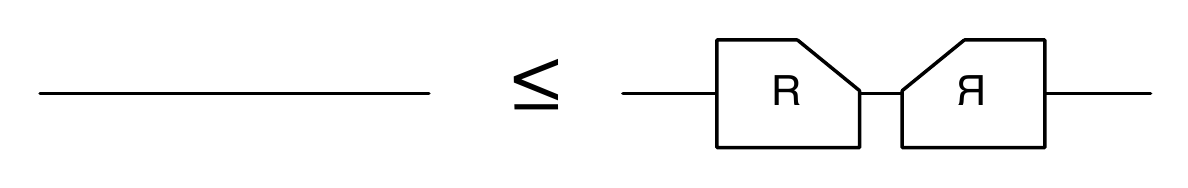}$}
\end{equation}
\end{multicols}
\begin{multicols}{2}
\noindent
\begin{equation}\label{eq:opadjcounit}
\lower10pt\hbox{$\includegraphics[height=1cm]{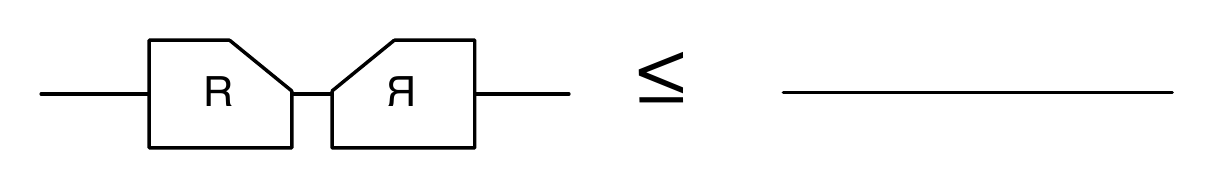}$}
\end{equation}
\begin{equation}
\label{eq:opasjunit}
\lower10pt\hbox{$\includegraphics[height=1cm]{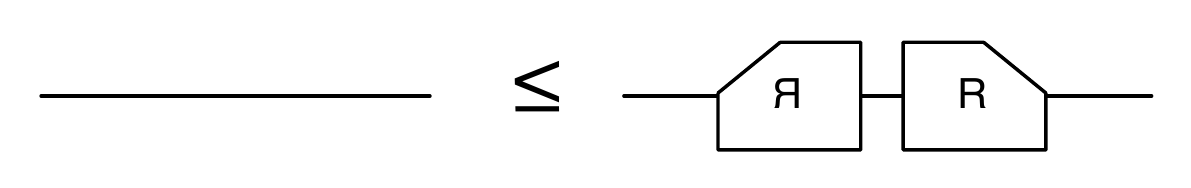}$}
\end{equation}
\end{multicols}
An arrow $R \colon m \to n$ is single valued iff \eqref{eq:adjcounit}, it is total iff \eqref{eq:adjunit}, it is injective iff \eqref{eq:opadjcounit} and it is surjective iff \eqref{eq:opasjunit}.
\end{lemma}
\begin{proof}
\eqref{eq:sv} $\Ra$ \eqref{eq:adjcounit}
\[
\includegraphics[height=1.5cm]{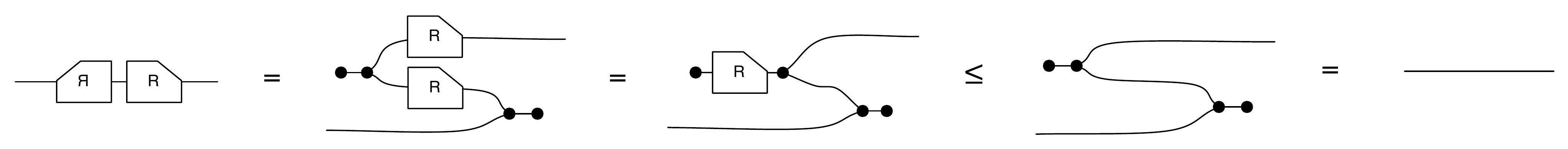}
\]
\eqref{eq:adjcounit} $\Ra$ \eqref{eq:sv}
\[
\includegraphics[height=1.5cm]{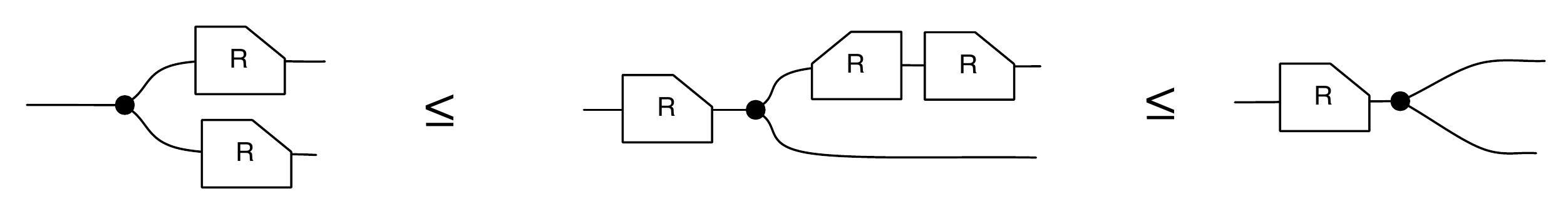}
\]
\eqref{eq:tot} $\Ra$ \eqref{eq:adjunit}
\[
\includegraphics[height=1.5cm]{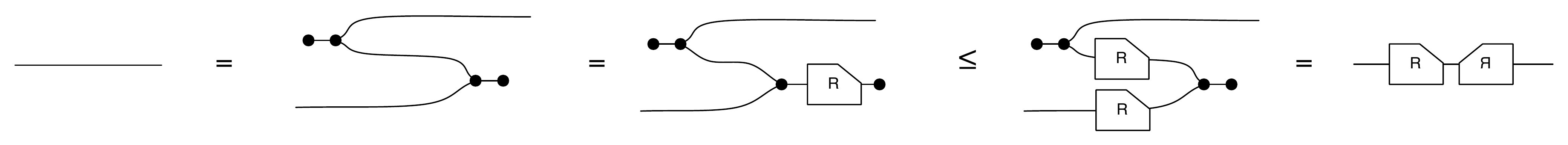}
\]
\eqref{eq:adjunit} $\Ra$ \eqref{eq:tot}
\[
\includegraphics[height=1cm]{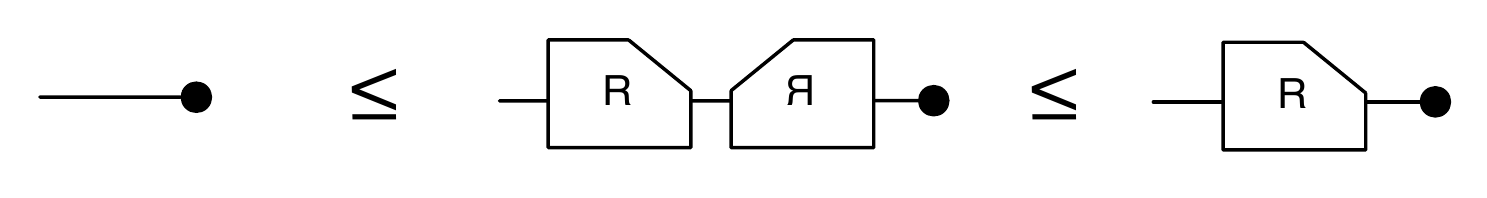}
\]

\eqref{eq:inj} $\Leftrightarrow$ \eqref{eq:opadjcounit} follows as corollary of \eqref{eq:sv} $\Leftrightarrow$ \eqref{eq:adjcounit}.
\eqref{eq:sur} $\Leftrightarrow$ \eqref{eq:opasjunit} follows as corollary of \eqref{eq:tot} $\Leftrightarrow$ \eqref{eq:adjunit}.
\end{proof}

\begin{corollary}\label{cor:mapssmaller}
Let $R,S\colon m\to n$ be two maps. If $R\leq S$ then $R=S$.
\end{corollary}

\begin{corollary}[Span cancellation]\label{cor:spancancellation}
An arrow $R$ is single valued and surjective iff 
$$\lower10pt\hbox{$\includegraphics[height=1cm]{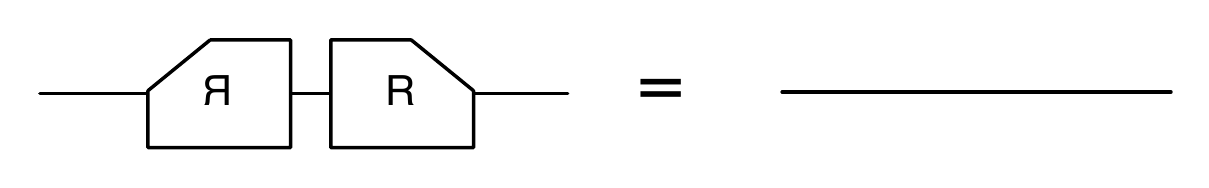}$}$$
\end{corollary}

\begin{corollary}[Cospan cancellation]\label{cor:cospancancellation}
An arrow $R$ is injective and total iff
$$\lower10pt\hbox{$\includegraphics[height=1cm]{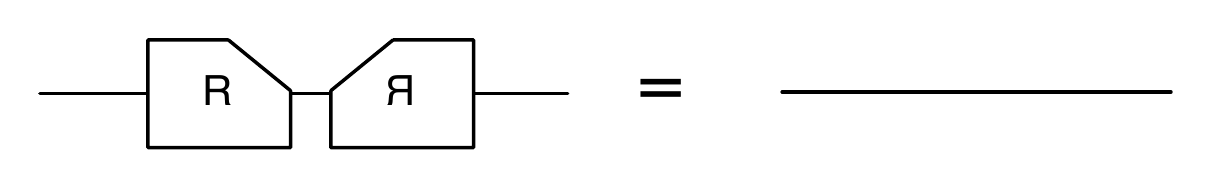}$}$$
\end{corollary}

\begin{lemma}\label{lem:dagger}
$R$ is a map iff it has a right adjoint.
Moreover, the right adjoint is $R^\dagger$.
\end{lemma}
\begin{proof}
If $R$ is a map, then it is total and single valued.
The fact that it is total gives the unit~\eqref{eq:adjunit} of the adjunction and the fact that it is single valued gives the 
counit~\eqref{eq:adjcounit}. 

Suppose that $R$ has right adjoint $S$, i.e.
\[
\includegraphics[height=1cm]{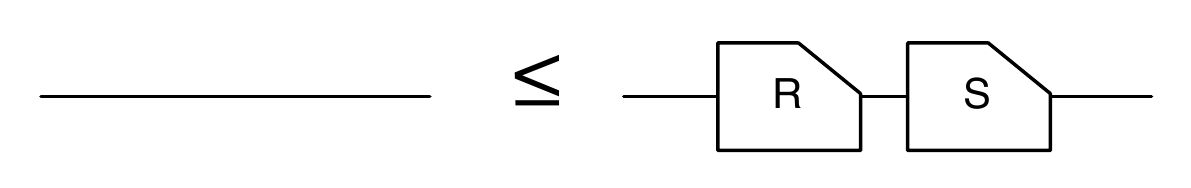}
\quad
\includegraphics[height=1cm]{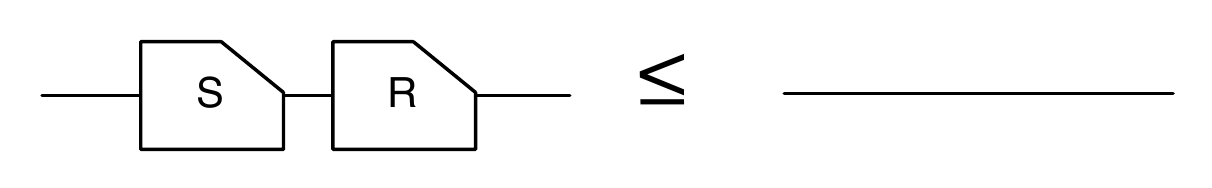}
\]
Then 
\[
\includegraphics[height=1.5cm]{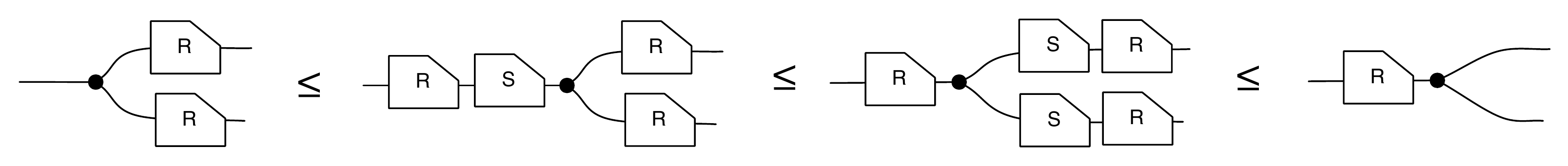}
\]
So $R$ is single valued, and 
\[
\includegraphics[height=1cm]{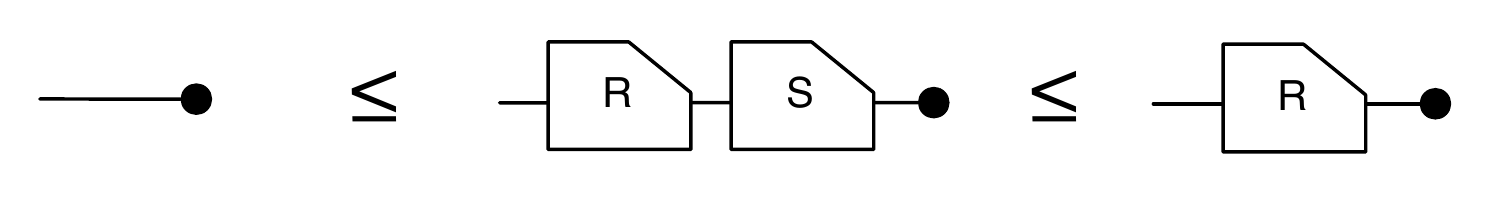}
\]
so $R$ is total. By Lemma~\ref{lem:characterizationmap},
$R^\dagger$ is also right adjoint to $R$. Thus, by the standard argument, we have
\[
\includegraphics[height=1cm]{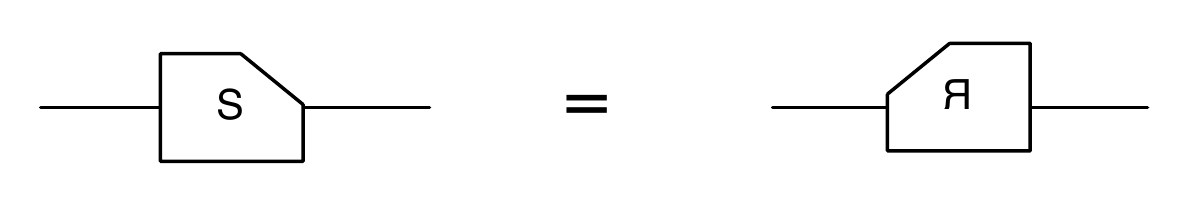}.
\]
\begin{corollary}\label{cor:inverse}
An arrow $R$ is comap iff it has a left adjoint $S$, and in that case $S=R^\dagger$. 
\end{corollary}
\begin{corollary}
An arrow $R$ is an isomorphism iff $R^{-1}=R^\dagger$.
\end{corollary}

\end{proof}

\paragraph{Convolution.} We now show that in a cartesian bicategory of relations every homset is a meet semi-lattice. 
Given arrows $A \colon m\to n$ and $B \colon m\to n$ we define the \emph{convolution} of $A$ and $B$, written $A\owedge B$, as follows:
\[
\owedge \Defeq \Delta_m \mathrel{;} A\oplus B \mathrel{;} \nabla_n = \lower20pt\hbox{\includegraphics[height=1.75cm]{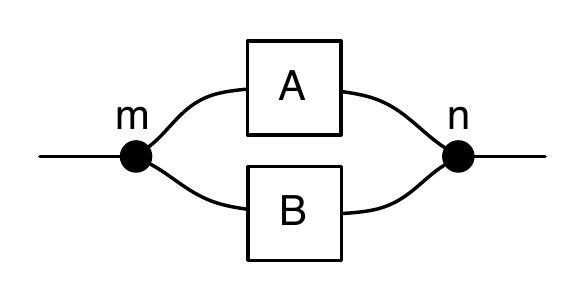}}
\]
\begin{lemma}
Convolution is associative, commutative, idempotent and unital, with unit: 
\[
\top_{m,n} \Defeq \dis_m ; \codis_n = \lower7pt\hbox{\includegraphics[height=1cm]{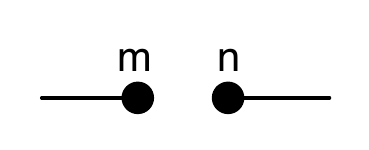}}
\]
\end{lemma}
\begin{proof}
The proof for associativity, commutativity and unitality is trivial using the fact that $(\dup, \dis)$ is a commutative comonoid and
$(\codup, \codis)$ is a commutative monoid. For idempotency we observe the following two inequalities.

\begin{equation*}
\lower10pt\hbox{$\includegraphics[height=0.8cm]{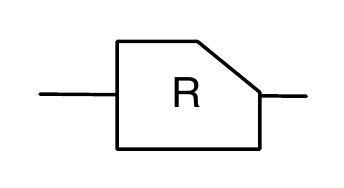}$}
=
\lower12pt\hbox{$\includegraphics[height=1.1cm]{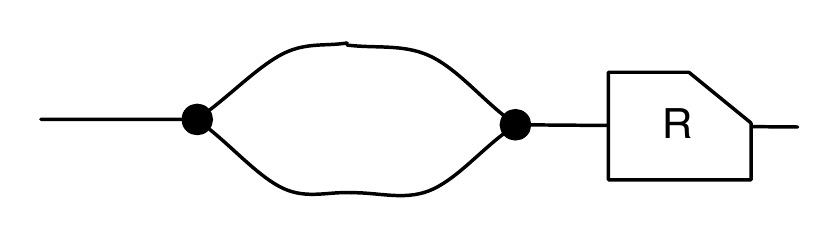}$}
\leq
\lower18pt\hbox{$\includegraphics[height=1.5cm]{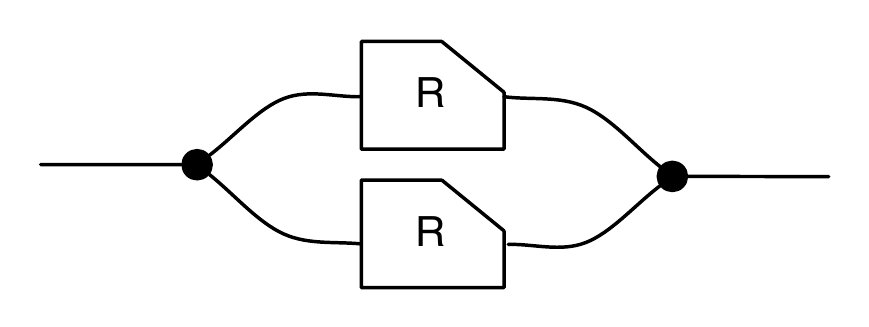}$}
\end{equation*}

\begin{equation*}
\lower16pt\hbox{$\includegraphics[height=1.5cm]{graffles/idempoctency3.pdf}$}
\leq
\lower16pt\hbox{$\includegraphics[height=1.5cm]{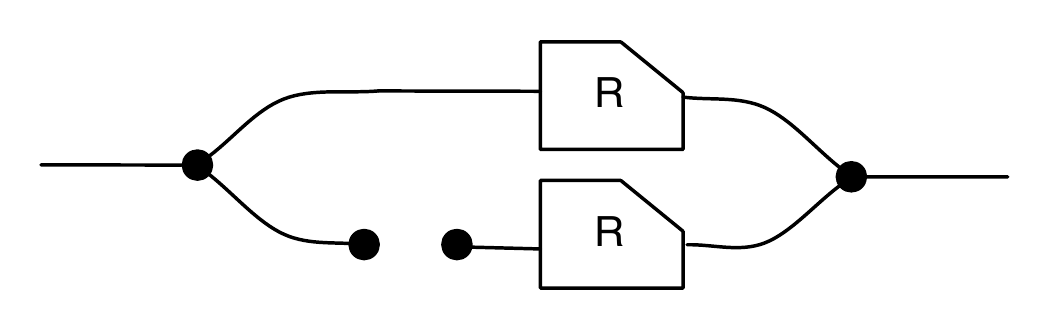}$}
\leq
\lower14pt\hbox{$\includegraphics[height=1.3cm]{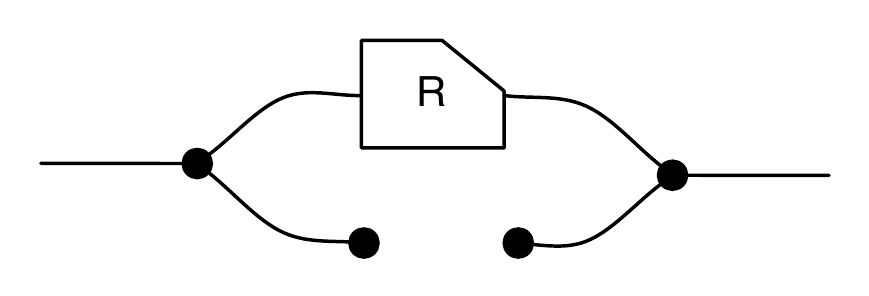}$}
=
\lower8pt\hbox{$\includegraphics[height=0.8cm]{graffles/idempotency1.pdf}$}
\end{equation*}

\end{proof}

Since $\owedge$ is associative, commutative and idempotent, it induces an ordering. In the following we show that this ordering is exactly $\leq$. 

\begin{lemma}\label{lemma:top}
For all arrows $R\colon n\to m$,
$R \leq \top_{m,n}$.
\end{lemma}
\begin{proof}
\begin{equation*}
\lower2pt\hbox{$
\lower8pt\hbox{$\includegraphics[height=.8cm]{graffles/idempotency1.pdf}$}
\leq
\lower8pt\hbox{$\includegraphics[height=.8cm]{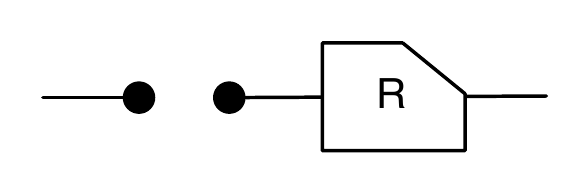}$}
\leq
\lower6pt\hbox{$\includegraphics[height=.6cm]{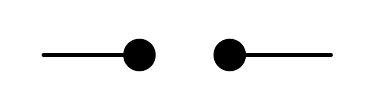}$}
$}
\end{equation*}
\end{proof}

\begin{lemma} 
$R\owedge S = S$ iff $S\leq R$.
\end{lemma}
\begin{proof}
Assume $R \owedge S = S$. Then $R = R  \owedge \top \geq R \ovee S = S$.
Assume $S\leq R$. Then $R \owedge S \geq S\owedge S = S$. Moreover $S = S\owedge \top \geq S \owedge R$.
\end{proof}

\begin{corollary} The partial order $\leq$ is a meet semi-lattice with top.
\end{corollary}

\subsection{Carboni-Walters categories and Frobenius theories}
In ordinary functorial semantics, any finite product category can serve as a semantic domain, with the category
of sets and functions serving as the default choice. 
Cartesian bicategories of relations will serve as the domain for our notion of functorial semantics, with the category
of sets and \emph{relations} as a particularly useful universe.

We call an ordered prop that is, additionally, a cartesian bicategory of relations a \emph{frop} (Frobenius prop) or a \emph{Carboni-Walters category}. 
%The prefix ``AR'' stands for \emph{A}urelio Carboni and \emph{R}obert Walter, the inventors of Cartesian bicategories. 
Similarly to how a cartesian theory results a Lawvere category, a Frobenius theory results in a Carboni-Walters category. To formally introduce this fact is convenient to consider the SMIT $\Theory{CW}$: the signature $\Sigma_{CW}$ consists of $\Bcomult \colon 1 \to 2$, $\Bcounit \colon 1 \to 0$, $\Bmult \colon 2 \to 1$ and  $\Bunit \colon 0 \to 1$. The set $I_{CW}$ contains the inequations for comonoids, monoids and \eqref{eq:adj1}-\eqref{eq:frobn} for $n=1$. The ordered PROP freely generated by $\Theory{CW}$ is a cartesian bicategory of relations: for every natural number $n$, the comonoids and monoid are defined analogously to Example \ref{ex:lawvere} (a). It is then easy to check that the inequations \eqref{eq:adj1}-\eqref{eq:CBCH2} hold.

Similarly to how as cartesian theories implicitly contain the SMT of comonoids when considered as props, Frobenius theories contain implicitly the SMIT $\Theory{CW}$.

\begin{definition}\label{def:FrobeniusTheory} 
A (presentation of a) \emph{Frobenius theory} (FT) is a pair $\Theory{T}=(\Sigma, I)$ consisting of a signature $\Sigma$ and a set of inequations $I$.
The signature $\Sigma$ is a set of \emph{generators} $o \: n\to m$ with \emph{arity} $n$ and \emph{coarity} $m$.  The set of inequations $I$ contains pairs $(t,t' \: n\to m)$ of \emph{Frobenius $\Sigma$-terms}, namely  arrows of the ordered PROP freely generated by the SMIT $(\Sigma \uplus \Sigma_{CW}, I_{CW} \uplus I_{LCH})$ where $I_{LCH}$ is the set containing the inequations \eqref{eq:lcomhom1} and \eqref{eq:lcomhom2} for each generator $\sigma \colon m \to n$ in $\Sigma$.

The Carboni-Walters category \emph{freely generated} by a Frobenius theory $\Theory{T}=(\Sigma,I)$, denoted by $\Frob{\Theory{T}}$,  is the ordered PROP freely generated by the SMIT $(\Sigma \uplus \Sigma_{CW}, I \uplus I_{CW} \uplus I_{LCH})$. The latter will be often referred to as the SMIT corresponding to the Frobenius theory $(\Sigma,I)$.
\end{definition}

\begin{example}\label{ex:FT}~
\begin{enumerate}[(a)]
\item The SMIT of commutative monoids $\CMtheory = (\Sigma_{M}, E_M \uplus E_M^{op})$ from Example \ref{ex:SMIT}(a) can be regarded also as a Frobenius theory. In the corresponding SMIT,  $(\Sigma_{M} \uplus \Sigma_{CW}, E_M \uplus E_M^{op} \uplus I_{CW} \uplus I_{LCH})$ one has two monoidal structures which we refer as  the \emph{white monoid} (coming from $\Sigma_{M}$) and the \emph{black monoid} (coming from $\Sigma_{CW}$). Moreover, since $\Frob{\Theory{\CMtheory}}$ is an CW-category, we have also $\cgr[height=20pt]{Wcomult.pdf} ::= (\cgr[height=20pt]{Wmult.pdf})^\dag$ and $\cgr[height=20pt]{Wcounit.pdf} ::= (\cgr[height=20pt]{Wunit.pdf})^\dag$. Graphically,
\begin{multicols}{2}
\noindent
\begin{equation*}
\lower2pt\hbox{$
\lower5pt\hbox{$\includegraphics[height=.6cm]{graffles/Wcomult.pdf}$}
::=
\lower17pt\hbox{$\includegraphics[height=1.7cm]{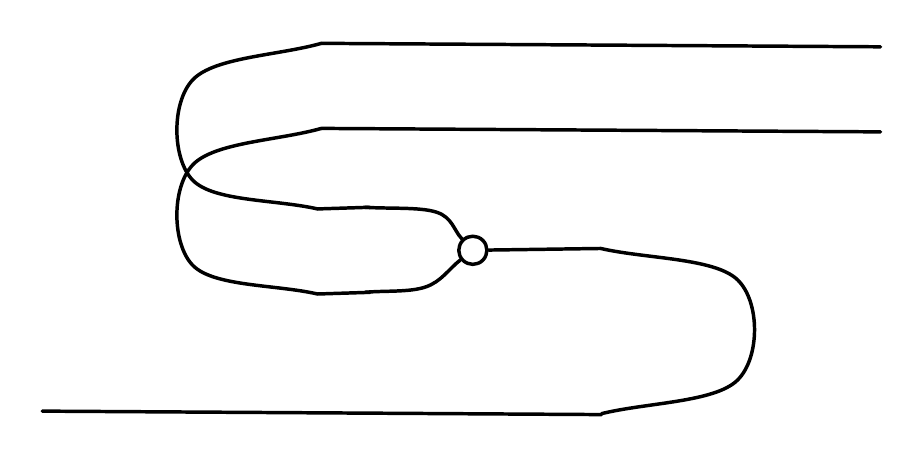}$}
$}
\end{equation*}
\begin{equation*}
\lower2pt\hbox{$
\lower5pt\hbox{$\includegraphics[height=.6cm]{graffles/Wcounit.pdf}$}
::=
\lower5pt\hbox{$\includegraphics[height=.6cm]{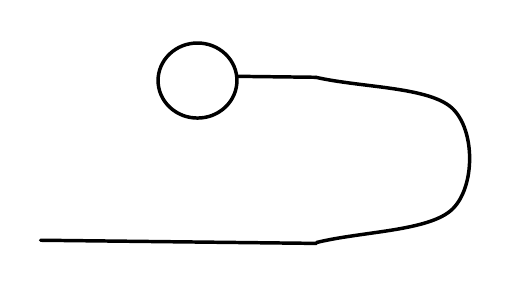}$}
$}
\end{equation*}
\end{multicols}
It is easy to prove that  $(\cgr[height=20pt]{Wcomult.pdf}, \cgr[height=20pt]{Wcounit.pdf})$ form a cocommutative comonoid. To distinguish it from the \emph{black} comonoid (from $\Sigma_{CW}$), we will often refer to it as to the \emph{white} comonoid. Observe that the set of inequations $I_{LCH}$ consists of \eqref{eq:lunitsl}-\eqref{eq:lbwbone} of lax bialgebras (Example \ref{ex:SMIT} (d)).
\item To the previous theory we can add inequalities \eqref{eq:olunitsl}-\eqref{eq:olbwbone} from the theory of oplax bialgebra (Example \ref{ex:SMIT}(e)). Since in all these equations, the left and the right hand side are Frobenius terms, the result is a Frobenius theory: we denote it by $\CMtheoryF = (\Sigma_{M}, E_M \uplus E_M^{op} \uplus B^{op})$.
% Observe that in the corresponding SMIT both $B$ and $B^{op}$ hold, that is the inequalities are actually equalities. For this reason, we say that the black comonoid and the white monoid form a bialgebra. 
Since both inequalities hold, the black and white structures form a bialgebra.
The corresponding SMIT is illustrated in Figure \ref{fig:cmonoid}.

In Section \ref{sec:frobcartesian}, we will see that this theory can be understood as the result of a more general construction and that, in a precise sense, it is equivalent to the cartesian theory of monoids (Example \ref{ex:cartesiantheory}(a)). For this reason, we will refer to this theory as \emph{the Frobenius theory of monoids}. In Section \ref{sec:monoid}, we will show that it provides a rich algebraic playground.
\end{enumerate}
\end{example}

We can define an appropriate category of models for Frobenius theories in a similar way to how it is defined for symmetric monoidal and cartesian theories, 

\begin{definition}\label{def:modelFrob}
Given a cartesian bicategory of relations $\catC$, 
a \emph{model} of a Frobenius theory $\Theory{T}$ in $\catC$ is a cartesian bifunctor $\funF \: \Frob{\Theory{T}} \to \catC$. 
A \emph{morphisms} of models $\funF\to \funG$ is a lax-natural monoidal transformation $\alpha\colon F \Rightarrow G$. This means that $\alpha$ is a family of $\catC$-morphisms  $\{\alpha_n\colon Fn \to Gn\}_{n\in \mathcal{N}}$ such that the following holds for all $f\colon n \to m$ in $\Frob{\Sigma, E}$
$$\xymatrix{
Fn \ar[dd]_{Ff} \ar[rr]^{\alpha_n} & & Gn \ar[dd]^{Gf}\\
\\
Fm\ar[rr]_{\alpha_m} \ar@{}[uurr]|{\leq} & & Gm
}$$
and $\alpha_{n + m} = \alpha_{n} \tns \alpha_{m}$ with $\alpha_0 = id_0$. 

The category of models of $\T$ in $\catC$ and their morphisms is denoted by $\frobmodel{\T,\catC}$. 
\end{definition}

Since a cartesian bifunctor is obliged to preserve lax products, it is forced to map the Frobenius structure of $\Frob{\Theory{T}}$ into the unique Frobenius structure of $\catC$ that determines the lax product. When $\catC = \Rel$, this means that any cartesian bifunctor $\funF \: \Frob{\Theory{T}} \to \Rel$ maps
$$\begin{array}{cc}
\Bcomult \mapsto \{( \,x, \,(x,x) \,) \text{ s.t. } x\in \funF1\} &
\Bcounit \mapsto \{  (x,\bullet) \text{ s.t. } x\in \funF1 \} 
 \\
\Bmult \mapsto \{( \,(x,x), \, x \,) \text{ s.t. } x\in \funF1\} & 
\Bunit \mapsto \{  (\bullet, x) \text{ s.t. } x\in \funF1 \}
\end{array}$$
where $\bullet$ is the unique element of the singleton set $\{\bullet\} = 1 = (F1)^0$.
Therefore, a model $F$ is determined by the object $F1$ and the arrows $Fo$ for all $o\in \Sigma$.
%Because of monoidality, a morphism of model $\alpha$ is uniquely determined by $\alpha_1$.
%
% The reason why, we relaxed naturality to lax naturality, will be clear by checking the several examples illustrated in the next section.
The implications of using lax natural transformations as model homomorphisms are explained in the next section.

\section{Examples of Frobenius Theories}\label{sec:ex}

In this section, we consider some examples of simple Frobenius theories and their models. We usually interpret the theory in the cartesian bicategory of relations $\Rel$.

\paragraph{The theory of sets.} We first answer the obvious question: what is the category of models for the empty Frobenius theory $(\emptyset, \emptyset)$?
The answer is at first sight surprising: this is just $\Set$ the category of sets and functions.
Indeed, a Cartesian bifunctor $F\colon \Frob{\emptyset , \emptyset}\to \Rel$ is uniquely determined by the object $F1$, which is just a set.

A morphisms of models $\alpha\colon F\Rightarrow G$ is determined by $\alpha_1\colon F1\to G1$ which is a relation satisfying the requirement that the following four squares laxly-commute.
$$\xymatrix{
F1 \ar[dd]_{\dup} \ar[rr]^{\alpha_1} & & G1 \ar[dd]^{\dup}\\
\\
F2\ar[rr]_{\alpha_1\tns \alpha_1 } \ar@{}[uurr]|{\leq} & & G2
}
\quad
\xymatrix{
F1 \ar[dd]_{\dis} \ar[rr]^{\alpha_1} & & G1 \ar[dd]^{\dis}\\
\\
F0\ar[rr]_{id_0 } \ar@{}[uurr]|{\leq} & & G0
}
\quad
\xymatrix{
F2 \ar[dd]_{\codup} \ar[rr]^{\alpha_1\tns \alpha_1} & & G1 \ar[dd]^{\codup}\\
\\
F1\ar[rr]_{\alpha_1 } \ar@{}[uurr]|{\leq} & & G1
}\quad
\xymatrix{
F0 \ar[dd]_{\codis} \ar[rr]^{ id_0 } & & G0 \ar[dd]^{\codis}\\
\\
F1\ar[rr]_{ \alpha_1} \ar@{}[uurr]|{\leq} & & G1
}
$$
The inequalities in the two rightmost squares hold for any relations. 
Instead the inclusion in the two leftmost squares holds if and only if the relation is a map, and maps in $\Rel$ coincide with functions.

\begin{remark}
Requiring morphisms of models to be strict natural transformations rather than just lax (as in Definition \ref{def:modelFrob}) would mean to force the four above inequalities to be equalities. In this case, a morphisms of model would be both a map and a comap, namely an isomorphism.
\end{remark}

\paragraph{The theory of non-empty sets.}
Let us now consider the Frobenius theory having empty signature and the following inequation.
\begin{equation}
\label{eq:ne}
\lower2pt\hbox{$
\lower5pt\hbox{$\includegraphics[height=.6cm]{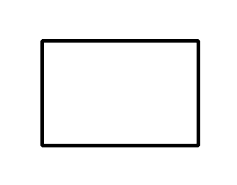}$}
\leq
\lower5pt\hbox{$\includegraphics[height=.6cm]{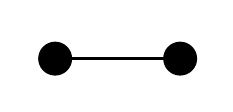}$}
$}
\end{equation}
Observe that the reverse of~\eqref{eq:ne} is~\eqref{eq:adj3}. Here, therefore, the so-called \emph{bone} equation holds:
\begin{equation*}
\lower2pt\hbox{$
\lower5pt\hbox{$\includegraphics[height=.6cm]{graffles/bbone.pdf}$}
=
id_0
$}
\end{equation*}

The corresponding SMIT has been studied in \cite{bruni2003some,coya2016corelations,dovsen2013syntax,zanasi2016algebra}. In these works it is proven that the resulting prop is isomorphic -- forgetting the posetal structure -- to the prop of equivalence relations (where a morphism $n\to m$ is an equivalence relation on $n+m$ regarded as a set).

\smallskip
From our perspective, this theory has quite a different meaning. 
Its models are sets that contain at least one element: indeed any Cartesian bifunctor $F$ to $\Rel$ maps $0$ to $(F1)^0$, i.e., the singleton set $\{\bullet\}$, the left-hand-side of \eqref{eq:ne} to $id_{\{\bullet\}} = \{(\bullet, \bullet)\}$ and the right-hand-side to the relational composition of $\{(\bullet, x) \text{ s.t. } x\in F1  \}$ and $\{(x, \bullet) \text{ s.t. } x\in F1\}$. 

Morphisms are functions, for the same reasons as in the previous example.

\paragraph{The theory of predicates.}
We start adding symbols to the signature: consider the theory containing an operation $\Wunit \: 0 \to 1$ and no equations. 
%As explained in [INSERTREF], in the corresponding SMIT, we have the following two inequations.
% \begin{multicols}{2}
%\noindent
%\begin{equation*}
%\lower4pt\hbox{$\includegraphics[height=.5cm]{graffles/unitsl.pdf}$}
%\geq
%\lower4pt\hbox{$\includegraphics[width=15pt]{graffles/idzerocircuit.pdf}$}
%\end{equation*}
%\begin{equation*}
%\lower5pt\hbox{$\includegraphics[height=.6cm]{graffles/runitsr.pdf}$}
%\geq
%\lower5pt\hbox{$\includegraphics[height=.6cm]{graffles/runitsl.pdf}$}
%\end{equation*}
%\end{multicols}
A model for this theory consists of a set $F1$ and a predicate $F(\Wunit) \subseteq F1$. A morphism of models $\alpha\colon F\Rightarrow G$ is uniquely determined by $\alpha_1\colon F1\to G1$ which is a function %preserving the predicate.
(for the same reason discussed above) which additionally satisfies the following inequality.
$$
\xymatrix{
F0 \ar[dd]_{F(\Wunit)} \ar[rr]^{ id_0 } & & G0 \ar[dd]^{G(\Wunit)}\\
\\
F1\ar[rr]_{ \alpha_1} \ar@{}[uurr]|{\subseteq} & & G1
}
$$
To make it more explicit the category of models for this theory is the category of predicates and predicate-preserving functions.

\paragraph{The theory of pointed sets.}
We can extend the theory of predicates by requiring $\Wunit$ to be total and single valued. That is we impose  inequations  \eqref{eq:olunitsr} and \eqref{eq:olbwbone}.

A model for this theory is a set $F1$ with a function $F(\Wunit)\colon F0 \to F1$. Since $F0$ is the singleton set $\{\bullet\}$, this is a pointed set. A morphism of models is a function preserving the point: in the following diagram
$$\xymatrix{
F0 \ar[dd]_{F\cgr[height=13pt]{Wunit.pdf}} \ar[rr]^{ id_0 } & & G0 \ar[dd]^{G\cgr[height=13pt]{Wunit.pdf}}\\
\\
F1\ar[rr]_{ \alpha_1} \ar@{}[uurr]|{\leq} & & G1
}
$$
all arrows are maps and therefore, by Corollary \ref{cor:mapssmaller}, they commutes strictly not just laxly.

Observe that the reverse of inequations  \eqref{eq:olunitsr} and \eqref{eq:olbwbone} are equations in the SMIT corresponding to this theory. Therefore we have that 
 \begin{multicols}{2}
\noindent
\begin{equation*}
\lower4pt\hbox{$\includegraphics[height=.5cm]{graffles/unitsl.pdf}$}
=
id_0
\end{equation*}
\begin{equation*}
\lower9pt\hbox{$\includegraphics[height=.9cm]{graffles/runitsr.pdf}$}
=
\lower8pt\hbox{$\includegraphics[height=.8cm]{graffles/runitsl.pdf}$}
\end{equation*}
\end{multicols}
With these equations, we can prove that every pointed set is non-empty, namely that inequation \eqref{eq:ne} holds:
\begin{equation*}
\lower2pt\hbox{$
id_0
=
\lower4pt\hbox{$\includegraphics[height=.5cm]{graffles/unitsl.pdf}$}
\leq
\lower5pt\hbox{$\includegraphics[height=.6cm]{graffles/bbone.pdf}$}
$}
\end{equation*}

Another simple graphical derivation proves that $\cgr[height=13pt]{Wunit.pdf}$ is injective
\begin{equation}\label{ineq:injective}
\lower2pt\hbox{$
\lower12pt\hbox{$\includegraphics[height=1cm]{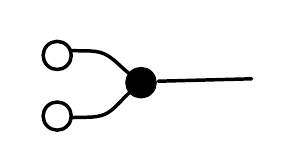}$}
=
\lower12pt\hbox{$\includegraphics[height=1cm]{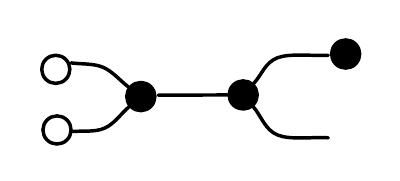}$}
\leq 
\lower12pt\hbox{$\includegraphics[height=1cm]{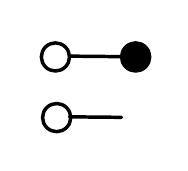}$}
=
\lower5pt\hbox{$\includegraphics[height=.6cm]{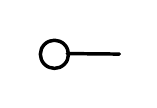}$}
$}
\end{equation}
and thus, by  Corollary~\ref{cor:cospancancellation}, the (white) bone law follows,
\begin{equation}
\label{eq:wbone}
\lower2pt\hbox{$
\lower5pt\hbox{$\includegraphics[height=.6cm]{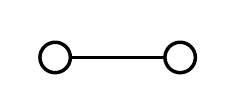}$}
=
id_0
$}
\end{equation}
where, as in Example \ref{ex:FT}(a), $\cgr[height=13pt]{Wcounit.pdf} \: 1 \to 0$ is the notation for $\cgr[height=13pt]{Wunit.pdf}^\dag$. 
%Graphically,
%\begin{equation*}
%\lower2pt\hbox{$
%\lower5pt\hbox{$\includegraphics[height=.6cm]{graffles/Wcounit.pdf}$}
%::=
%\lower5pt\hbox{$\includegraphics[height=.6cm]{graffles/wunitcc.pdf}$}
%$}
%\end{equation*}

The last property that we are going to show is the following.
\begin{equation}\label{eq:wubm}
\lower2pt\hbox{$
\lower12pt\hbox{$\includegraphics[height=1cm]{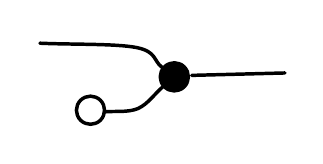}$}
=
\lower12pt\hbox{$\includegraphics[height=1cm]{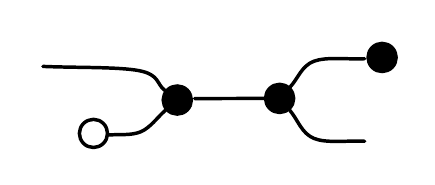}$}
=
\lower15pt\hbox{$\includegraphics[height=1.2cm]{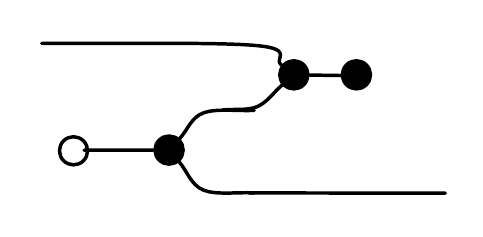}$}
=
\lower12pt\hbox{$\includegraphics[height=1cm]{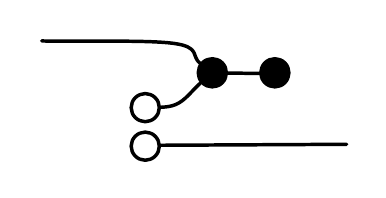}$}
=
\lower5pt\hbox{$\includegraphics[height=.6cm]{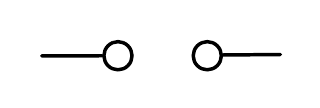}$}
$}
\end{equation}
The above equation proves that the Carboni-Walters category generated by the Frobenius theory of pointed sets is isomorphic (disregarding the posetal structure) to the prop of \emph{partial equivalence relations} studied in \cite{zanasi2016algebra}.

\paragraph{The theory of binary relations.}
We now consider the theory containing an operation $R \: 1 \to 1$ and no equations.
A model for this theory consists of a set $F1$ and a relation $FR \subseteq F1\times F1$. A morphism of models $\alpha\colon F\Rightarrow G$ is determined by $\alpha_1\colon F1\to G1$ which is a function that satisfies the following inequality:
$$
\xymatrix{
F0 \ar[dd]_{FR} \ar[rr]^{ \alpha_1 } & & G0 \ar[dd]^{GR}\\
\\
F1\ar[rr]_{ \alpha_1} \ar@{}[uurr]|{\leq} & & G1
}
$$
This simply means that the function $\alpha_1$ \emph{preserves} the relation $R$: if $(x,y)\in FR$, then $(\alpha_1(x),\alpha_1(y))\in GR$. The category of models is, therefore, the category of binary relations.

\paragraph{The theory of partial orders.}
The category of partial order and monotone maps can be obtained as the category of models of a Frobenius theory :
the signature consists of a single symbol $\cgr[height=15pt]{lesseq.pdf} \: 1 \to 1$ representing a relation and the following three inequations impose reflexivity, transitivity and antisymmetry.
\begin{multicols}{3}
\noindent
\begin{equation*}
\lower8pt\hbox{$
%\tag{Ref}
\lower3pt\hbox{$\includegraphics[height=.4cm]{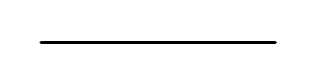}$}
\leq
\lower8pt\hbox{$\includegraphics[height=.8cm]{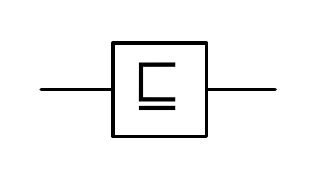}$}
$}
\end{equation*}
\begin{equation*}
\label{eq:unitsr}
%\tag{Trans}
\lower8pt\hbox{$
\lower8pt\hbox{$\includegraphics[height=.8cm]{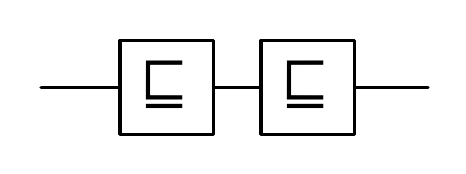}$}
\leq
\lower8pt\hbox{$\includegraphics[height=.8cm]{graffles/lesseq.pdf}$}
$}
\end{equation*}
\begin{equation*}
\label{eq:antisim}
%\tag{Antisim}
\lower16pt\hbox{$\includegraphics[height=1.4cm]{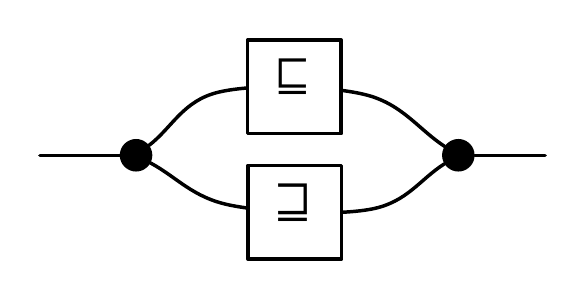}$}
\leq
\lower3pt\hbox{$\includegraphics[height=.4cm]{graffles/identity.pdf}$}
\end{equation*}
\end{multicols}

In the rightmost inequation, $\cgr[height=15pt]{greatereq.pdf} \: 1 \to 1$ is the inverse relation of $\cgr[height=13pt]{lesseq.pdf}$ formally defined as $\cgr[height=15pt]{lesseq.pdf}^\dag$. Graphically,
\begin{equation*}
\lower2pt\hbox{$
\lower5pt\hbox{$\includegraphics[height=.6cm]{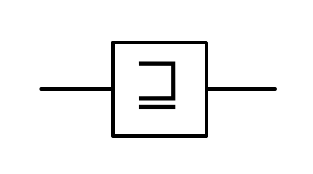}$}
::=
\lower15pt\hbox{$\includegraphics[height=1.2cm]{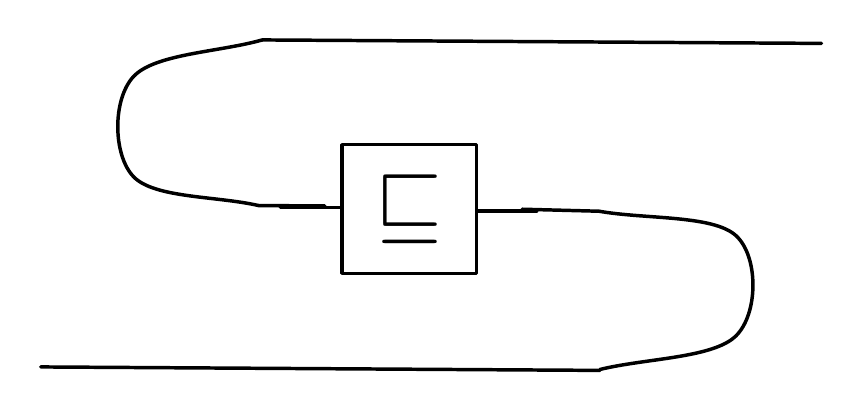}$}
$}
\end{equation*}
It is easy to see that a model for this theory is exactly a set equipped with a relation that is a partial order.
By spelling out the definitions of morphism of model, the same arguments of the previous example applies: a morphism is just a function that preserves the relation, namely a monotone function. The category of models is thus the usual category of posets and monotone functions.

\begin{remark}
Modifying the above, one obtains many familiar theories. For instance, by replacing antisymmetry (the rightmost inequations) by symmetry ($\cgr[height=15pt]{lesseq.pdf} \leq \cgr[height=15pt]{greatereq.pdf}$ ) one obtains equivalence relations. By dropping it, one gets pre-orders. Now replacing the underlying category $\Rel$, by the category of relations over a quantale, one obtains generalized metric spaces in the sense of Lawvere \cite{lawvere1973metric}.
\end{remark}

\paragraph{The theory of deterministic automata.}
A deterministic automaton (DA) on a finite alphabet $A$ consists of a triple $(X,t,f)$, where $X$ is a set of states, $t\colon X\to X^A$ is the transition function and $f\subseteq X$ is its set of final states. Given DAs $(X_1,t_1,f_1)$ and $(X_2,t_2,f_2)$, a function $h \colon X_1 \to X_2$ is a \emph{simulation} if (1) $h(f_1)\subseteq f_2$ and (2), for all $x\in X$, and $a\in A$, $h(t_1(x)(a)) = t_2(h(x)(a))$.

\medskip

We now introduce a Frobenius theory $\Theory{DA}$ for deterministic automata and simulations. 
The signature consists of two symbols $f\colon 1\to 0$ and $t\colon 1 \to |A|$. The set of inequations only requires that $t$ is a map: $\dup \poi (t\tns t) \leq t \poi \dup_{|A|}$ and $\dis \leq t \poi \dis_{|A|}$.

Then models are exactly deterministic automata. Indeed a model $F\colon \Frob{\Theory{DA}}\to \Rel$ consists of a set $F1$ (representing the states) a function  $Ft \colon F1 \to (F1)^{|A|}$ (the transition) and a relation $Ff \colon F1 \to 1 = (F1)^0$, namely a predicate (the final states).

Morphisms are exactly simulations. They are functions for the usual reason and the following lax naturality squares
$$
\xymatrix{
F1 \ar[dd]_{Ff} \ar[rr]^{\alpha_1} & & G1 \ar[dd]^{Gf}\\
\\
F0\ar[rr]_{id_0 } \ar@{}[uurr]|{\leq} & & G0
}
\quad
\xymatrix{
F1 \ar[dd]_{Ft} \ar[rr]^{\alpha_1} & & G1 \ar[dd]^{Gt}\\
\\
F|A|\ar[rr]_{\alpha_{|A|} } \ar@{}[uurr]|{\leq} & & G|A|
}
$$
impose, respectively, conditions (1) and (2) of the definition of simulation. 

\begin{remark}
This example suggests that one may use Frobenius theory to express coalgebraic structures. Indeed, one can encode a function $f\colon X\to X^n$ either by having in the signature a symbol $f\colon1\to n$ and imposing the axioms for functions or as a symbol $f\colon n\to 1$ and impose the axioms of cofunctions. Then the lax-naturality conditions make the morphism the expected coalgebra homomorphisms.

Unfortunately, one can express only coalgebras for functors of the shapes $F(X)=X^n$, which are usually not particularly interesting: the final coalgebras is always the one element set $1$. More interesting functors, such as $F(X)=X \times X +1$ (finite and infinite trees),  cannot be expressed in our framework: intuitively we are missing the coproduct $+$.
\end{remark}

\paragraph{The theory of non-negative monoids.}\label{ssec:non-neg}
Recall the Frobenius theory of monoids $\CMtheoryF$ from Example \ref{ex:FT} (b). It is now easy to see that its models in $\Rel$ are simply ordinary monoids in $\Set$, and morphisms of models are just monoid homomorphisms.

By adding to $\CMtheoryF$ the following inequation
\begin{equation}
\label{eq:nneg}
\lower2pt\hbox{$
\lower7pt\hbox{$\includegraphics[height=.8cm]{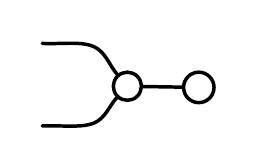}$}
\leq
\lower7pt\hbox{$\includegraphics[height=.8cm]{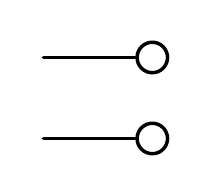}$}
$}
\end{equation}
we obtain the theory of \emph{non-negative monoids}. These are monoids $(X,+,0)$ with the additional property that $\forall x,y \in X$, if $x+y= 0$, then $x=0$ and $y=0$.
Note that the string diagrams on both sides of \eqref{eq:nneg} are Frobenius terms: the definition of the white counit requires the Frobenius structure.

\section{From cartesian theories to Frobenius theories}\label{sec:frobcartesian}
In Example \ref{ex:cartesiantheory}(a) we introduced the cartesian theory of commutative monoids $\CMtheory$. In Example \ref{ex:FT}(b) we showed that the Frobenius theory of commutative monoid $\CMtheoryF$ is obtained from $\CMtheory$ by adding inequality of oplax bialgebras, that force the (white) monoid structure to be a oplax (black) comonoid homomorphism. Since in any Frobenius theory, the generators are implicitly forced to be lax comonoid homomorphism (inequations \eqref{eq:lcomhom1} and \eqref{eq:lcomhom2}), one has that in the SMIT corresponding to $\CMtheoryF$, as in the SMT corresponding to $\CMtheory$, the white monoid is a \emph{strict} black comonoid homomorphism. 

\medskip

This construction can be generalised as follows: given any cartesian theory $\Theory{T} = (\Sigma, E)$ one builds the Frobenius theory $\FTheory{T}=(\Sigma, E\uplus E^{op} \uplus I_{OLCH} )$ where $I_{OLCH}$ consists of the following two inequalities
%\begin{multicols}{2}
%\noindent
\begin{equation}
\label{eq:cf1}
\lower2pt\hbox{$
\lower12pt\hbox{$\includegraphics[height=1cm]{graffles/law2.pdf}$}
\leq
\lower12pt\hbox{$\includegraphics[height=1cm]{graffles/law1.pdf}$}
$}
\end{equation}
\begin{equation}
\label{eq:cf2}
\lower20pt\hbox{$\includegraphics[height=1.5cm]{graffles/law4.pdf}$}
\leq
\lower12pt\hbox{$\includegraphics[height=1cm]{graffles/law3.pdf}$}
\end{equation}
%\end{multicols}
for each generator $\cgr[height=13pt]{operator.pdf} \: n \to 1$ in the signature $\Sigma$.
The two theories are closely related: indeed, the category of cartesian models for $\Theory{T}$ in $\Set$ is isomorphic to the category of Frobenius models for  $\FTheory{T}$ in $\Rel$.

\medskip

One can further generalize by replacing $\Rel$ with an arbitrary cartesian bicategory of relations $\catC$ and $\Set$ by its \emph{category of maps} $\Map{\catC}$. This is the defined to be the sub-bicategory of $\catC$ having as arrows the maps of $\catC$. Indeed, $\catC$ has finite products.
\begin{proposition}
Let $\catC$ be a cartesian bicategory of relations. Then $\Map{\catC}$ is a cartesian category.
\end{proposition}
\begin{proof}
See Theorem 6.1(i) of \cite{Carboni1987}.
\end{proof}

We can now state the main result of this section.
\begin{theorem}
Let $\Theory{T} = (\Sigma, E)$ be a cartesian theory and  $\FTheory{T}=(\Sigma, E\uplus E^{op} \uplus I_{OLCH} )$. Then $\frobmodel{\FTheory{T},\catC} \cong  \carmodel{\Theory{T},\Map{\catC}}$. 
\end{theorem}
\begin{proof}
This theorem can be proved in several ways. We give a concrete proof that uses the intuitions developed so far for models and their morphisms.

A cartesian bifunctor $F\colon \Frob{\FTheory{T}} \to \catC$ is uniquely determined by: (a) the object $F1$ in $\catC$ and (b) arrows $Fo$ in $\catC$ for any generator $o\in \Sigma$. Inequations \eqref{eq:cf1} and \eqref{eq:cf2} force each of the $Fo$ to be a map, namely an arrow in $\Map{\catC}$. The same data uniquely determine a cartesian functor $F\colon \Law{\Theory{T}} \to \Map{\catC}$.

For morphisms, observe that $\alpha\colon F\to G$ of Frobenius models is uniquely determined by an arrow of $\catC$ $\alpha_1 \colon F1\to G1$ that is a map (recall from Section \ref{sec:ex} the example of the theory of sets). The morphism $\alpha$ is additionally required to be a lax natural transformation, but since all the $Fo$ are maps then, by Corollary \ref{cor:mapssmaller},  this just means to be a natural transformation. This data uniquely induces a morphism of cartesian models.
\end{proof}

Thus Frobenius theories are at least as expressive as cartesian theories. The examples of non-empty sets, predicates, binary relations, partial orders, DAs and non-negative monoids in Section \ref{sec:ex} show that Frobenius theories are strictly more expressive than classical (cartesian) algebraic theories.
In the remainder of this paper, we explore several Frobenius theories arising from well-known cartesian theories.

\section{The theory of commutative monoids }\label{sec:monoid}
In this section we explore the algebraic properties of the Frobenius theory of commutative monoids $\CMtheoryF$ introduced in Example \ref{ex:FT} (b). Its SMIT is illustrated in Figure \ref{fig:cmonoid}.

\begin{figure}
\begin{center}
\includegraphics[height=7cm]{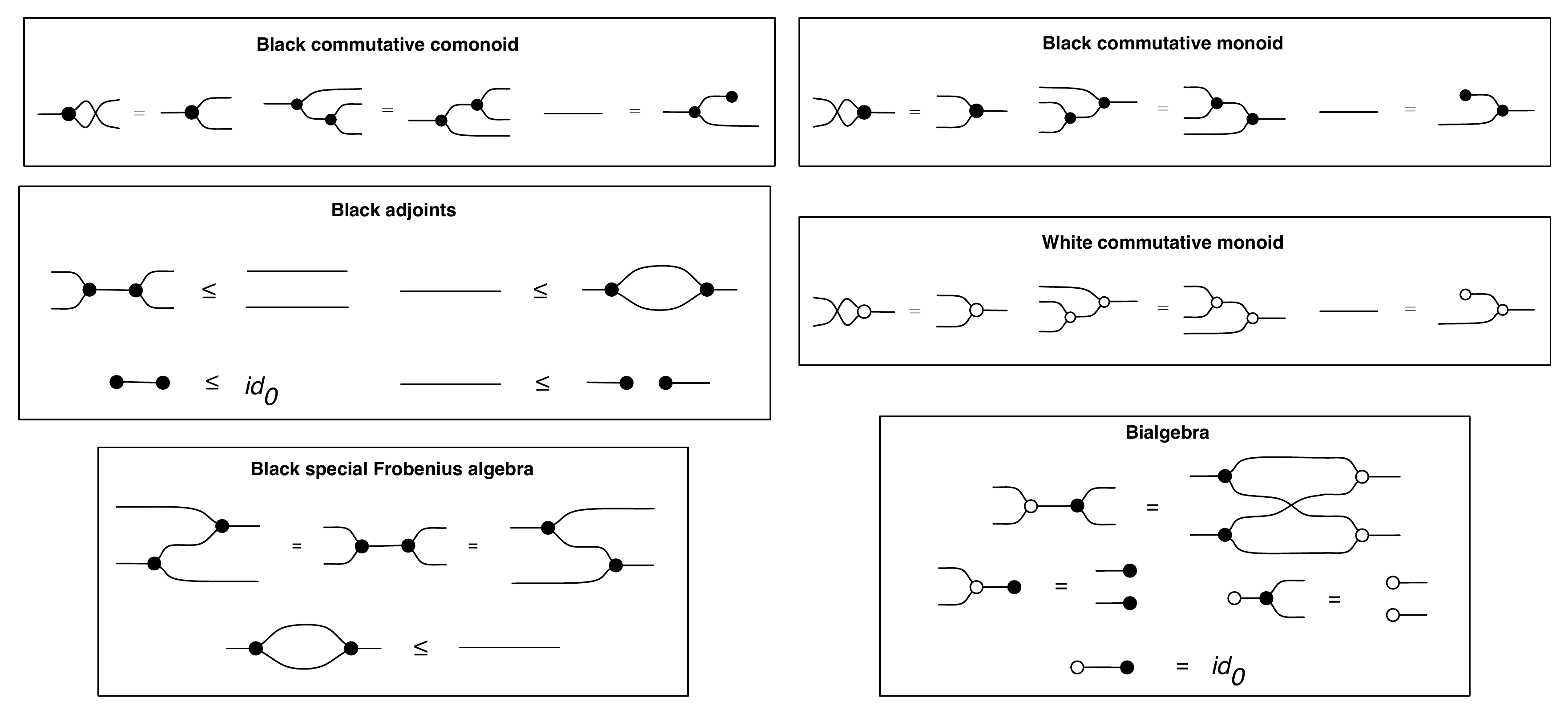}
\end{center}
\caption{The SMIT corresponding to the Frobenius theory of commutative monoids}\label{fig:cmonoid}
\end{figure}

Observe that the white monoid and the black comonoid form a \emph{bialgebra} structure (Example \ref{ex:equationalprops}(d)). This has two interesting consequences; First, by Lemma \ref{lemma:dagger}(iv), we have that also the black monoid and the white comonoid give rise to a bialgebra. Second, we have that both $\cgr[height=15pt]{Wmult.pdf}$ and  $\cgr[height=15pt]{Wunit.pdf}$ are total and single valued. By Lemma \ref{lem:characterizationmap}, we have the following inequalities.
\begin{multicols}{4}
\noindent
\begin{equation*}
\label{eq:nameless}
\lower2pt\hbox{$
\lower7pt\hbox{$\includegraphics[height=.8cm]{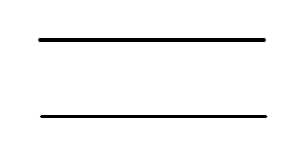}$}
\leq
\lower7pt\hbox{$\includegraphics[height=.8cm]{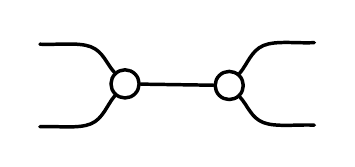}$}
$}
\end{equation*}
\begin{equation*}
\label{eq:nameless}
\lower2pt\hbox{$
\lower7pt\hbox{$\includegraphics[height=.8cm]{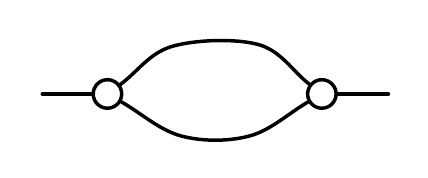}$}
\leq
\lower3pt\hbox{$\includegraphics[height=.4cm]{graffles/identity.pdf}$}
$}
\end{equation*}
\begin{equation*}
\label{eq:nameless}
\lower2pt\hbox{$
id_0
\leq
\lower5pt\hbox{$\includegraphics[height=.6cm]{graffles/wbone.pdf}$}
$}
\end{equation*}
\begin{equation*}
\label{eq:nameless}
\lower2pt\hbox{$
\lower5pt\hbox{$\includegraphics[height=.6cm]{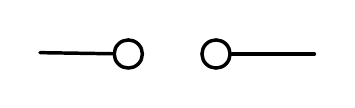}$}
\leq
\lower3pt\hbox{$\includegraphics[height=.4cm]{graffles/identity.pdf}$}
$}
\end{equation*}
\end{multicols}
By using Lemma \ref{lemma:top}, it is easy to see that $\cgr[height=13pt]{Wmult.pdf}$ is surjective
\begin{equation}\label{ineq:surjective}
\lower2pt\hbox{$
\lower12pt\hbox{$\includegraphics[height=1cm]{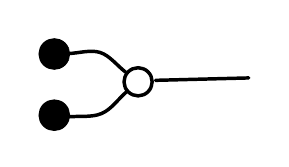}$}
\geq\lower12pt\hbox{$\includegraphics[height=1cm]{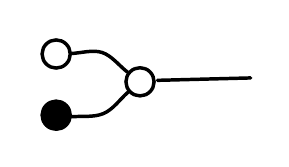}$}
= 
\lower5pt\hbox{$\includegraphics[height=.6cm]{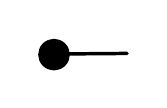}$}
$}
\end{equation}
and thus, by  Corollary~\ref{cor:spancancellation}, the \emph{white special law} holds.
\begin{equation}
\lower12pt\hbox{$\includegraphics[height=1.2cm]{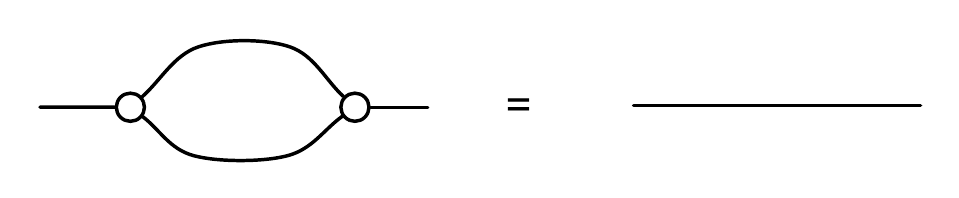}$}
\end{equation}

As in the theory of pointed sets, $\cgr[height=13pt]{Wunit.pdf}$ is injective
\begin{equation}\label{ineq:injective}
\lower2pt\hbox{$
\lower12pt\hbox{$\includegraphics[height=1cm]{graffles/inj1.pdf}$}
=
\lower12pt\hbox{$\includegraphics[height=1cm]{graffles/inj2.pdf}$}
\leq 
\lower12pt\hbox{$\includegraphics[height=1cm]{graffles/inj3.pdf}$}
=
\lower5pt\hbox{$\includegraphics[height=.6cm]{graffles/inj4.pdf}$}
$}
\end{equation}
and thus, by  Corollary~\ref{cor:cospancancellation}, the so-called \emph{white bone law} holds.
\begin{equation*}
\lower2pt\hbox{$
\lower5pt\hbox{$\includegraphics[height=.6cm]{graffles/wbone.pdf}$}
=
id_0
$}
\end{equation*}

Observe that the inverse inclusions for \eqref{ineq:surjective} and \eqref{ineq:injective} hold by \eqref{eq:inverselaxhom1} and \eqref{eq:inverselaxhom2}. Therefore we have that
\begin{multicols}{2}
\noindent
\begin{equation*}
\lower2pt\hbox{$
\lower12pt\hbox{$\includegraphics[height=1cm]{graffles/bubuwm1.pdf}$}
= 
\lower5pt\hbox{$\includegraphics[height=.6cm]{graffles/bubuwm3.pdf}$}
$}
\end{equation*}
\begin{equation*}
\lower2pt\hbox{$
\lower12pt\hbox{$\includegraphics[height=1cm]{graffles/inj1.pdf}$}
=
\lower5pt\hbox{$\includegraphics[height=.6cm]{graffles/inj4.pdf}$}
$}
\end{equation*}
\end{multicols}
These two equations can be weakened at some extent. For the rightmost, we have that:
\begin{equation}\label{eq:wubm}
\lower2pt\hbox{$
\lower12pt\hbox{$\includegraphics[height=1cm]{graffles/wubm.pdf}$}
=
\lower12pt\hbox{$\includegraphics[height=1cm]{graffles/wubm1.pdf}$}
=
\lower15pt\hbox{$\includegraphics[height=1.2cm]{graffles/wubm2.pdf}$}
=
\lower12pt\hbox{$\includegraphics[height=1cm]{graffles/wubm3.pdf}$}
=
\lower5pt\hbox{$\includegraphics[height=.6cm]{graffles/wcwu.pdf}$}
$}
\end{equation}
The analogous of the leftmost only holds laxly by Lemma \ref{lemma:top}.
\begin{equation}
\label{eq:buwm}
\lower2pt\hbox{$
\lower12pt\hbox{$\includegraphics[height=1cm]{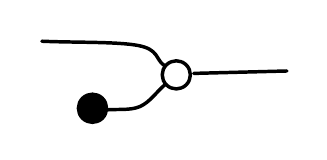}$}
\leq
\lower5pt\hbox{$\includegraphics[height=.6cm]{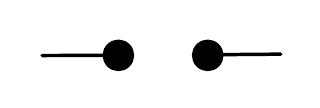}$}
$}
\end{equation}
The other inclusion would mean that for all $x,z\in X$, there exists $y$ such that $x+y=z$. Later we will show that this holds for abelian groups.

\medskip

Another interesting law which can be derived in the theory of commutative monoid is the \emph{black bone}:
\begin{equation}
\label{eq:bbone}
\lower2pt\hbox{$
\lower5pt\hbox{$\includegraphics[height=.6cm]{graffles/bbone.pdf}$}
=
id_0
$}
\end{equation}
Indeed, the left-to-right inequation is just \eqref{eq:adj3}. For the right-to-left, observe that:
\begin{equation*}
\lower2pt\hbox{$
id_0
=
\lower5pt\hbox{$\includegraphics[height=.6cm]{graffles/unitsl.pdf}$}
\leq
\lower5pt\hbox{$\includegraphics[height=.6cm]{graffles/bbone.pdf}$}
$}
\end{equation*}
Recall that \eqref{eq:bbone} does not hold in the theory of sets: the empty set does not satisfy it. Intuitively, \eqref{eq:bbone} states that every commutative monoid has a non empty carrier set; indeed, it must contain at least a unit element.

\bigskip
The pair of white-monoid white-comonoid forms a \emph{lax Frobenius} structure. This is proved as follows:
\begin{equation*}
\lower2pt\hbox{$
\lower15pt\hbox{$\includegraphics[height=1.2cm]{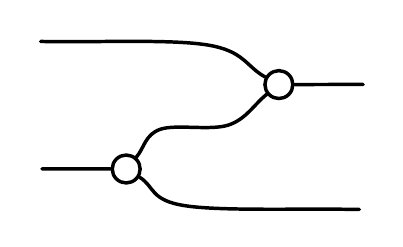}$}
\leq 
\lower15pt\hbox{$\includegraphics[height=1.2cm]{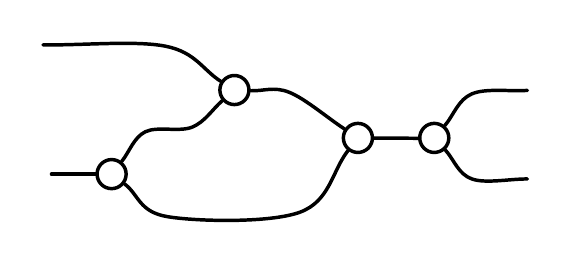}$}
=
\lower15pt\hbox{$\includegraphics[height=1.2cm]{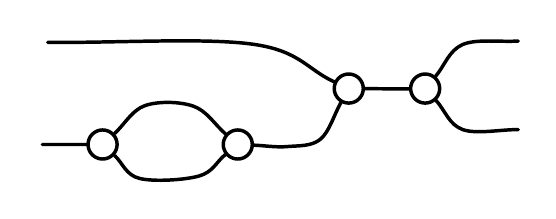}$}
=
\lower12pt\hbox{$\includegraphics[height=1cm]{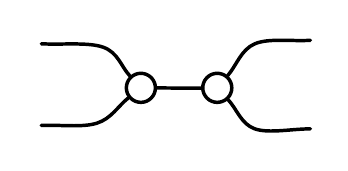}$}
$}
\end{equation*}
The following inclusion follows by Lemma \ref{lemma:dagger} (iv).
\begin{equation*}
\lower2pt\hbox{$
\lower15pt\hbox{$\includegraphics[height=1.2cm]{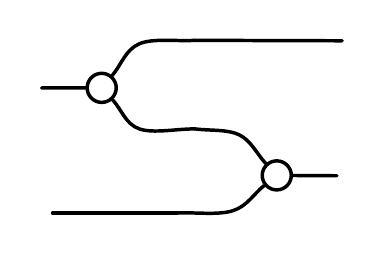}$}
\leq 
\lower12pt\hbox{$\includegraphics[height=1cm]{graffles/laxwfrob4.pdf}$}
$}
\end{equation*}

%\marginpar{[CAN WE HAVE A LAX-SPIDER THEOREM? IT WOULD BE EXTREMELY USEFUL...]}

\medskip

The pair of white monoid-comonoid forms a \emph{lax bialgebra}. We have seen that \eqref{eq:wbone} holds. The other laws only hold laxly.
\begin{equation*}
\lower2pt\hbox{$
\lower12pt\hbox{$\includegraphics[height=1cm]{graffles/nonneg2.pdf}$}
\leq 
\lower12pt\hbox{$\includegraphics[height=1cm]{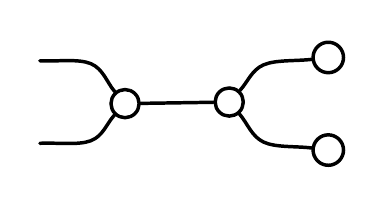}$}
=
\lower12pt\hbox{$\includegraphics[height=1cm]{graffles/nonneg1.pdf}$}
$}
\end{equation*}
Note that the other direction does not hold in general for monoid: it would mean that every monoid is non-negative (see Section \ref{sec:ex}).
The remaining rule can be proved by using the lax Frobenius as follows. 
\begin{equation*}
\lower2pt\hbox{$
\lower15pt\hbox{$\includegraphics[height=1.2cm]{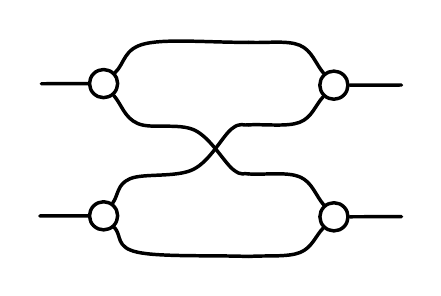}$}
=
\lower15pt\hbox{$\includegraphics[height=1.2cm]{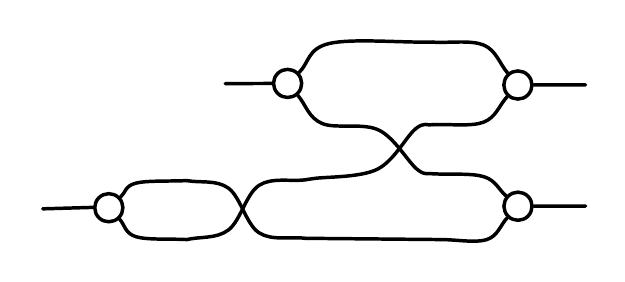}$}
=
\lower15pt\hbox{$\includegraphics[height=1.2cm]{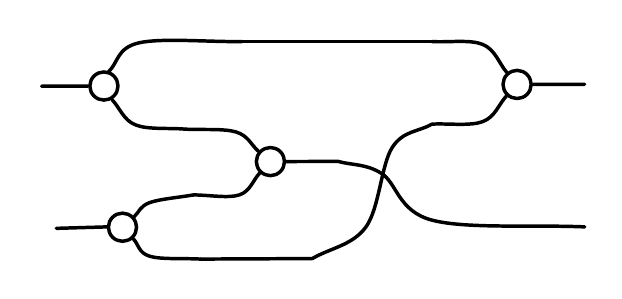}$}
\leq
\lower15pt\hbox{$\includegraphics[height=1.2cm]{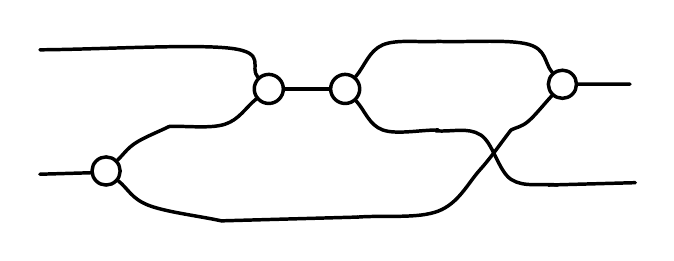}$}
\leq
\lower15pt\hbox{$\includegraphics[height=1.2cm]{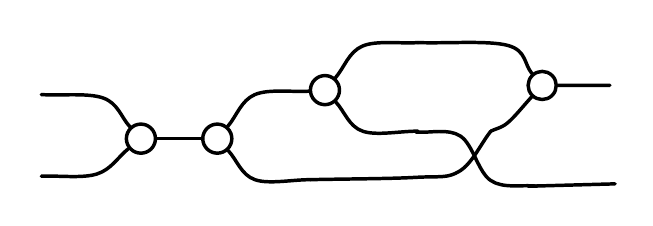}$}$}
\end{equation*}
\begin{equation*}
\lower2pt\hbox{$
=
\lower15pt\hbox{$\includegraphics[height=1.2cm]{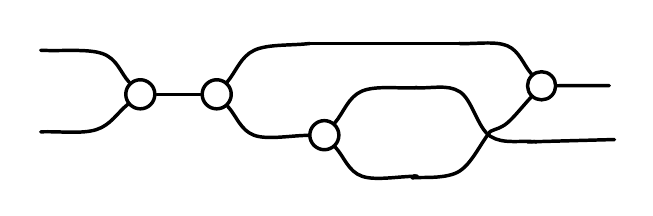}$}
=
\lower15pt\hbox{$\includegraphics[height=1.2cm]{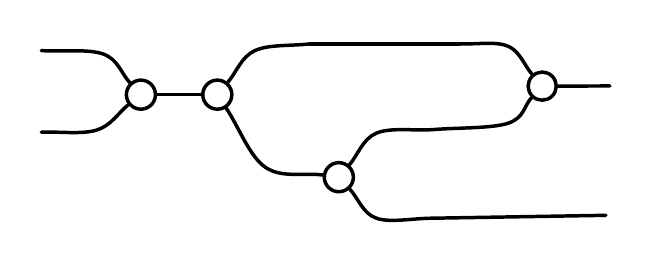}$}
\leq
\lower12pt\hbox{$\includegraphics[height=1cm]{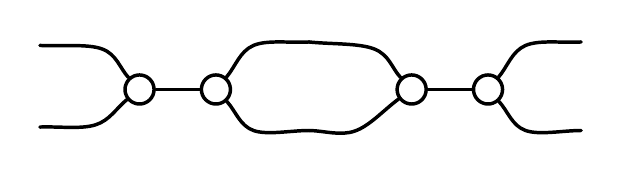}$}
=
\lower12pt\hbox{$\includegraphics[height=1cm]{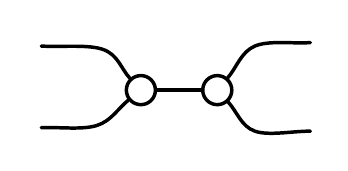}$}
$}
\end{equation*}
%\marginpar{[THIS WOULD BE NICER TO BE SHOWN WITH A GENERALIZED LAX SPIDER THEOREM]}

Again, the inverse inequation does not hold in general for monoids. We will see in Section \ref{sec:ag} that it holds for abelian groups. %and not for Join-Semilattices. (For a counterexample, take a set with four incomparable elements $a,b,c,d$, plus $\top$ and $\bot$. This is clearly a join-semilattice. We have that $a+b = \top =c+d$, so $\{a,b,c,d\}$ belongs to the rightmost diagram but not to the leftmost).

\medskip

The following laws state that the black multiplication laxly distributes over the white one.
\begin{multicols}{2}
\noindent
\begin{equation}\label{ineq:bwmlaxdist}
\lower2pt\hbox{$
\lower20pt\hbox{$\includegraphics[height=1.5cm]{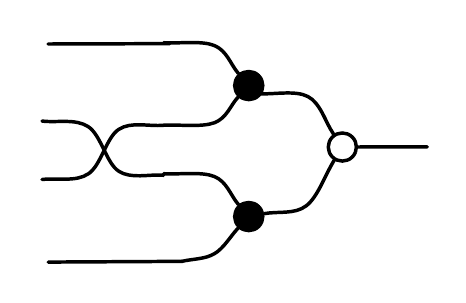}$}
\leq
\lower20pt\hbox{$\includegraphics[height=1.5cm]{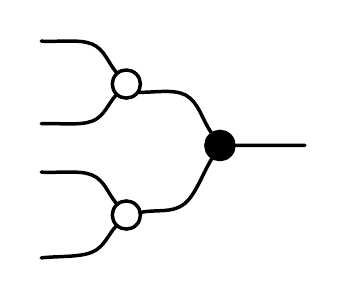}$}
$}
\end{equation}
\begin{equation}\label{ineq:bwmlaxdist}
\lower2pt\hbox{$
\lower20pt\hbox{$\includegraphics[height=1.5cm]{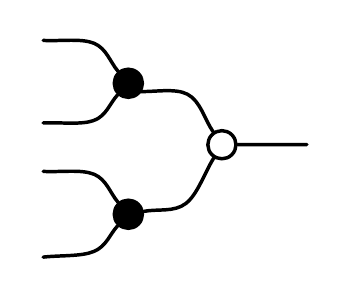}$}
\leq
\lower20pt\hbox{$\includegraphics[height=1.5cm]{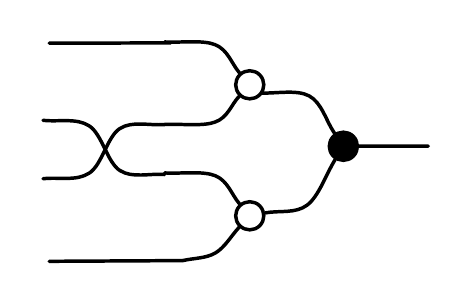}$}
$}
\end{equation}
\end{multicols}
Observe that the leftmost rule is an instance of \eqref{eq:inverselaxhom2}, while the rightmost can be obtained simply by precomposing with a permutation. It is easy to see that the inverse implication does not hold.

\medskip

Another law, which is useful in several occasions, is the following 
\begin{equation}
\label{eq:weird}
\lower2pt\hbox{$
\lower15pt\hbox{$\includegraphics[height=1.2cm]{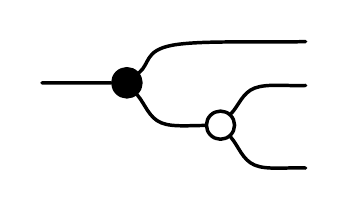}$}
=
\lower20pt\hbox{$\includegraphics[height=1.5cm]{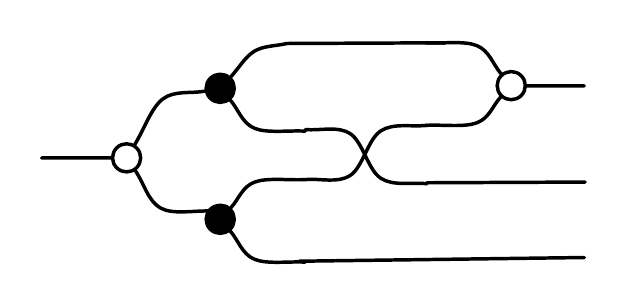}$}
$}
\end{equation}
which can be derived as shown below:
\begin{multline*}
\lower15pt\hbox{$\includegraphics[height=1.2cm]{graffles/weirdlawL.pdf}$}
=
\lower15pt\hbox{$\includegraphics[height=1.2cm]{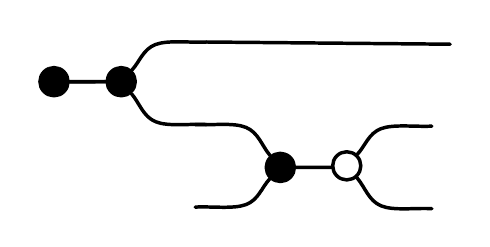}$}
=
\lower20pt\hbox{$\includegraphics[height=1.5cm]{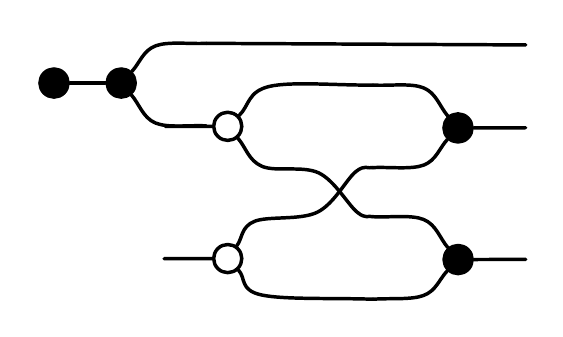}$}
\\
=
\lower28pt\hbox{$\includegraphics[height=2.2cm]{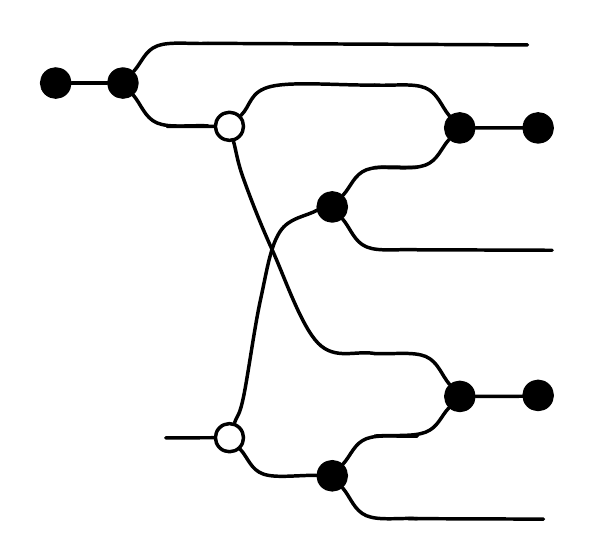}$}
=
\lower28pt\hbox{$\includegraphics[height=2.2cm]{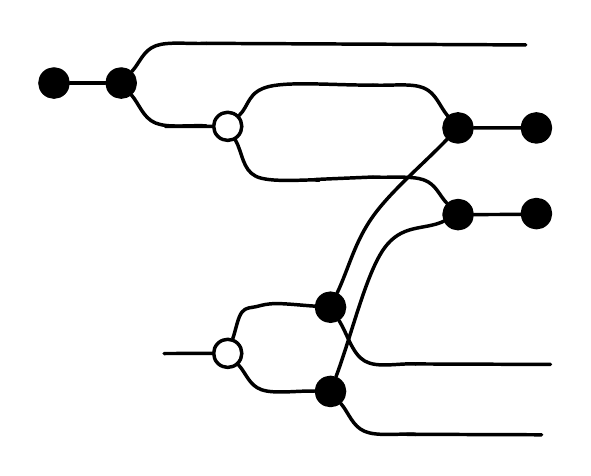}$}
=
\lower20pt\hbox{$\includegraphics[height=1.5cm]{graffles/weirdlawR.pdf}$}
\end{multline*}

\paragraph{Additive arrows.} We now study a general properties of the arrows in $\Frob{\CMtheoryF}$.
An arrow $R\colon m\to n$ is said to be \emph{additive} if the two inequations below hold. %It is said to \emph{have zero} if the rightmost inequation holds.

\begin{multicols}{2}
\noindent
\begin{equation}
\lower2pt\hbox{$
%\tag{R-Add}
\lower17pt\hbox{$\includegraphics[height=1.4cm]{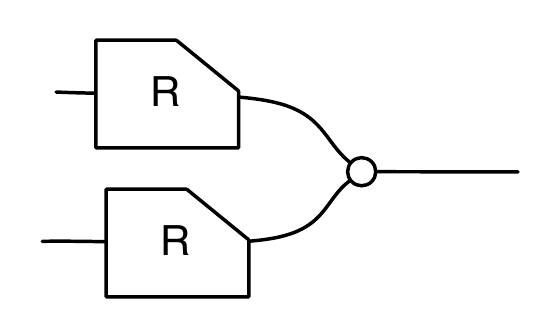}$}
\leq
\lower12pt\hbox{$\includegraphics[height=1cm]{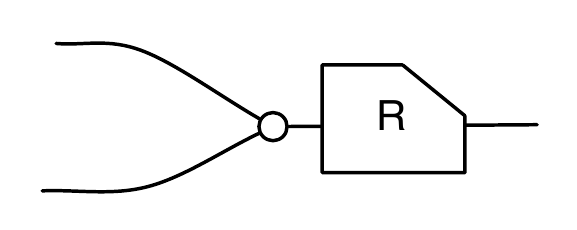}$}
$}
\end{equation}
\begin{equation}
\label{eq:unitsr}
%\tag{R-zero}
\lower5pt\hbox{$\includegraphics[height=.6cm]{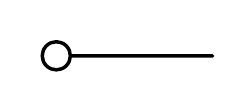}$}
\leq
\lower10pt\hbox{$\includegraphics[height=1cm]{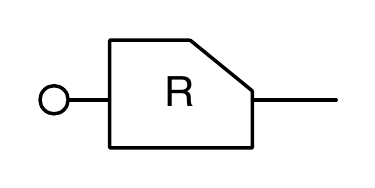}$}
\end{equation}
\end{multicols}
In $\Rel$, a relation $R$ is additive when is closed w.r.t. the structure of a monoid $(X,+,0)$, that is: (a) if $(x_1,y_1)\in R$ and $(x_2,y_2)\in R$, then $(x_1+x_2, y_1+y_2)\in R$ and (b) $(0,0)\in R$ for all $x_1,x_2,y_1,y_2\in X$.

An additive arrow is, therefore, an oplax white monoid homomorphism.
The fact that is also an oplax white comonoid homomorphism easily follows from the fact that white multiplication and unit are maps.
\begin{equation*}
\lower15pt\hbox{$\includegraphics[height=1.2cm]{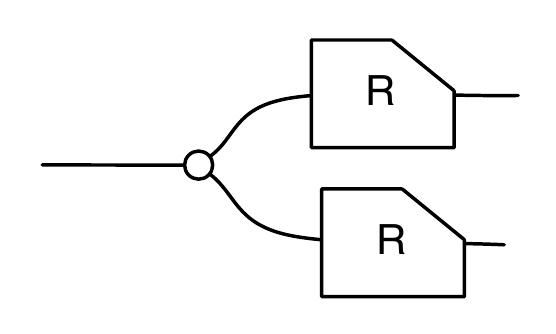}$}
\ \leq\ 
\lower15pt\hbox{$\includegraphics[height=1.2cm]{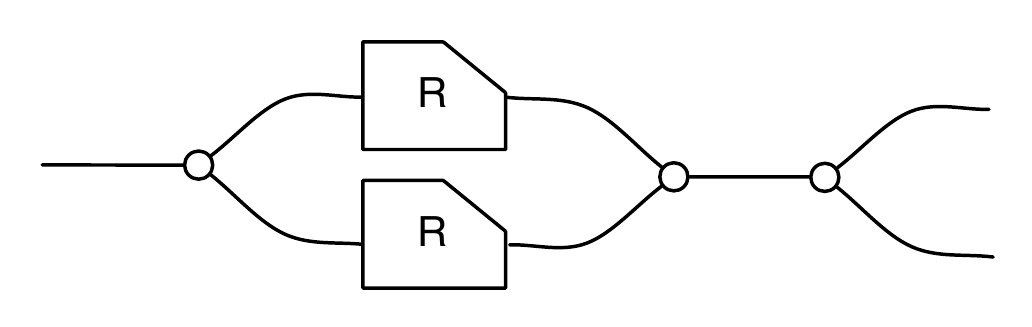}$}
\ \leq\ 
\lower12pt\hbox{$\includegraphics[height=1cm]{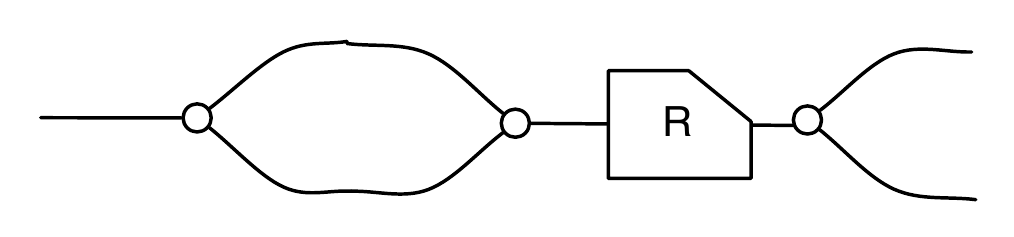}$}
\ \leq\ 
\lower12pt\hbox{$\includegraphics[height=1cm]{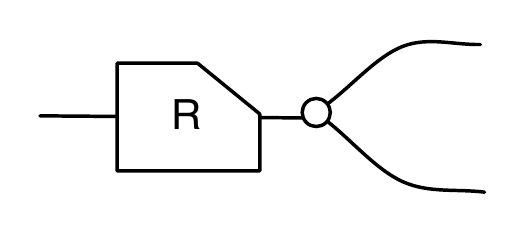}$}
\end{equation*}

\begin{equation*}
\lower5pt\hbox{$\includegraphics[height=.6cm]{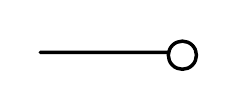}$}
\leq
\lower5pt\hbox{$\includegraphics[height=.6cm]{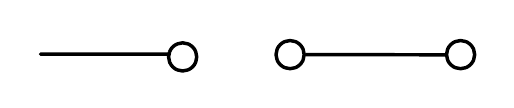}$}
\leq
\lower10pt\hbox{$\includegraphics[height=1cm]{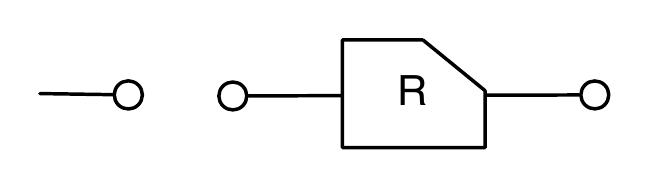}$}
\leq
\lower10pt\hbox{$\includegraphics[height=1cm]{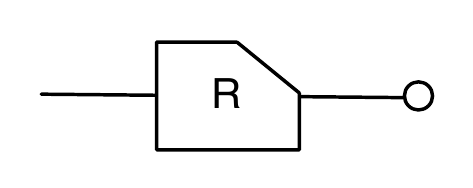}$}
\end{equation*}

\begin{lemma}\label{lem:dagadd}
If $R$ is additive, then $R^\dag$ is additive. %If $R$ has zero, then $R^\dag$ has zero.
\end{lemma}
\begin{proof}
Trivial by using the inequations above.
\end{proof}

\begin{lemma}\label{lem:compadd}
Let $R$ and $S$ be two arrows.
%\begin{enumerate}
\item If $R$ and $S$ are additive then $R;S$ and $R\tns S$ are additive.
%\item If $R$ and $S$ have zero then $R;S$ and $R\tns S$ have zero.
%\end{enumerate}
\end{lemma}
\begin{proof}
Trivial by definition.
\end{proof}

\begin{proposition}\label{prop:monoidadditive}
All arrows in $\Frob{\CMtheoryF}$ are additive.
\end{proposition}
\begin{proof}
The proof proceeds by induction on  $\Frob{\CMtheoryF}$. For the inductive cases, we use Lemma \ref{lem:compadd}. For the base cases, we proceed with a case analysis. We show the cases for both $\cgr[height=13pt]{Wmult.pdf}$,  $\cgr[height=13pt]{Wunit.pdf}$, $\cgr[height=13pt]{Bcomult.pdf}$, $\cgr[height=13pt]{Bcounit.pdf}$ and  $\cgr[height=13pt]{Bmult.pdf}$,  $\cgr[height=13pt]{Bunit.pdf}$, $\cgr[height=13pt]{Wcomult.pdf}$, $\cgr[height=13pt]{Wcounit.pdf}$. The latter are redundant by Lemma \ref{lem:dagadd}, but this extended case analysis will turn out to be useful later on.
\begin{itemize}
\item For $\cgr[height=20pt]{Wmult.pdf}$, 
\begin{multicols}{2}
\noindent
\begin{equation*}
\lower2pt\hbox{$
\lower20pt\hbox{$\includegraphics[height=1.5cm]{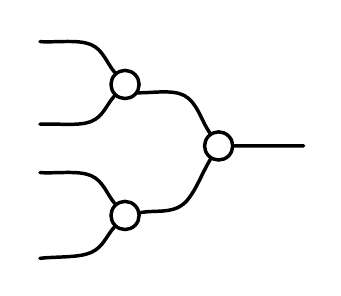}$}
\leq
\lower20pt\hbox{$\includegraphics[height=1.5cm]{graffles/wtree.pdf}$}
$}
\end{equation*}
\begin{equation*}
\label{eq:unitsr}
\lower5pt\hbox{$\includegraphics[height=.6cm]{graffles/Wunit.pdf}$}
\leq
\lower8pt\hbox{$\includegraphics[height=.8cm]{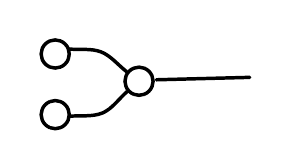}$}
\end{equation*}
\end{multicols}
\item For $\cgr[height=16pt]{Wunit.pdf}$, 
\begin{multicols}{2}
\noindent
\begin{equation*}
\lower2pt\hbox{$
\lower8pt\hbox{$\includegraphics[height=.8cm]{graffles/wuum.pdf}$}
\leq
\lower5pt\hbox{$\includegraphics[height=.6cm]{graffles/Wunit.pdf}$}
$}
\end{equation*}
\begin{equation*}
\label{eq:unitsr}
\lower5pt\hbox{$\includegraphics[height=.6cm]{graffles/Wunit.pdf}$}
\leq
\lower5pt\hbox{$\includegraphics[height=.6cm]{graffles/Wunit.pdf}$}
\end{equation*}
\end{multicols}
\item For $\cgr[height=20pt]{Bcomult.pdf}$, 
\begin{multicols}{2}
\noindent
\begin{equation*}
\lower2pt\hbox{$
\lower15pt\hbox{$\includegraphics[height=1.2cm]{graffles/bialgr.pdf}$}
\leq
\lower9pt\hbox{$\includegraphics[height=.9cm]{graffles/bialgl.pdf}$}
$}
\end{equation*}
\begin{equation*}
\label{eq:unitsr}
\lower7pt\hbox{$\includegraphics[height=.8cm]{graffles/runitsr.pdf}$}
\leq
\lower7pt\hbox{$\includegraphics[height=.8cm]{graffles/runitsl.pdf}$}
\end{equation*}
\end{multicols}
\item For $\cgr[height=16pt]{Bcounit.pdf}$, 
\begin{multicols}{2}
\noindent
\begin{equation*}
\lower2pt\hbox{$
\lower7pt\hbox{$\includegraphics[height=.8cm]{graffles/lunitsr.pdf}$}
\leq
\lower7pt\hbox{$\includegraphics[height=.8cm]{graffles/lunitsl.pdf}$}
$}
\end{equation*}
\begin{equation*}
\label{eq:unitsr} 
id_0
\leq
\lower4pt\hbox{$\includegraphics[height=.5cm]{graffles/unitsl.pdf}$}
\end{equation*}
\end{multicols}
\end{itemize}
The inequations that we have displayed so far are rather trivial. Those for their opposite are much more interesting.
\begin{itemize}
\item For $\cgr[height=20pt]{Bmult.pdf}$, 
\begin{multicols}{2}
\noindent
\begin{equation*}
\lower2pt\hbox{$
\lower20pt\hbox{$\includegraphics[height=1.5cm]{graffles/bwmlaxdist3.pdf}$}
\leq
\lower20pt\hbox{$\includegraphics[height=1.5cm]{graffles/bwmlaxdist4.pdf}$}$}
\end{equation*}
\begin{equation*}
\label{eq:unitsr}
\lower5pt\hbox{$\includegraphics[height=.6cm]{graffles/inj4.pdf}$}
\leq
\lower8pt\hbox{$\includegraphics[height=.8cm]{graffles/inj1.pdf}$}
\end{equation*}
\end{multicols}
\item For $\cgr[height=13pt]{Bunit.pdf}$, 
\begin{multicols}{2}
\noindent
\begin{equation*}
\lower2pt\hbox{$
\lower8pt\hbox{$\includegraphics[height=.8cm]{graffles/bubuwm1.pdf}$}
\leq
\lower5pt\hbox{$\includegraphics[height=.6cm]{graffles/bubuwm3.pdf}$}
$}
\end{equation*}
\begin{equation*}
\label{eq:unitsr}
\lower5pt\hbox{$\includegraphics[height=.6cm]{graffles/Wunit.pdf}$}
\leq
\lower5pt\hbox{$\includegraphics[height=.6cm]{graffles/Bunit.pdf}$}
\end{equation*}
\end{multicols}
\item For $\cgr[height=20pt]{Wcomult.pdf}$, 
\begin{multicols}{2}
\noindent
\begin{equation*}
\lower2pt\hbox{$
\lower15pt\hbox{$\includegraphics[height=1.2cm]{graffles/laxwbialg1.pdf}$}
\leq
\lower8pt\hbox{$\includegraphics[height=.8cm]{graffles/laxwbialg9.pdf}$}
$}
\end{equation*}
\begin{equation*}
\label{eq:unitsr} 
\lower8pt\hbox{$\includegraphics[height=.8cm]{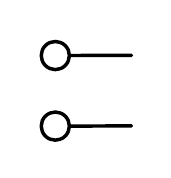}$}
\leq
\lower8pt\hbox{$\includegraphics[height=.8cm]{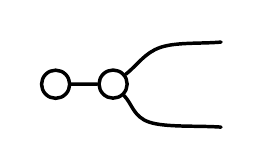}$}
\end{equation*}
\end{multicols}
\item For $\cgr[height=13pt]{Wcounit.pdf}$, 
\begin{multicols}{2}
\noindent
\begin{equation*}
\lower2pt\hbox{$
\lower8pt\hbox{$\includegraphics[height=.8cm]{graffles/nonneg2.pdf}$}
\leq
\lower8pt\hbox{$\includegraphics[height=.8cm]{graffles/nonneg1.pdf}$}$}
\end{equation*}
\begin{equation*}
\label{eq:unitsr} 
id_0
\leq
\lower5pt\hbox{$\includegraphics[height=.6cm]{graffles/wbone.pdf}$}
\end{equation*}
\end{multicols}
\end{itemize}
\end{proof}

\begin{remark} The second sets of inequations in the above proof contains several of those that we have proved in the main text of this section, e.g., the lax bialgebra for the white structure. Observe that these follow from Lemma \ref{lem:dagadd} and the first sets of inequations. This is a more efficient and more structured way to prove those laws. We decided to keep the direct derivations for the sake of exposition.
\end{remark}

\begin{remark}
Suppose now that we want to invert the direction of the above inequations. The first set of inequations would continue to hold: they are actually equalities. The second set no: in particular it would never hold in any non-trivial theory where
\begin{equation*}
\label{eq:unitsr}
\lower5pt\hbox{$\includegraphics[height=.6cm]{graffles/Wunit.pdf}$}
\neq
\lower5pt\hbox{$\includegraphics[height=.6cm]{graffles/Bunit.pdf}$}
\end{equation*}
This also tells us that an analogous of Lemma \ref{lem:dagadd} does not hold for the dual properties of additivity. 
%The notion that in Lyon we called \emph{extensive relations} should divided in four different properties, similar to what happens with injectivity, surjectivity, totality and uniquely definedness.
\end{remark}

\paragraph{White convolution.} Previously we showed that in any cartesian bicategory of relations, every homset carries a meet-semilattice structure. In $\Frob{\CMtheoryF}$ this is actually a lattice. Given arrows $R \colon m\to n$ and $S\colon m\to n$ we define the \emph{white convolution} of $R$ and $S$, written $R\ovee S$, as follows:
\[
R \ovee S \Defeq   \lower20pt\hbox{\includegraphics[height=1.75cm]{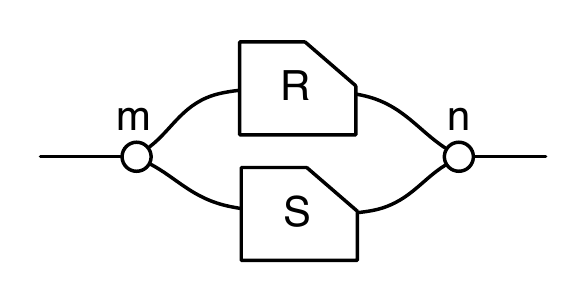}}
\]
\begin{lem}
Convolution is associative, commutative, idempotent and unital. The unit is 
\[
\bot_{m,n} \ \Defeq  \lower8pt\hbox{\includegraphics[height=1cm]{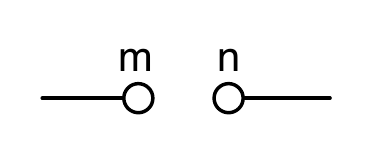}}
\]
\end{lem}
\begin{proof}
The proof for associativity, commutativity and unitality is trivial.
For idempotency, we use Proposition \ref{prop:monoidadditive}:  
\begin{equation*}
\lower2pt\hbox{$
\lower13pt\hbox{$\includegraphics[height=1.3cm]{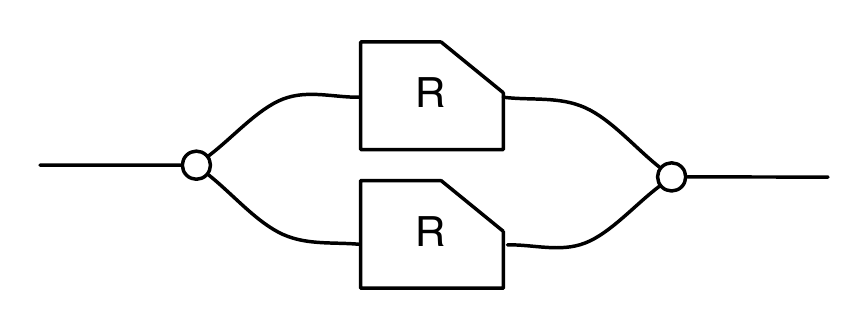}$}
\ \leq\ 
\lower11pt\hbox{$\includegraphics[height=1.1cm]{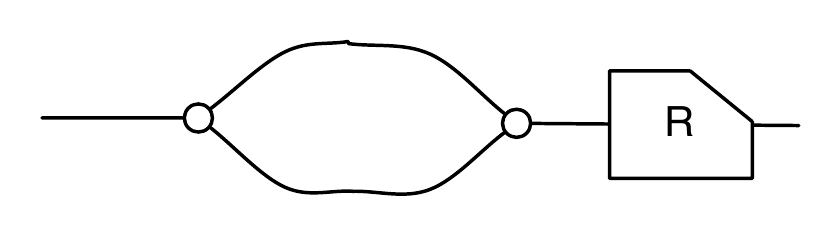}$}
\ =\ 
\lower8pt\hbox{$\includegraphics[height=.8cm]{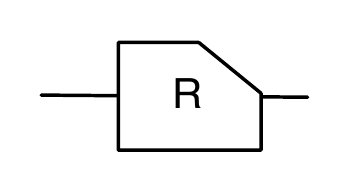}$}
$}
\end{equation*}

\begin{equation*}
\lower2pt\hbox{$
\lower8pt\hbox{$\includegraphics[height=.8cm]{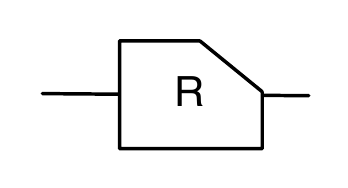}$}
\ =\ 
\lower12pt\hbox{$\includegraphics[height=1.2cm]{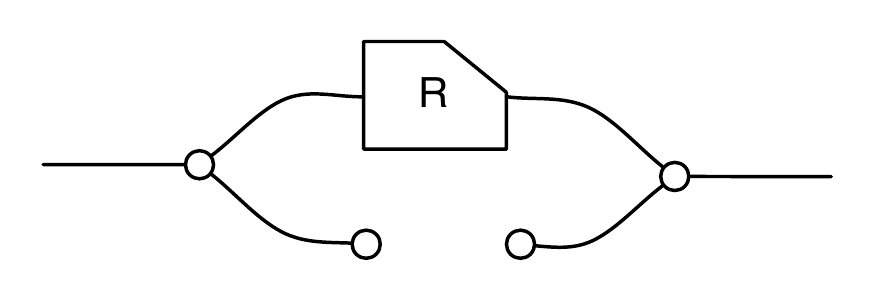}$}
\ \leq\ 
\lower12pt\hbox{$\includegraphics[height=1.2cm]{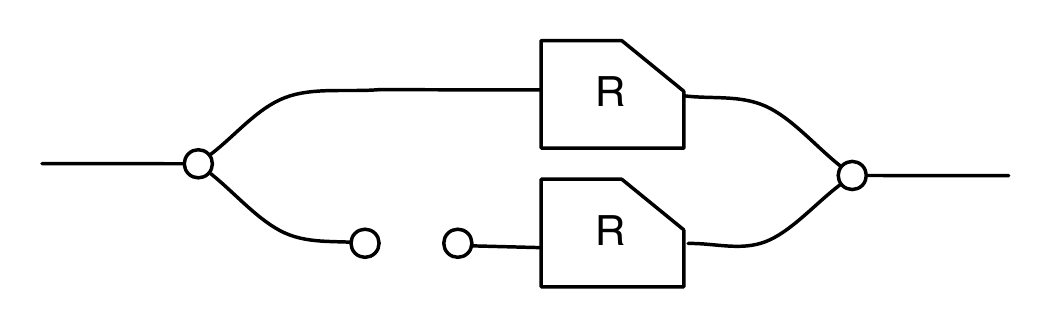}$}
\ \leq\ 
\lower14pt\hbox{$\includegraphics[height=1.3cm]{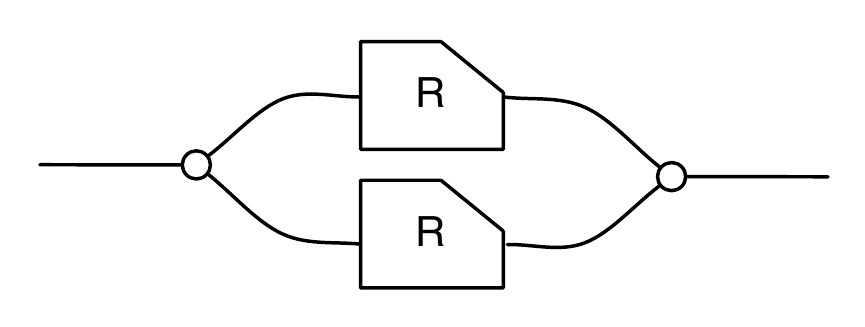}$}
$}
\end{equation*}
\end{proof}

Since $\ovee$ is associative, commutative and idempotent, it induces an ordering. In the following we show that this ordering is exactly $\leq$. 

\begin{lemma}
For all arrows $R\colon n\to m$,
$\bot_{m;n} \leq R$.
\end{lemma}
\begin{proof}
\begin{equation*}
\lower2pt\hbox{$
\lower5pt\hbox{$\includegraphics[height=.6cm]{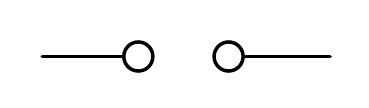}$}
\leq
\lower10pt\hbox{$\includegraphics[height=1cm]{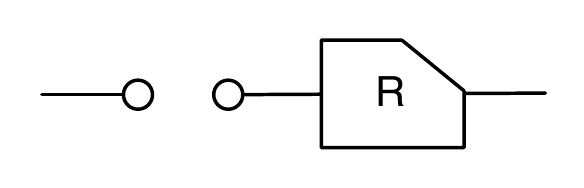}$}
\leq
\lower10pt\hbox{$\includegraphics[height=1cm]{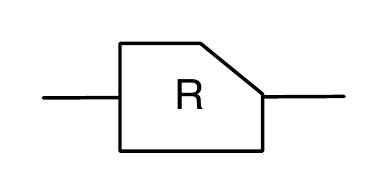}$}
$}
\end{equation*}
\end{proof}

\begin{lemma} 
$R\ovee S = R$ iff $S\leq R$.
\end{lemma}

\begin{proof}
%{\emph{This could be done also graphically!!!!}}
%
Assume$R\ovee S = R$. Then $S = \bot \ovee S \leq R \ovee S = R$.
Assume $S\leq R$. Then $R\ovee S \leq R\ovee R = R$. Moreover $R = R\ovee \bot \leq R \ovee S$.
\end{proof}

\begin{corollary} 
The partial order $\leq$ is a lattice with top and bottom.
\end{corollary}

The following fact holds in any lattice.
%has an extremely cool graphical proof using idempotency and \eqref{ineq:bwmlaxdist}, but it holds in any lattice, so I am not sure whether it is really worth to write it down.
\begin{lemma}\label{lemma:laxdistr}
$R\owedge (S \ovee T) \geq (R\owedge S) \ovee (R\owedge T)$ and 
$R\ovee (S \owedge T) \leq (R\ovee S) \owedge (R\ovee T)$.
\end{lemma} 

Two other useful properties are displayed below.

\begin{lemma}
$(R\ovee S) ; T  \geq (R ; T) \ovee (S ; T)$ and $T ; (R\ovee S)   \geq ( T ; R) \ovee ( T ; S)$.
\end{lemma}
\begin{proof}
It follows immediately from the fact that $R$, $S$ and $T$ are additive.
\end{proof}

\begin{lemma}
$(R\ovee S)^\dag =  R^\dag \ovee S^\dag $.
\end{lemma}
\begin{proof}
Trivial by definition.
\end{proof}

\paragraph{Involution.}
For every arrow $R$, we define  $R^\circ$ as the arrow that is obtained by switching white and black coloring. To be more formal $(-)^\circ \colon \Frob{\CMtheoryF} \to \Frob{\CMtheoryF}$ is the unique symmetric monoidal functor switching black and white colouring. This can be defined inductively: $(R;S)^\circ = R^\circ ; S^\circ$ and $(R\tns S)^\circ = R^\circ \tns S^\circ$  and for the base cases it is just the switching of colours, e.g., $(\Bcomult)^{\circ} = \Wcomult$.

\begin{lemma} The functor $(-)^\circ$ is involutive: $(R^\circ)^\circ = R$.
\end{lemma}
\begin{proof}
Trivial by definition.
\end{proof}

\begin{lemma}
$(R\ovee S)^\circ = R^\circ \owedge S^\circ$ and $(R\owedge S)^\circ = R^\circ \ovee S^\circ$
\end{lemma}
\begin{proof}
Trivial by definition.
\end{proof}

Observe that $(-)^\circ$ does not preserve the posetal structure. Indeed, it is not true in general that $R \leq S$ entails that $S^\circ \leq R^\circ$. Take for instance the black Frobenius equation: we know that the white Frobenius only holds laxly. We will see in Section \ref{sec:ag} that the involution operator satisfies such property when considering Abelian groups.

\paragraph{The algebra of additive relations.}
Our set of operations $$\top, \;\; \bot, \;\; \owedge, \;\; \ovee, \;\; (-)^\dag, \;\; (-)^\circ,\;\; ; , \;\; id$$
is similar to the algebra of relations \cite{jonsson1948representation}: full relation, empty relation, intersection, union, inverse, complement, composition and identity relation. While the allegorical fragment \cite{freyd1990categories}
$$\top,  \;\; \owedge,  \;\; (-)^\dag, \;\; ; , \;\; id$$
-- given by the black structure -- coincides exactly with the one in  the algebra of relations, the remaining part 
$$ \bot,  \;\; \ovee,  \;\; (-)^\circ$$
is not exactly the same: $\bot$ is not the empty relation and $\ovee$ is not the union.

We believe that the connection with relational algebra is worth exploring further.

%\begin{figure}
%\includegraphics[height=5cm]{graffles/SymmetricSummarymonoid.pdf}
%\caption{A more symmetric presentation for the SMIT of Commutative Monoids}
%\end{figure}

%\marginpar{PIPPO: Please Pawel Check that we only need the axioms in the figures}

\section{The theory of abelian groups}\label{sec:ag}
Recall the cartesian theory of abelian groups $\AGtheory$ (Example \ref{ex:FT}(b)). In this section we study the Frobenius theory $\AGtheoryF$ which is obtained by adding to $\AGtheory$ the inequations \eqref{eq:olunitsl}-\eqref{eq:olbwbone} of oplax bialgebras (Example \ref{ex:SMIT}(e)) and the following inequalities stating that the antipode is a map (single valued and total).
\begin{multicols}{2}
\noindent
\begin{equation}
\label{eq:antipodelax1}
\lower2pt\hbox{$
\lower7pt\hbox{$\includegraphics[height=.8cm]{graffles/antipodef1.pdf}$}
\leq
\lower10pt\hbox{$\includegraphics[height=1cm]{graffles/antipodef2.pdf}$}
$}
\end{equation}
\begin{equation}
\label{eq:antipodelax2}
\lower9pt\hbox{$\includegraphics[height=1cm]{graffles/antipodef3.pdf}$}
\leq
\lower5pt\hbox{$\includegraphics[height=.6cm]{graffles/antipodef4.pdf}$}
\end{equation}
\end{multicols}

\begin{figure}
\begin{center}
\includegraphics[height=9cm]{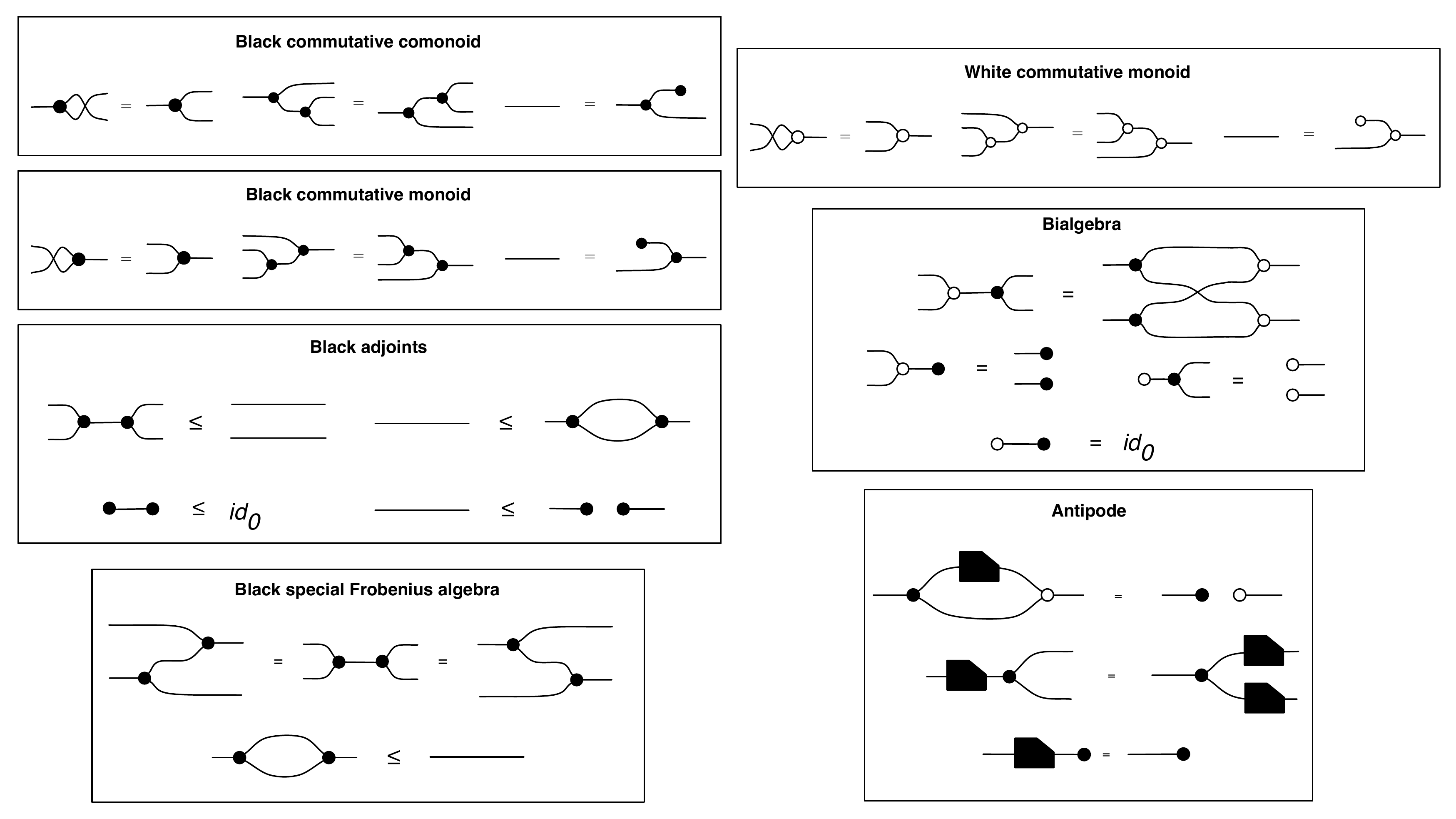}
\end{center}
\caption{The SMIT corresponding to the Frobenius theory of abelian groups.}\label{fig:summaryAbelianGroup}
\end{figure}
The corresponding SMIT is shown in Figure \ref{fig:summaryAbelianGroup}.
As usual, we adopt the following graphical convention. 
$$\lower1pt\hbox{$\cgr[height=20pt]{antipodeinv.pdf}$}::= (\cgr[height=20pt]{antipode.pdf})^\dag$$

We start our treatment by showing that the antipode is surjective 
\begin{equation*}
\lower2pt\hbox{$
\lower5pt\hbox{$\includegraphics[height=.6cm]{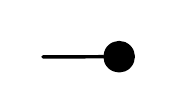}$}
=
\lower5pt\hbox{$\includegraphics[height=.6cm]{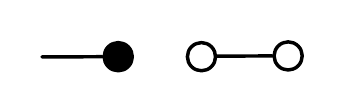}$}
=
\lower11pt\hbox{$\includegraphics[height=1.1cm]{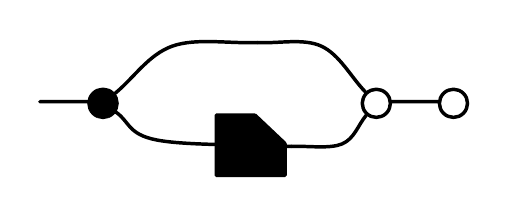}$}
\leq
\lower11pt\hbox{$\includegraphics[height=1.1cm]{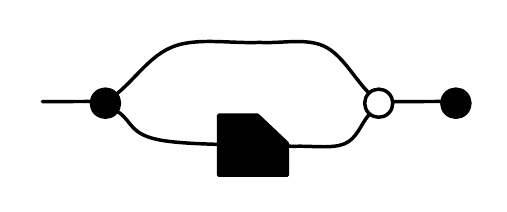}$}
=
\lower11pt\hbox{$\includegraphics[height=1.1cm]{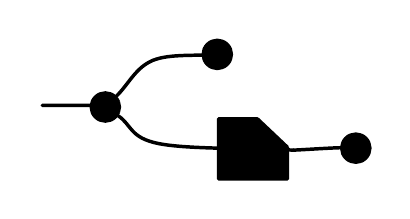}$}
=
\lower8pt\hbox{$\includegraphics[height=.8cm]{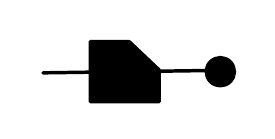}$}
$}
\end{equation*}
and injective.
\begin{multline*}
\lower8pt\hbox{$\includegraphics[height=.8cm]{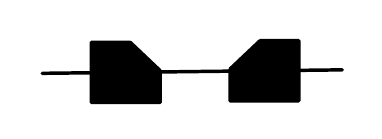}$}
=
\lower12pt\hbox{$\includegraphics[height=1.2cm]{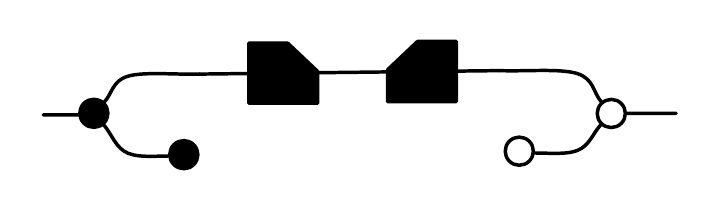}$}
=
\lower15pt\hbox{$\includegraphics[height=1.5cm]{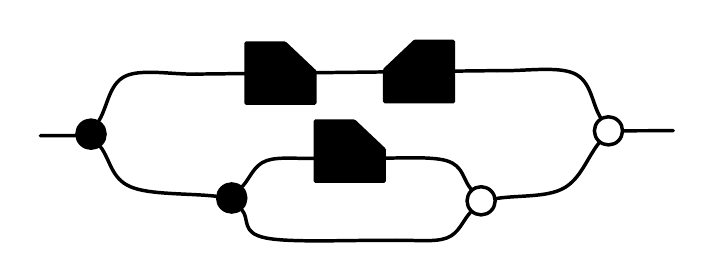}$}
\\
=
\lower10pt\hbox{$\includegraphics[height=1.5cm]{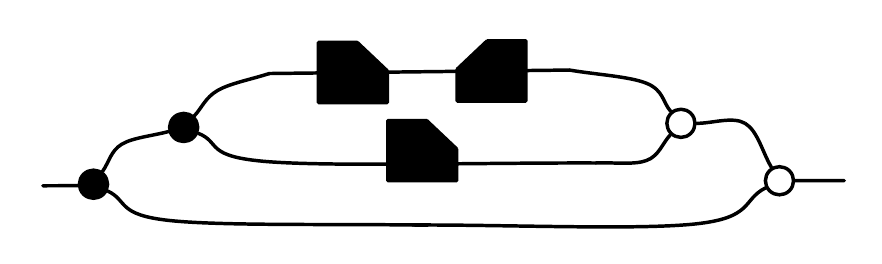}$}
=
\lower10pt\hbox{$\includegraphics[height=1.5cm]{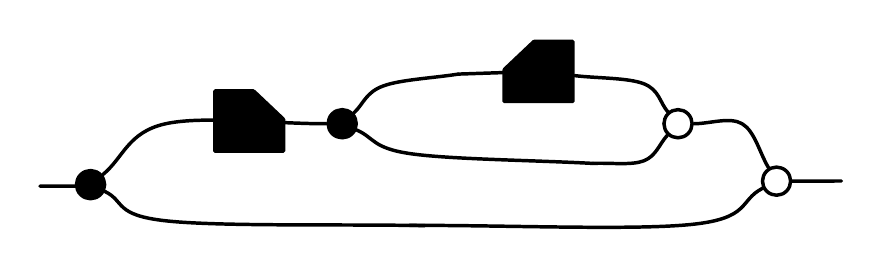}$}
\leq
\lower10pt\hbox{$\includegraphics[height=1.5cm]{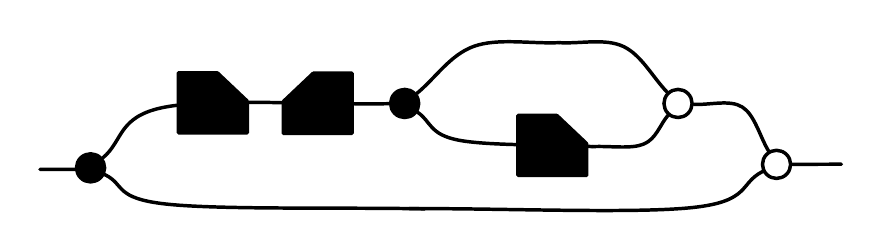}$}
\\
=
\lower9pt\hbox{$\includegraphics[height=1.2cm]{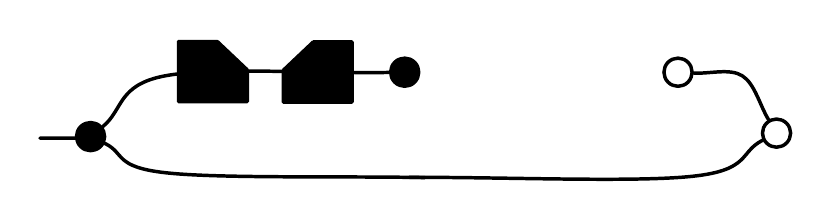}$}
=
\lower8pt\hbox{$\includegraphics[height=1cm]{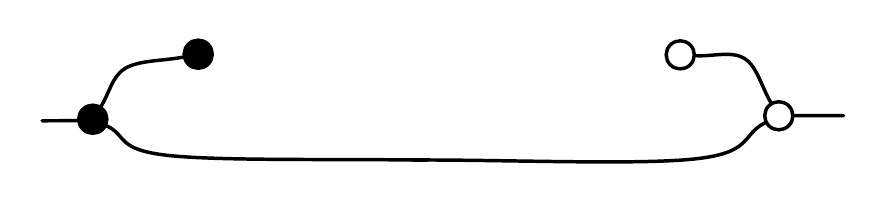}$}
=
\lower5pt\hbox{$\includegraphics[height=.6cm]{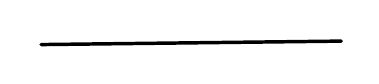}$}
\end{multline*}
In a similar way, one can show that 
\begin{multline*}
\label{eq:inverseantipode}
\lower8pt\hbox{$\includegraphics[height=.8cm]{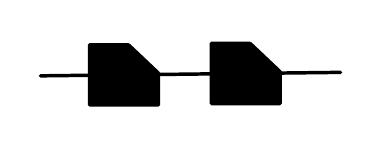}$}
=
\lower12pt\hbox{$\includegraphics[height=1.2cm]{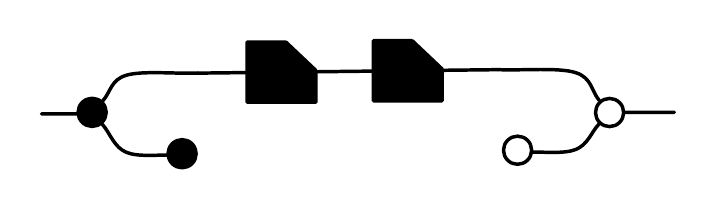}$}
=
\lower18pt\hbox{$\includegraphics[height=1.5cm]{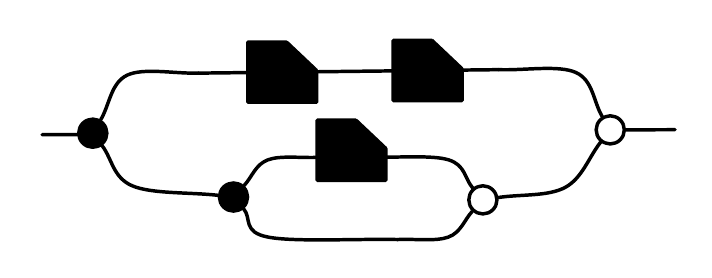}$}
\\
=
\lower10pt\hbox{$\includegraphics[height=1.5cm]{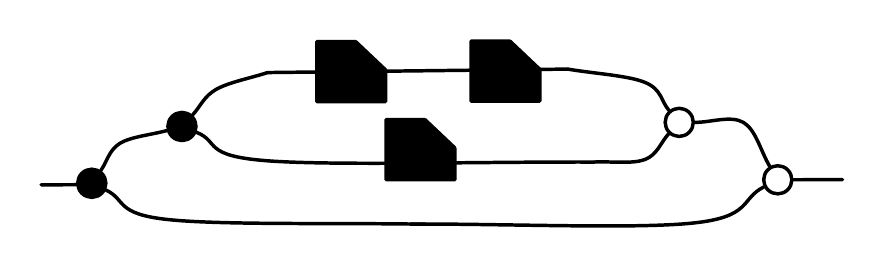}$} 
= 
\lower10pt\hbox{$\includegraphics[height=1.5cm]{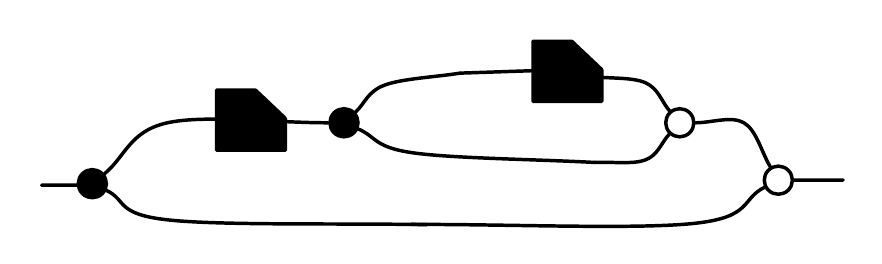}$}
=
\\
\lower10pt\hbox{$\includegraphics[height=1.2cm]{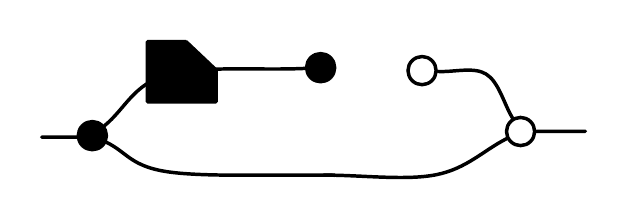}$}
=
\lower8pt\hbox{$\includegraphics[height=1cm]{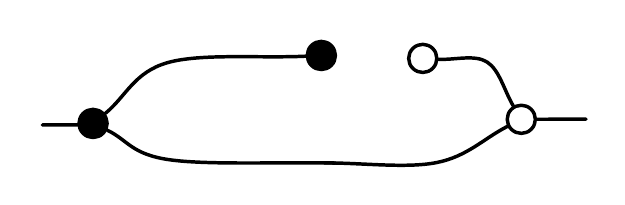}$}
=
\lower5pt\hbox{$\includegraphics[height=.6cm]{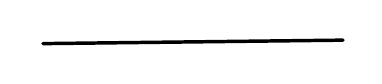}$}
\end{multline*}
which, by Corollary \eqref{cor:inverse}, means that $\cgr[height=17pt]{antipodeinv.pdf}= \cgr[height=17pt]{antipode.pdf}$. This fact justifies the adoption of following graphical notation.
$$ \lower9pt\hbox{$\includegraphics[height=23pt]{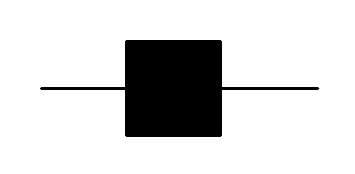}$} ::= 
\lower9pt\hbox{$\includegraphics[height=23pt]{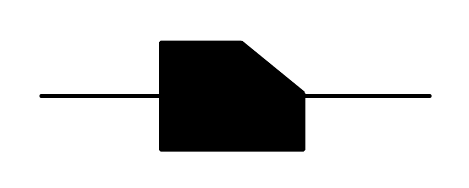}$}=
\lower9pt\hbox{$\includegraphics[height=23pt]{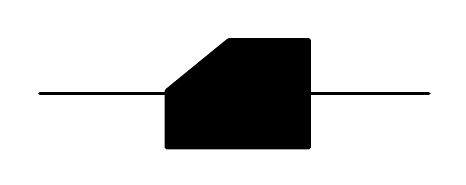}$}
 $$

\medskip

We now investigate how the antipode interacts with the white monoid. The following derivation shows that it distributes over the white unit.

\begin{equation}
\label{eq:antipodewunit}
\lower2pt\hbox{$
\lower5pt\hbox{$\includegraphics[height=.6cm]{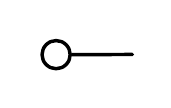}$}
=
\lower5pt\hbox{$\includegraphics[height=.6cm]{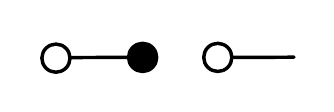}$}
=
\lower14pt\hbox{$\includegraphics[height=1.2cm]{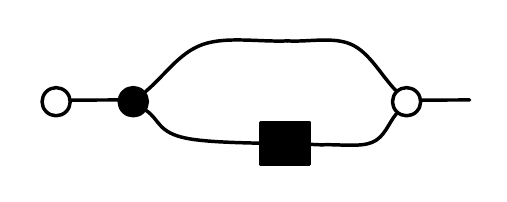}$}
=
\lower14pt\hbox{$\includegraphics[height=1.2cm]{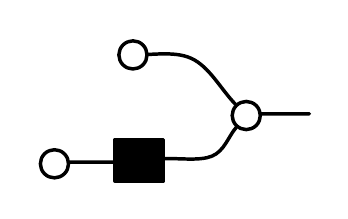}$}
=
\lower8pt\hbox{$\includegraphics[height=.8cm]{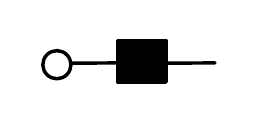}$}
$}
\end{equation}

The same happens for the white multiplication. This is easy to show, using the following useful ``De Morgan'' property: 
\begin{equation}\label{eq:demorgan}
\lower13pt\hbox{$\includegraphics[height=1.1cm]{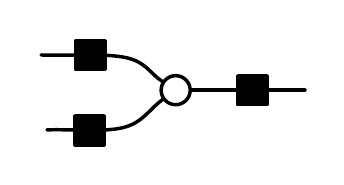}$}
= 
\lower12pt\hbox{$\includegraphics[height=1cm]{graffles/Wmult.pdf}$}
\end{equation}
First we prove~\eqref{eq:demorgan}:
\[
\lower13pt\hbox{$\includegraphics[height=1.1cm]{graffles/demorgan.pdf}$}
\leq 
\lower12pt\hbox{$\includegraphics[height=1cm]{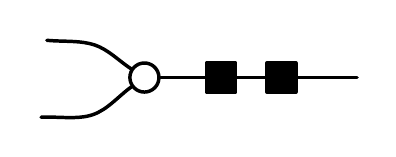}$}
=
\lower12pt\hbox{$\includegraphics[height=1cm]{graffles/Wmult.pdf}$}
\]
\[
\lower12pt\hbox{$\includegraphics[height=1cm]{graffles/Wmult.pdf}$}
=
\lower13pt\hbox{$\includegraphics[height=1.1cm]{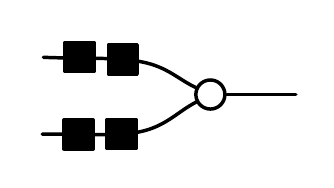}$}
\leq
\lower13pt\hbox{$\includegraphics[height=1.1cm]{graffles/demorgan.pdf}$}
\]
And now it follows that
\begin{equation}
\label{eq:antipodewmult}
\lower13pt\hbox{$\includegraphics[height=1.1cm]{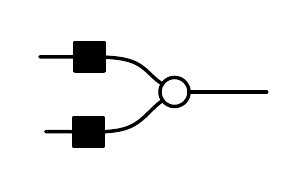}$}
=
\lower13pt\hbox{$\includegraphics[height=1.1cm]{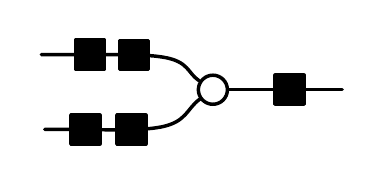}$}
=
\lower12pt\hbox{$\includegraphics[height=1cm]{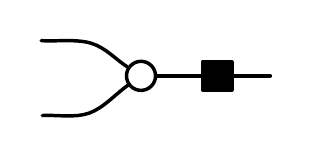}$}
\end{equation}

A useful ``quasi-Frobenius'' interaction between the white structure is the following.
\begin{equation}\label{eq:quasifrob}
\lower17pt\hbox{$\includegraphics[height=1.5cm]{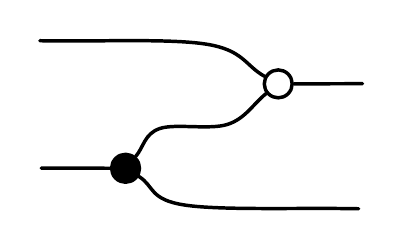}$}
=
\lower17pt\hbox{$\includegraphics[height=1.5cm]{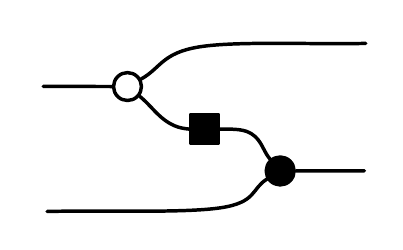}$}
\end{equation}
To see this, let 
\[
\lower12pt\hbox{$\includegraphics[height=1cm]{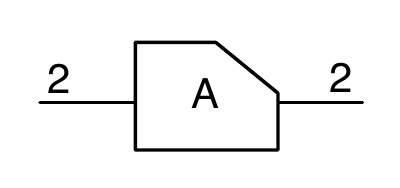}$}
:=
\lower15pt\hbox{$\includegraphics[height=1.2cm]{graffles/quasifrob1.pdf}$}
\qquad
\lower12pt\hbox{$\includegraphics[height=1cm]{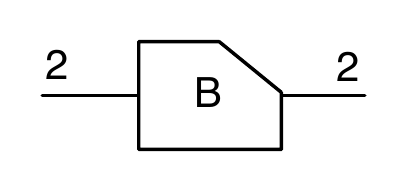}$}
:=
\lower15pt\hbox{$\includegraphics[height=1.2cm]{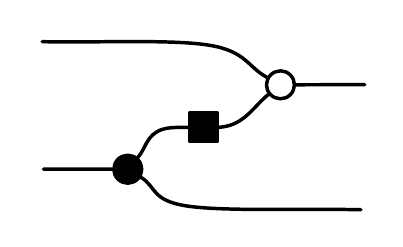}$}
\]
Now it is easy to show that $A\poi B = B\poi A = id_2$. By the conclusion of Lemma~\ref{lem:dagger}, it follows that 
\[
\lower12pt\hbox{$\includegraphics[height=1cm]{graffles/Abox.pdf}$}
=
\lower12pt\hbox{$\includegraphics[height=1cm]{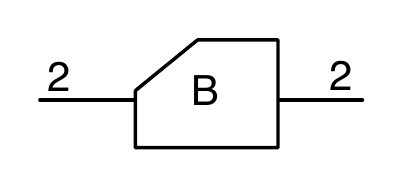}$}
\]
which is just~\eqref{eq:quasifrob}.

Another important law is the following
\begin{equation}
\label{eq:bwcc}
\lower2pt\hbox{$
\lower10pt\hbox{$\includegraphics[height=1cm]{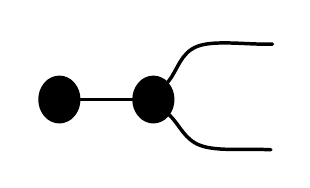}$}
=
\lower10pt\hbox{$\includegraphics[height=1cm]{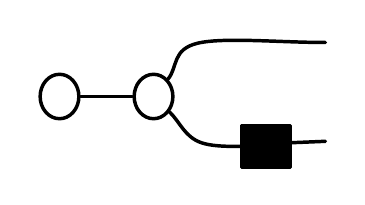}$}
$}
\end{equation}
This follows from~\eqref{eq:quasifrob}:
\[
\lower10pt\hbox{$\includegraphics[height=1cm]{graffles/bwcc2.pdf}$}
=
\lower20pt\hbox{$\includegraphics[height=2cm]{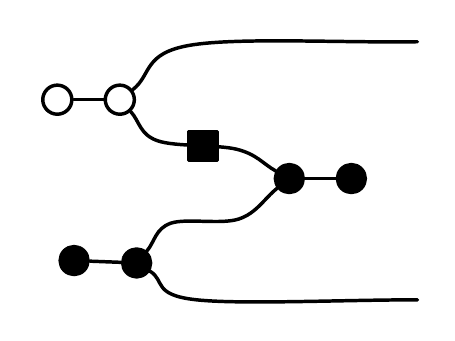}$}
=
\lower20pt\hbox{$\includegraphics[height=2cm]{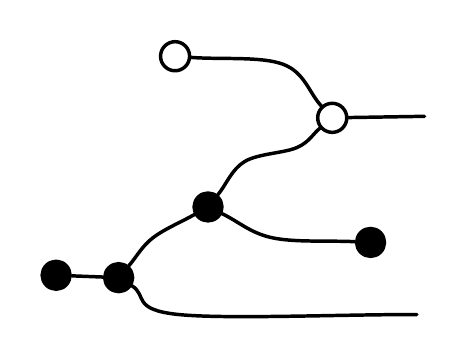}$}
=
\lower10pt\hbox{$\includegraphics[height=1cm]{graffles/bwcc.pdf}$}
\]

\begin{proposition}\label{prop:whiteFrobenius}
White monoid and comonoid form a Frobenius structure.
\end{proposition}
\begin{proof}
Since the black structure is Frobenius, it suffices to show:
\begin{multline*}
\lower15pt\hbox{$\includegraphics[height=1.5cm]{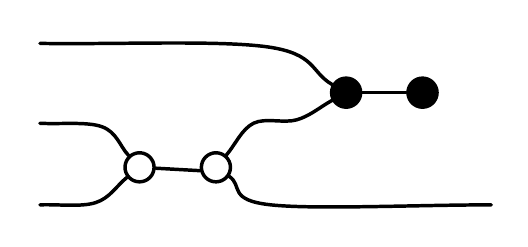}$}
=
\lower17pt\hbox{$\includegraphics[height=1.7cm]{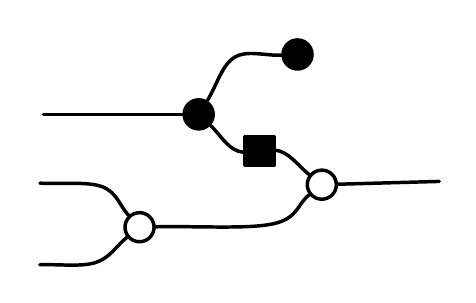}$}
=
\lower15pt\hbox{$\includegraphics[height=1.5cm]{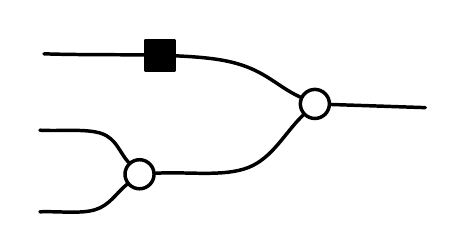}$} \\
=
\lower15pt\hbox{$\includegraphics[height=1.5cm]{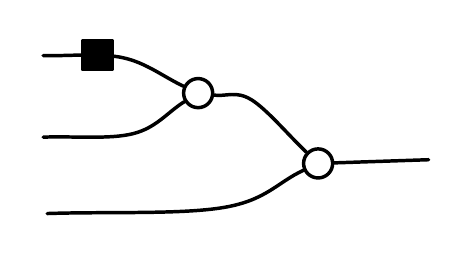}$}
=
\lower19pt\hbox{$\includegraphics[height=1.9cm]{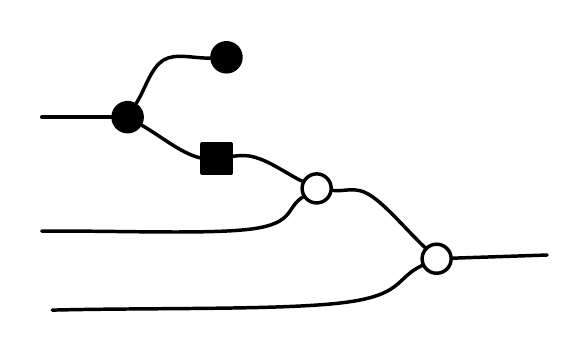}$}
=
\lower18pt\hbox{$\includegraphics[height=1.8cm]{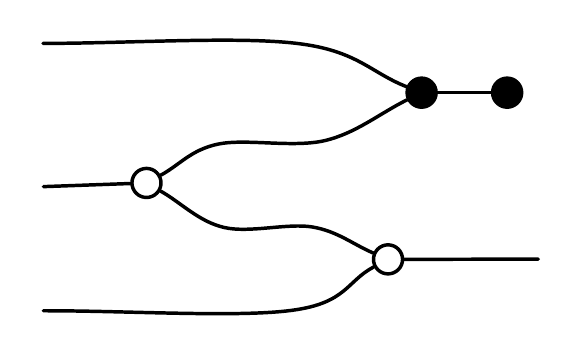}$}
\end{multline*}
\end{proof}

\paragraph{An alternative presentation.} Thanks to Proposition \ref{prop:whiteFrobenius}, it is easy to see that all the (in)equations in Figure \ref{fig:summaryAlternativeAG} hold in $\Frob{\AGtheoryF}$. Actually, one can prove that the ordered PROP freely generated by the SMIT in Figure \ref{fig:summaryAlternativeAG} is isomorphic to $\Frob{\AGtheoryF}$: it is enough to define the antipode as either of the following (equivalent) terms 
\[
\includegraphics[height=1.5cm]{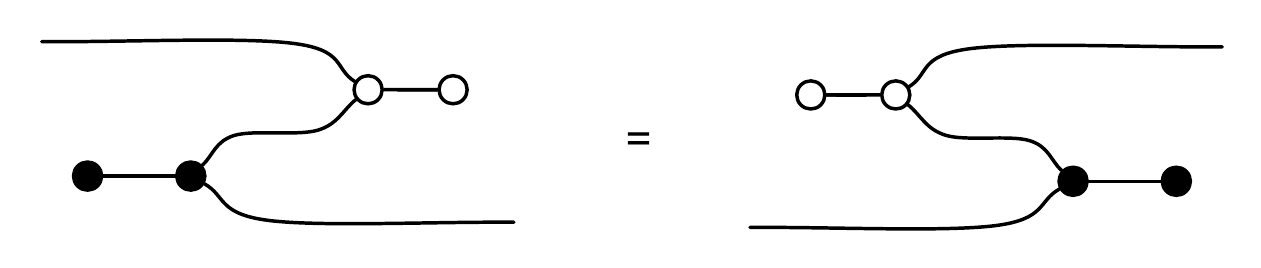}
\]
and check that the (in)equations in Figure \ref{fig:summaryAbelianGroup} are entailed by those in in Figure \ref{fig:summaryAlternativeAG}. From the isomorphism, it follows that $\Frob{\AGtheoryF}$ is an \emph{abelian bicategory} \cite{Carboni1987}.

%An alternative presentation for the theory $\AGtheoryF$ is shown in Figure \ref{fig:summaryAlternativeAG} and discussed further in Section \ref{sec:alternativeAG}. By the fact that the two presentations are equivalent,  it follows this interesting corollary. %
%\begin{corollary}
%$\Frob{\Theory{AG}}$ is the free Abelian bi-Category in the sense of Carboni and Walters.
%\end{corollary}
%
\begin{figure}
\includegraphics[height=8cm]{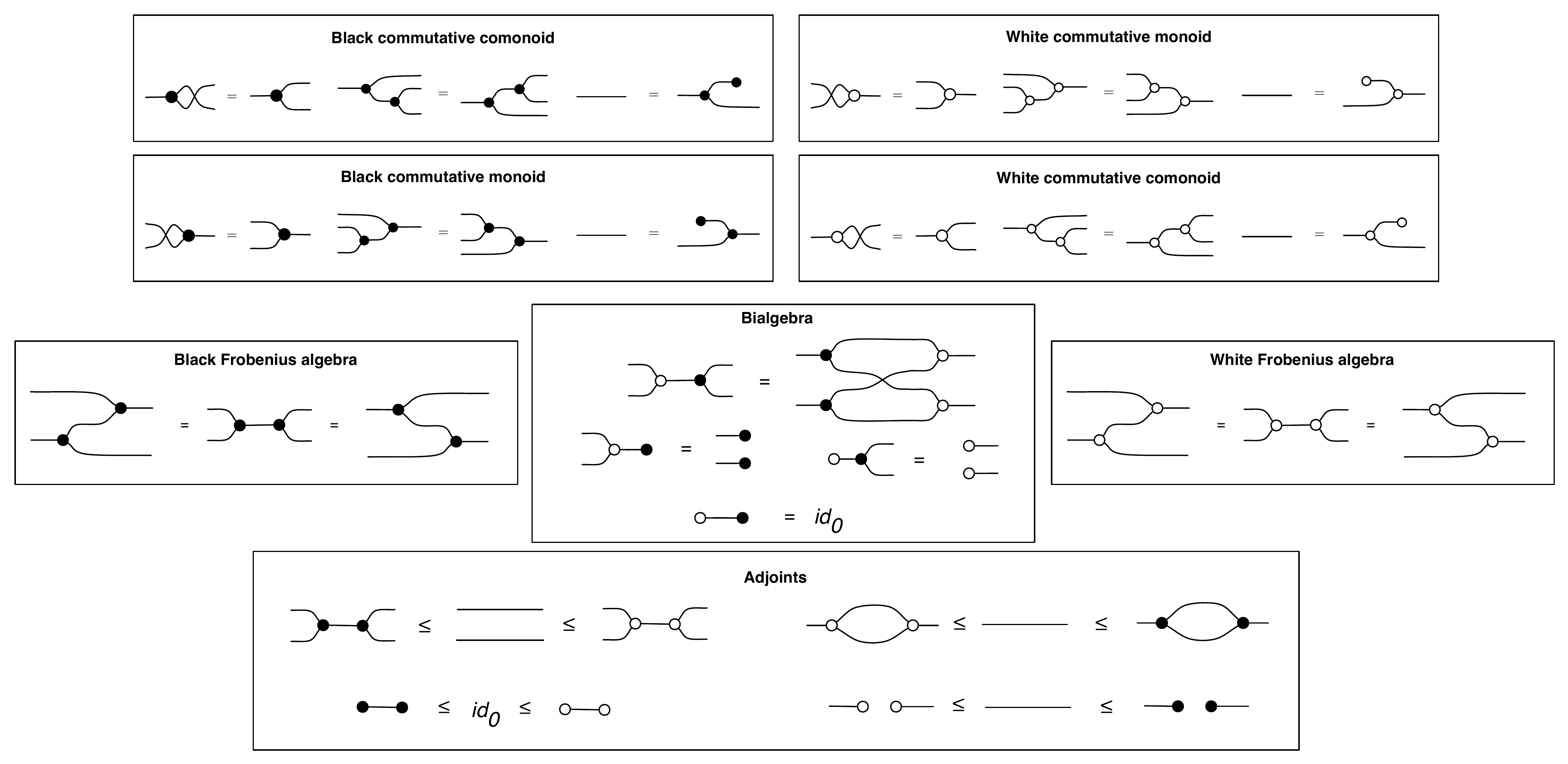}
\caption{An alternative presentation for the SMIT of abelian Groups}\label{fig:summaryAlternativeAG}
\end{figure}

\paragraph{Involution.}
Involution on $\Frob{\CMtheoryF}$can be easily extended to $\Frob{\AGtheoryF}$ by defining 
$$ (\lower5pt\hbox{$\includegraphics[height=.6cm]{graffles/sqantipode.pdf}$})^\circ ::=  \lower5pt\hbox{$\includegraphics[height=.6cm]{graffles/sqantipode.pdf}$}\text{.}$$
Now, the fact that white monoid and white comonoid form a special Frobenius algebra gives us the following important proposition.

\begin{proposition}\label{prop:involutionreverse}
$R \leq S$ iff $R^\circ \geq S^\circ$.
\end{proposition}
\begin{proof}
It is enough to check that, for each of the inequation $R\leq S$ in the axiomatization in Fig. \ref{fig:summaryAbelianGroup}, $R^\circ \leq S^\circ$ holds.
\end{proof}

The above proposition can be used as an effective proof technique.
For instance, to prove that 
\begin{equation}
\label{eq:buwm}
\lower2pt\hbox{$
\lower8pt\hbox{$\includegraphics[height=.8cm]{graffles/buwmL.pdf}$}
=
\lower5pt\hbox{$\includegraphics[height=.6cm]{graffles/buwmR.pdf}$}
$}
\end{equation}
it is enough to recall equation \eqref{eq:wubm} and apply Proposition \ref{prop:involutionreverse}.

\medskip

Since the antipode is a morphism of white monoid (see \eqref{eq:antipodewunit} and \eqref{eq:antipodewmult}), it is an additive arrow. By induction and Lemma \ref{lem:compadd}, one can easily prove that all the morphisms in $\Frob{\AGtheoryF}$ are additive. As a consequence, all the results about $\ovee$, $\bot$, $\owedge$, $\top$, and $(-)^\circ$ proved for commutative monoids, still hold for abelian groups. Observe that these operators do \emph{not} form a Boolean algebra as $\ovee$ and $\owedge$ distribute over each other only laxly.
Other axioms of Boolean algebras that fail are $R\owedge R^\circ = \bot$ and $R\ovee R^\circ = \top$: to see this, it is enough to take $R=id$.
%Maybe we have that $R\owedge A(R^\circ) = \bot$ and $R\ovee A(R^\circ) = \top$ where $A(-)$ adds an antipode, but I am not sure whether this makes much sense. We should better think about it.
%
% REFLEXIVE AND TRANSITIVE CLOSURE
%
%\begin{remark}
%The algebra of additive relations over an Abelian group seems to enjoy more interesting properties than the algebra over a monoid. For instance, for additive relations, the following holds: if $R$ is reflexive, than it is also symmetric and transitive.
%As a consequence the reflexive and transitive closure of $R$ is just $R\ovee 1$.
%\end{remark}

%Actually, one can prove that these operators form a Boolean algebra: by Lemma \ref{lemma:laxdistr}, we have that
%$$R\owedge (S \ovee T) \geq (R\owedge S) \ovee (R\owedge T) \text{ and } 
%R\ovee (S \owedge T) \leq (R\ovee S) \owedge (R\ovee T)$$
% for all arrows $R,S$ and $T$. This holds in particular for $R^\circ,S^\circ$ and $T^\circ$, i.e.,
%$$R^\circ\owedge (S^\circ \ovee T^\circ) \geq (R^\circ\owedge S^\circ) \ovee (R^\circ\owedge T^\circ) \text{ and } 
%R^\circ\ovee (S^\circ \owedge T^\circ) \leq (R^\circ\ovee S^\circ) \owedge (R^\circ\ovee T^\circ)\text{.}$$
%Using Proposition \ref{prop:involutionreverse}, we have that 
%$$R\ovee (S \owedge T) \leq (R\ovee S) \owedge (R\ovee T) \text{ and } 
%R\owedge (S \ovee T) \geq (R\owedge S) \ovee (R\ovee T)$$

\paragraph{Antipodal arrows.} Since the white structure joins the properties of special Frobenius algebras, it induces a compact closed structure and a contravariant monoidal 2-functor $(-)^\ddag$ mapping every arrow $R$ in
\begin{equation*}
\lower2pt\hbox{$
R^\ddag ::=
\lower14pt\hbox{$\includegraphics[height=1.4cm]{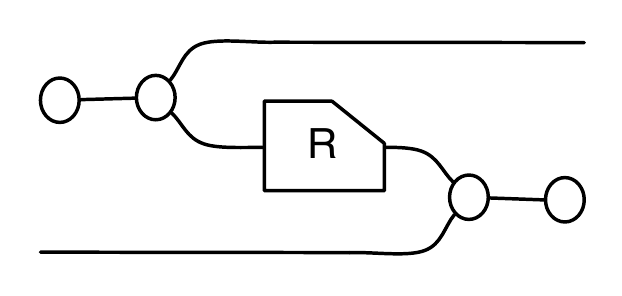}$}
$}\text{.}
\end{equation*}

In order to have that  $R^\ddag = R^\dag$, we need $R$ to be \emph{antipodal}, that is 
\begin{equation*}
\lower2pt\hbox{$
\lower8pt\hbox{$\includegraphics[height=.8cm]{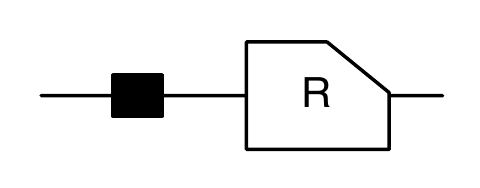}$}
\leq
\lower8pt\hbox{$\includegraphics[height=.8cm]{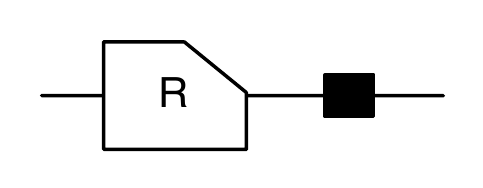}$}
$}
\end{equation*}
which entails
\begin{equation*}
\lower2pt\hbox{$
\lower8pt\hbox{$\includegraphics[height=.8cm]{graffles/antipodal1.pdf}$}
=
\lower8pt\hbox{$\includegraphics[height=.8cm]{graffles/antipodal2.pdf}$}
$}\text{.}
\end{equation*}
So $R$ is antipodal iff it is an homomorphism w.r.t. the antipode.

\begin{proposition}\label{prop:antipodalcc}
If $R$ is antipodal, then $R^\ddag = R^\dag$.
\end{proposition}
%\begin{proof}
%The graphical derivation is trivial by using \eqref{eq:bwcc}, antipodality and \eqref{eq:inverseantipode}.
%\end{proof}

%In the following we refer to arrows which are Additive, with Zero and Antipodal  as AZA.

\begin{proposition}
If $R$ is additive and antipodal, then the following hold.
\begin{itemize}
\item 
\begin{equation}
\label{eq:singlevaluedAG}
\lower2pt\hbox{$
\tag{Single Valued}
\lower12pt\hbox{$\includegraphics[height=1.2cm]{graffles/singlevaluedAG1.pdf}$}
\leq
\lower8pt\hbox{$\includegraphics[height=.8cm]{graffles/singlevaluedAG2.pdf}$}
\text{ iff }
\lower8pt\hbox{$\includegraphics[height=.8cm]{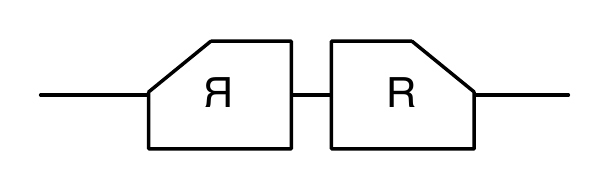}$}
\leq
\lower5pt\hbox{$\includegraphics[height=.6cm]{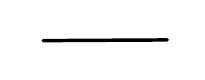}$}
\text{ iff }
\lower8pt\hbox{$\includegraphics[height=.8cm]{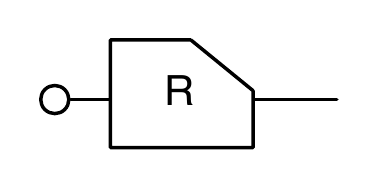}$}
\leq
\lower5pt\hbox{$\includegraphics[height=.6cm]{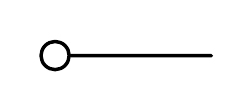}$}
$}
\end{equation}
\item 
\begin{equation}
\label{eq:totalAG}
\lower2pt\hbox{$
\tag{Total}
\lower5pt\hbox{$\includegraphics[height=.6cm]{graffles/totalAG1.pdf}$}
\leq
\lower7pt\hbox{$\includegraphics[height=.8cm]{graffles/totalAG2.pdf}$}
\text{ iff }
\lower5pt\hbox{$\includegraphics[height=.6cm]{graffles/idonly.pdf}$}
\leq
\lower7pt\hbox{$\includegraphics[height=.8cm]{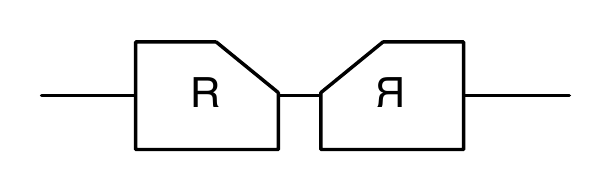}$}
\text{ iff }
\lower8pt\hbox{$\includegraphics[height=.8cm]{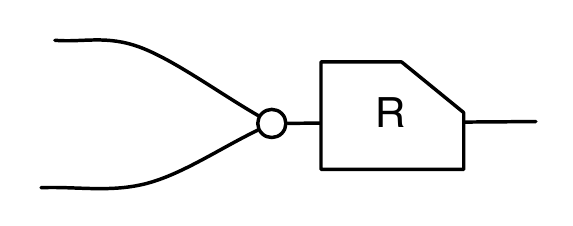}$}
\leq
\lower12pt\hbox{$\includegraphics[height=1.2cm]{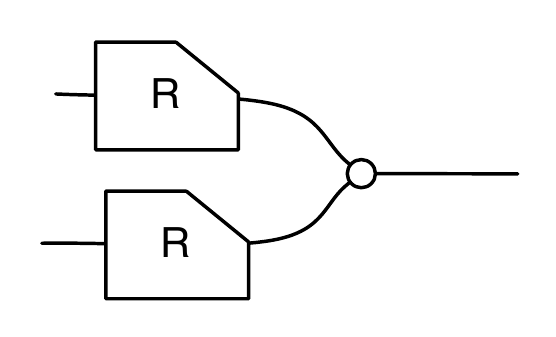}$}
$}
\end{equation}
\item
\begin{equation}
\label{eq:surjectiveAG}
\lower2pt\hbox{$
\tag{Surjective}
\lower5pt\hbox{$\includegraphics[height=.6cm]{graffles/surjectiveAG1.pdf}$}
\leq
\lower7pt\hbox{$\includegraphics[height=.8cm]{graffles/surjectiveAG2.pdf}$}
\text{ iff }
\lower5pt\hbox{$\includegraphics[height=.6cm]{graffles/idonly.pdf}$}
\leq
\lower7pt\hbox{$\includegraphics[height=.8cm]{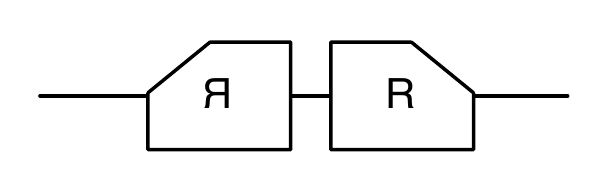}$}
\text{ iff }
\lower8pt\hbox{$\includegraphics[height=.8cm]{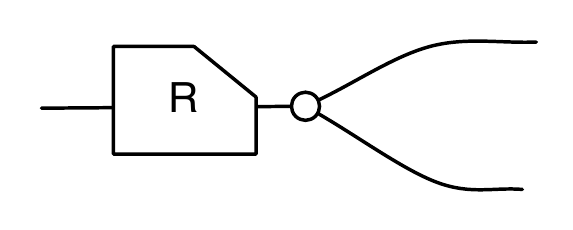}$}
\leq
\lower12pt\hbox{$\includegraphics[height=1.2cm]{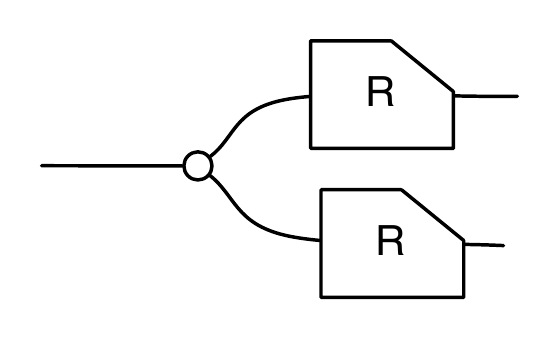}$}
$}
\end{equation}
\item Injective:
\begin{equation}
\label{eq:injectiveAG}
\lower2pt\hbox{$
\tag{Injective}
\lower12pt\hbox{$\includegraphics[height=1.2cm]{graffles/injectiveAG1.pdf}$}
\leq
\lower7pt\hbox{$\includegraphics[height=.8cm]{graffles/injectiveAG2.pdf}$}
\text{ iff }
\lower7pt\hbox{$\includegraphics[height=.8cm]{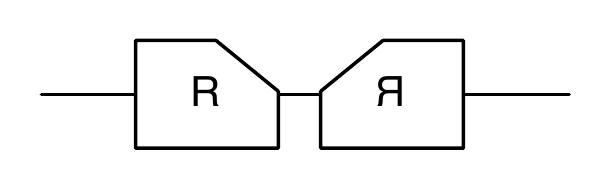}$}
\leq
\lower5pt\hbox{$\includegraphics[height=.6cm]{graffles/idonly.pdf}$}
\text{ iff }
\lower7pt\hbox{$\includegraphics[height=.8cm]{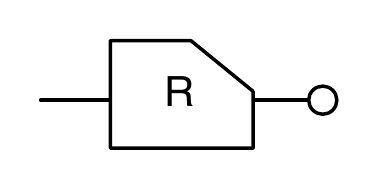}$}
\leq
\lower5pt\hbox{$\includegraphics[height=.6cm]{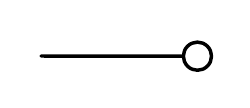}$}
$}
\end{equation}
\end{itemize}
\end{proposition}
\begin{proof}
Since $R$ is antipodal, by Proposition \ref{prop:antipodalcc}, $R^\dag = R^\ddag$: we can refer to the black and white compact close structure as the same thing.

The axioms of cartesian bicategories imply that $R$ is a \emph{lax black} comonoid and lax monoid homomorphism. This fact, together with the black Frobenius structure was used in Lemma \ref{lem:characterizationmap} to establish the \emph{iff} between the first two columns in the four rows above.

Since $R$ is additive and has zero, it is a \emph{oplax white} monoid (by definition) and an oplax white comonoid homomorphism (Lemma \ref{lem:dagadd}). 

To prove the the \emph{iff} between the third and the second column, it is enough to repeat exactly the same proofs of Lemma \ref{lem:characterizationmap} by exchanging colours and the direction of $\leq$ and $\geq$. For instance, the \emph{iff} between the second and the third column in \eqref{eq:singlevaluedAG} can be retrieved by the \emph{iff} between the second and the first column in \eqref{eq:surjectiveAG}.
\end{proof}

\begin{proposition}
All arrows in $\Frob{\AGtheoryF}$ are additive and antipodal.
\end{proposition}
\begin{proof}
By induction. The cases are given by Lemma \ref{lem:compadd}, and can be easily extended to take into account antipodality. For the base cases, one can just use those of Proposition \ref{prop:monoidadditive} and observe that every operation of the theory is an antipode homomorphism.
\end{proof}

\paragraph{Cancellation.} We conclude this section by introducing the cancellation proof technique that will be useful for proving properties about modules in the next section.

\begin{proposition}\label{prop:cancellation}
Let $R$ and $S$ be arrows in $\Frob{\AGtheoryF}$ and $X$ be a map. 
\begin{equation*}
\text{If } \lower15pt\hbox{$\includegraphics[height=1.5cm]{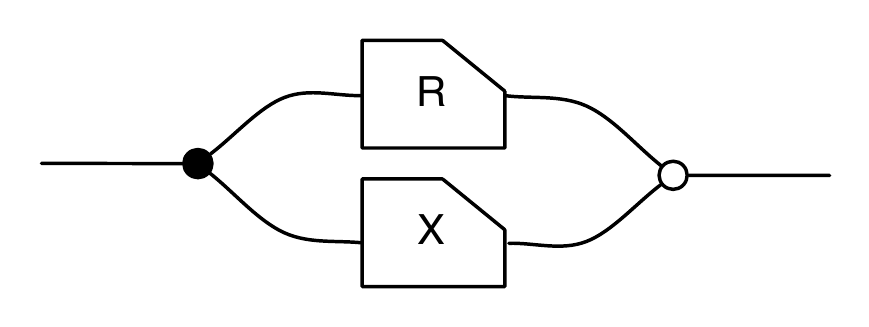}$}
=
\lower15pt\hbox{$\includegraphics[height=1.5cm]{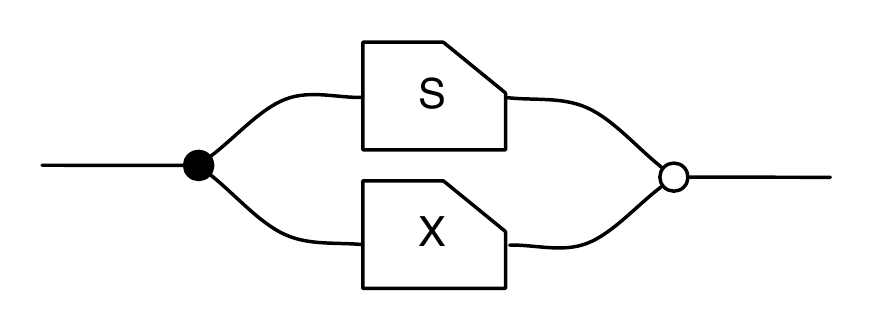}$} 
\text{, then }
\lower7pt\hbox{$\includegraphics[height=.8cm]{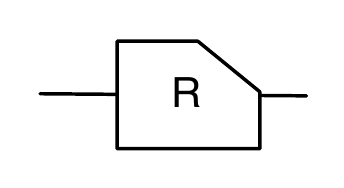}$}
=
\lower7pt\hbox{$\includegraphics[height=.8cm]{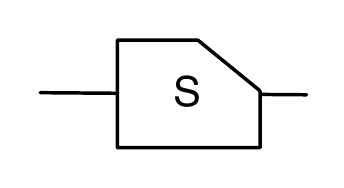}$}\text{.}
\end{equation*}
\end{proposition}

\begin{proof}
\begin{multline*}
\lower7pt\hbox{$\includegraphics[height=.8cm]{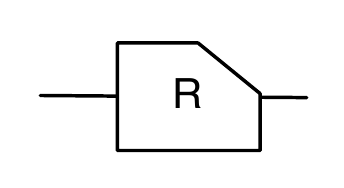}$}
=
\lower11pt\hbox{$\includegraphics[height=1.1cm]{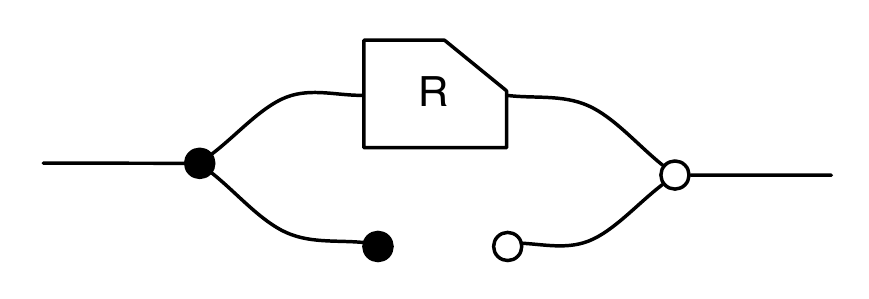}$}
=
\lower12pt\hbox{$\includegraphics[height=1.2cm]{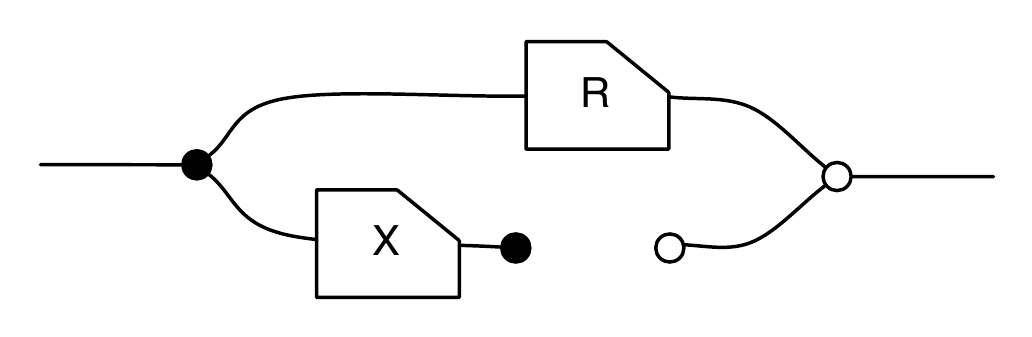}$}
=
\lower13pt\hbox{$\includegraphics[height=1.3cm]{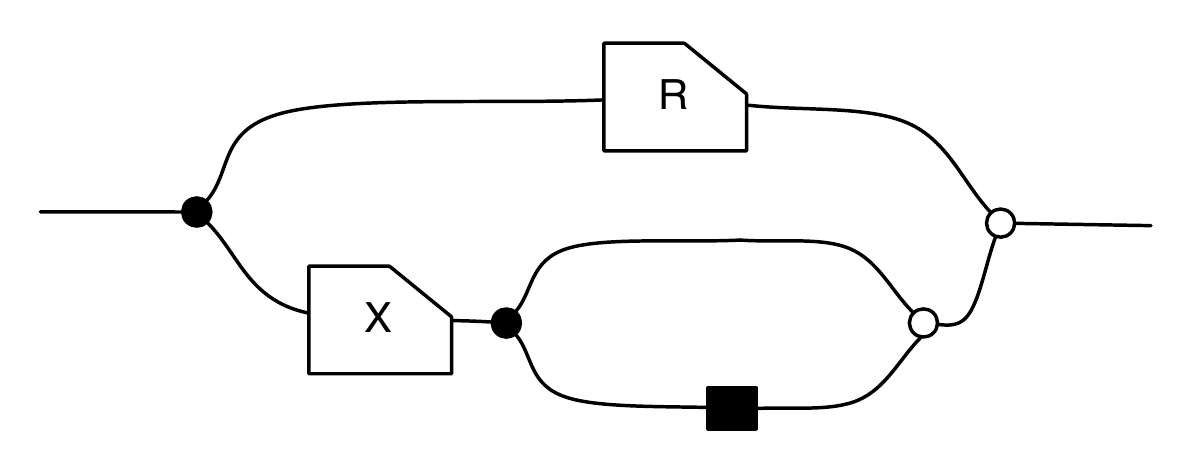}$}
\\
=
\lower15pt\hbox{$\includegraphics[height=1.5cm]{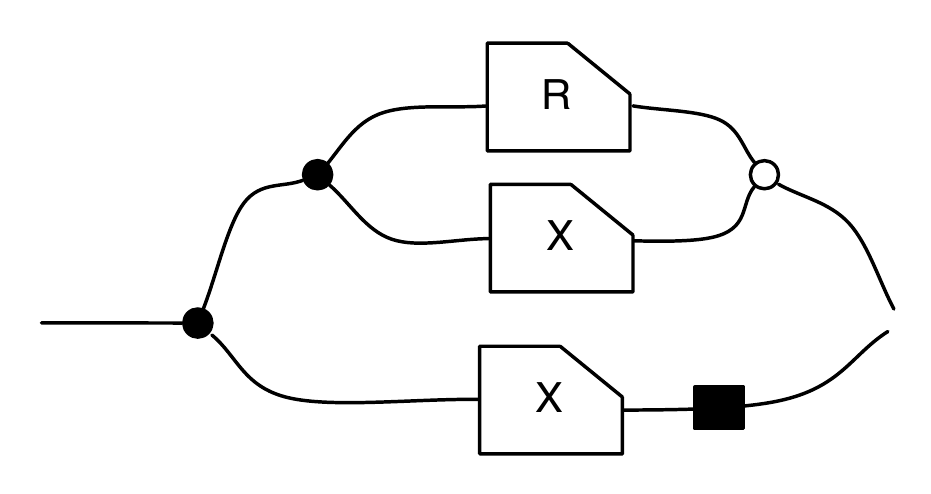}$}
=
\lower15pt\hbox{$\includegraphics[height=1.5cm]{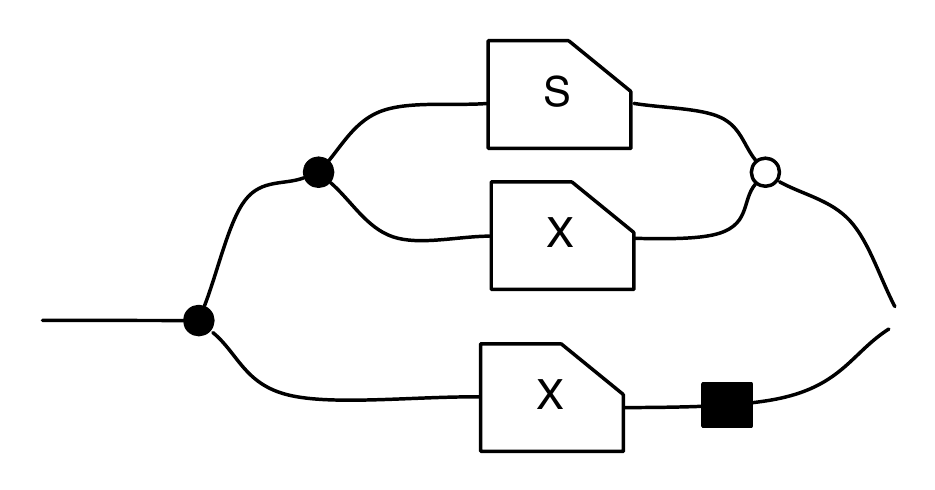}$}
= \dots =
\lower7pt\hbox{$\includegraphics[height=.8cm]{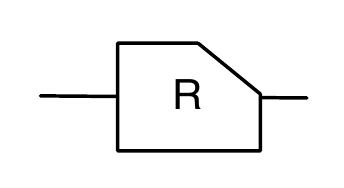}$}.
\end{multline*}
\end{proof}

\section{The theory of $\rring{R}$-Modules}\label{sec:modules}
Let $\rring{R}$ be a ring. In order to avoid distinguishing in between left and right modules, we assume $\rring{R}$ to be commutative. We leave as future work to investigate the case of non commutative rings, which could prove interesting  for the algebra of linear relations.
%(Another extremely interesting case t\ModtheoryFFo be investigated in details is the one where $\rring{R}$ is a semiring: in this case one has to additionally impose an axiom for the zero, otherwise \eqref{eq:zero} would not hold).

\bigskip

All we have to do is to follow the standard recipe of extending the theory of abelian groups with a scalar operator $\cgr[height=20pt]{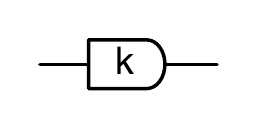}\colon 1\to 1$ for   each $k\in \rring{R}$, impose the usual four axioms for modules and, additionally,  require that scalars are maps. The inequations (of the corresponding SMIT) are summarised in Figure \ref{fig:summaryModules}. We call the underlying Frobenius theory as $\ModRtheoryF$ and the freely generated AR category as $\Frob{\ModRtheoryF}$.

\begin{figure}
\begin{center}\includegraphics[height=13cm]{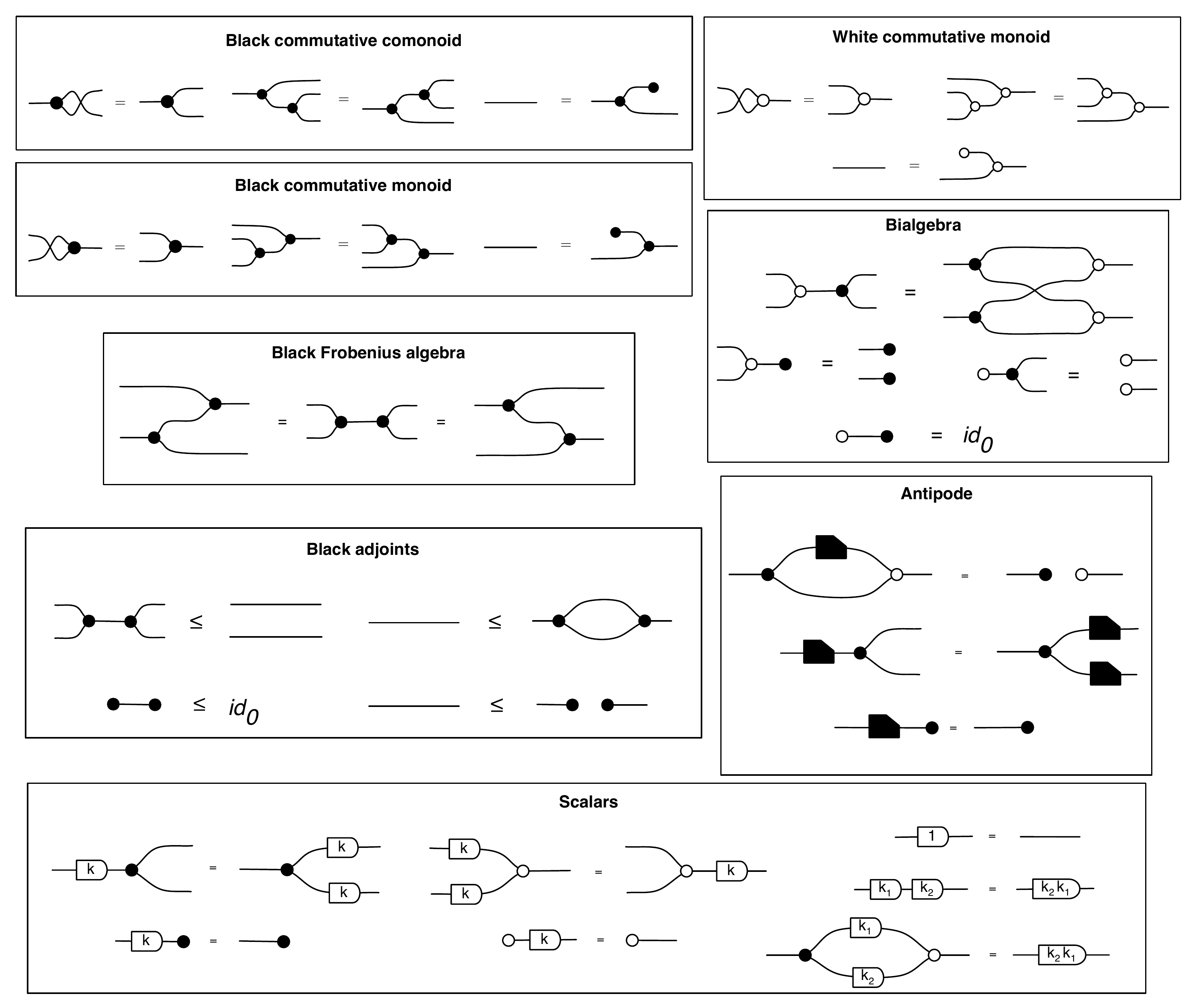}\end{center}
\caption{The SMIT corresponding to the Frobenius theory of $\rring{R}$-Modules}\label{fig:summaryModules}
\end{figure}

\medskip

At first, we observe that the following two important equalities hold.
\begin{multicols}{2}
\noindent
\begin{equation}\label{eq:zero}
\lower2pt\hbox{$
\lower7pt\hbox{$\includegraphics[height=.8cm]{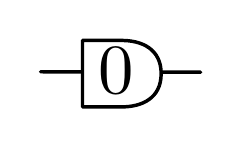}$}
=
\lower5pt\hbox{$\includegraphics[height=.6cm]{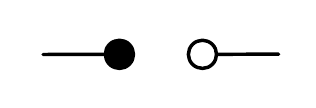}$}
$}
\end{equation}
\begin{equation}
\lower7pt\hbox{$\includegraphics[height=.8cm]{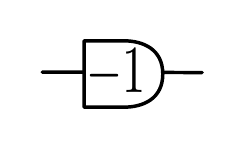}$}
=
\lower5pt\hbox{$\includegraphics[height=.6cm]{graffles/antipode.pdf}$}
\end{equation}
\end{multicols}
They can be proved by using Proposition \ref{prop:cancellation} and the two derivations below.

\begin{equation*}
\lower2pt\hbox{$
\lower15pt\hbox{$\includegraphics[height=1.1cm]{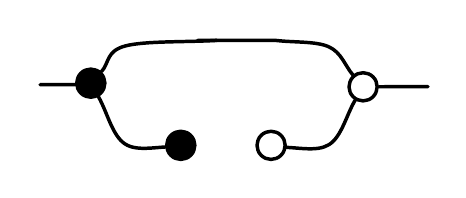}$}
=
\lower5pt\hbox{$\includegraphics[height=.6cm]{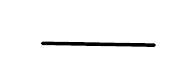}$}
=
\lower8pt\hbox{$\includegraphics[height=.8cm]{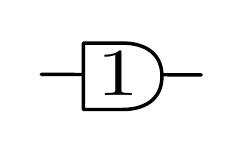}$}
=
\lower14pt\hbox{$\includegraphics[height=1.2cm]{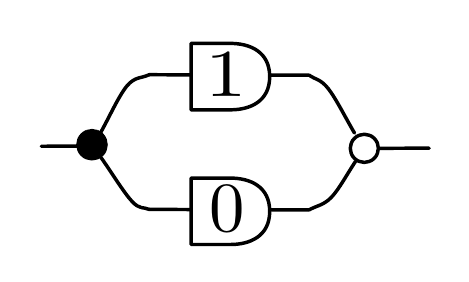}$}
=
\lower14pt\hbox{$\includegraphics[height=1.1cm]{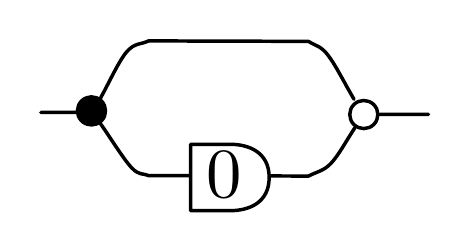}$}
$}
\end{equation*}

\begin{equation*}
\lower2pt\hbox{$
\lower15pt\hbox{$\includegraphics[height=1.2cm]{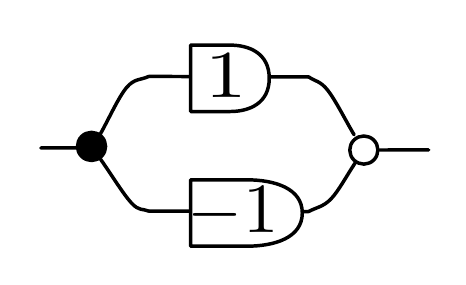}$}
=
\lower8pt\hbox{$\includegraphics[height=.8cm]{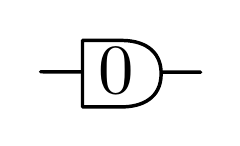}$}
=
\lower5pt\hbox{$\includegraphics[height=.6cm]{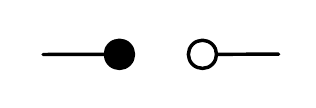}$}
=
\lower13pt\hbox{$\includegraphics[height=1.1cm]{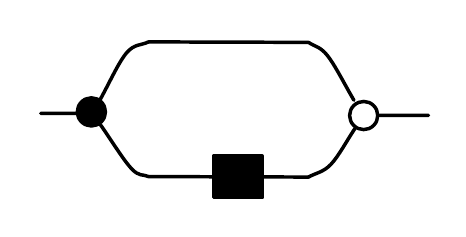}$}
=
\lower15pt\hbox{$\includegraphics[height=1.2cm]{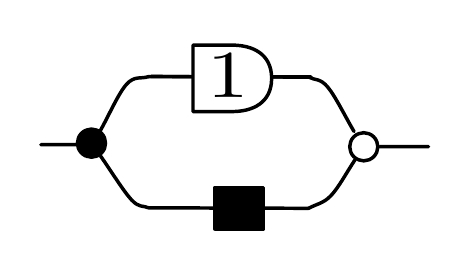}$}
$}
\end{equation*}

These two simple facts impose that the Frobenius category for the theory of abelian groups is isomorphic to the one for the theory of modules over the ring of integers $\rring{Z}$.

\begin{theorem}
$\Frob{\AGtheory} \cong \Frob{\ModtheoryF{\rring{Z}}}$ .
\end{theorem}
%\begin{proof}
%To be written down: but it is rather simple using the two equations above and the obvious encoding of scalars of $\rring{Z}$ into arrows of the theory of Abelian groups.
%\end{proof}

Observe that every scalar is -- by the (in)equations in Figure \ref{fig:summaryModules} -- total and single valued. Usually, they are not injective and surjective. However if $\rring{R}$ is a \emph{field}, every \emph{non-zero} scalar $k$ is both injective and surjective:

\begin{equation}\label{eq:inversescalarfield}
\lower2pt\hbox{$
\lower8pt\hbox{$\includegraphics[height=.8cm]{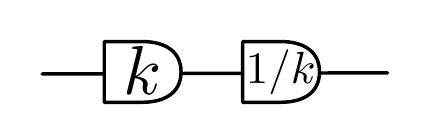}$}
=
\lower5pt\hbox{$\includegraphics[height=.6cm]{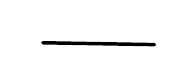}$}
=
\lower8pt\hbox{$\includegraphics[height=.8cm]{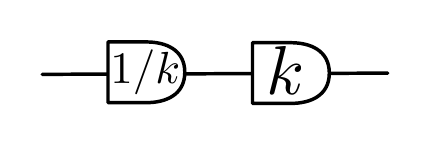}$}\text{.}
$}
\end{equation}
By Corollary \ref{cor:inverse} this means that 

\begin{equation}\label{eq:inversefield}
\lower2pt\hbox{$
\lower8pt\hbox{$\includegraphics[height=.8cm]{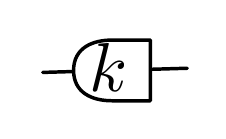}$}
=
\lower8pt\hbox{$\includegraphics[height=.8cm]{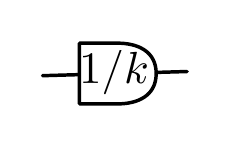}$}\text{.}
$}
\end{equation}

This is what happens in the category $\cat{IH}_{\rring{R}}$ of \emph{Interacting Hopf algebras} over $\rring{R}$ \cite{interactinghopf}.
%Since modules over a field are just vector spaces and the axioms in Figure \ref{fig:summaryModules} allows to derive all the laws of Interacting Hopf Algebras, we have the following Theorem.

\begin{theorem} If $\rring{R}$ is a field, then $\cat{IH}_{\rring{R}} \cong \Frob{\ModRtheoryF}$. 
\end{theorem}
\begin{proof}
%First, all the axioms of $\cat{IH}_{\rring{R}}$ are entailed by those in Figure \ref{fig:summaryModules}. 
%Second, all the axioms in Figure \ref{fig:summaryModules}, apart from the inequations for the black adjoint, are also axioms of $\cat{IH}_{\rring{R}}$. These inequations can be simply derived in $\cat{IH}_{\rring{R}}$ extended with the law \emph{white unit} $\leq$ \emph{black unit}. 
It is enough to check that all the equations in Figure \ref{fig:summaryModules} entails those in axioms of $\cat{IH}_{\rring{R}}$ and vice-versa.
\end{proof}

Since modules over a field are just vector spaces, we obtain a surprising result.

\begin{corollary}
If $\rring{R}$ is a field, then $\frobmodel{\cat{IH}_{\rring{R}}, \Rel}$ is the category of vector spaces and linear maps over $\rring{R}$.
\end{corollary}

\paragraph{Field of fractions.}
Now we assume $\rring{R}$ to be an ideal domain. Note that this is not necessarily \emph{principal} as required in \cite{interactinghopf}.

Since  $\rring{R}$ is an ideal domain, we can build its field of fractions $\rring{FR(R)}$: elements are pairs $p/q$ of elements of $\rring{R}$ with $q\neq 0$ quotiented by the equivalence $\equiv$ defined as $p_1 / q_1 \equiv p_2 / q_2$ iff $p_1 \times q_2  = p_2 \times q_1$. Sum and multiplication are defined as  $p_1 / q_1 + p_2 / q_2 = p_1\times q_2+ p_2\times q_1 / q_1\times q_2$ and $p_1 / q_1 \times p_2 / q_2 = p_1\times p_2 / q_1 \times q_2$. In the following, we will often use the morphism of rings $(-)/1 \colon \rring{R} \to \rring{FR(R)}$ mapping any $k\in \rring{R}$ into $k/1$.

\medskip

We now show that the theory of modules over $\rring{FR(R)}$ can be retrieved by the one of modules over $\rring{R}$ by additionally requiring that every non-zero scalar is both injective and surjective. To be entirely formal, we call $\Theory{MODIS}_{\rring{R}}$ the theory $\ModRtheoryF$ extended with the following inequations
\begin{equation}\label{eq:scalaris}
\lower2pt\hbox{$
\lower8pt\hbox{$\includegraphics[height=.8cm]{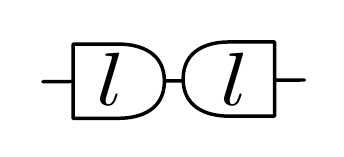}$}
\leq
\lower5pt\hbox{$\includegraphics[height=.6cm]{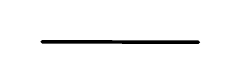}$}
\leq
\lower8pt\hbox{$\includegraphics[height=.8cm]{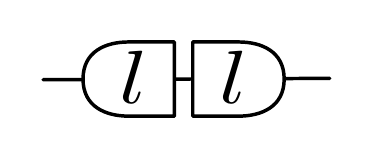}$}
$}
\end{equation}
for each scalar $l \in \rring{R}$ different than $0$.

\begin{theorem}
Let $\rring{R}$ be an ideal domain. Then $\Frob{\ModtheoryF{\rring{FR(R)}}} \cong \Frob{\Theory{MODIS}_{\rring{R}}}$.
\end{theorem}
\begin{proof}
As a first step, we define a Cartesian bifunctor $\iota \colon \Frob{\Theory{MODIS}_{\rring{R}}} \to \Frob{\ModtheoryF{\rring{FR(R)}}}$ by induction. The inductive cases are given by the fact that it should preserve $\poi$ and $\tns$. For the base cases, $\iota$ maps the scalars along  $(-)/1 \colon \rring{R} \to \rring{FR(R)}$ and for all the others, it behaves as the identity.

In order to prove that this is well-defined, we need to check that for all $R,S \in \Frob{\Theory{MODIS}_{\rring{R}}}$,
\begin{equation}\label{eq:localproof}\tag{$\alpha$}\text{ if $R\leq S$, then $\iota(R) \leq \iota(S)$.}
\end{equation} To prove this, it is enough to check that \eqref{eq:localproof} holds for each of the inequation of $\Theory{MODIS}_{\rring{R}}$. For inequations that do not involve scalars, \eqref{eq:localproof} is trivial. For the axioms that do involve scalars, we have two cases: 
\begin{enumerate}[(a)] 
\item if $R\leq S$ is an axiom in \eqref{eq:scalaris}, then it holds in $\Frob{\ModtheoryF{\rring{FR(R)}}}$  by \eqref{eq:inversescalarfield} (since $\rring{FR(R)}$ is a field); 
\item if $R\leq S$ is an axiom in Figure \ref{fig:summaryModules}, then it holds, since $(-)/1 \colon \rring{R} \to \rring{FR(R)}$ is a morphism of rings.
\end{enumerate}

The next step is to define $\kappa \colon \Frob{\ModtheoryF{\rring{FR(R)}}} \to   \Frob{\Theory{MODIS}_{\rring{R}}} $ by induction. Again the inductive cases are fixed. For the base cases, $\kappa$ maps each scalar $p/q \in \rring{FR(R)}$ into 
$$\hbox{$\includegraphics[height=.8cm]{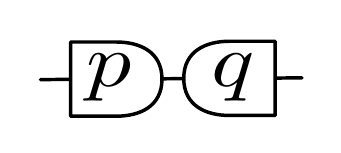}$}$$
 and for all the others, it behaves as the identity. Observe that $\equiv$ is preserved by $\kappa$: if $p_1/q_1 \equiv p_2/q_2$, then $p_1 q_2 = p_2q_1$ and therefore
\begin{equation*}
\lower2pt\hbox{$
\lower8pt\hbox{$\includegraphics[height=.8cm]{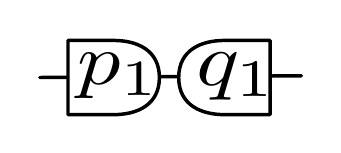}$}
=
\lower8pt\hbox{$\includegraphics[height=.8cm]{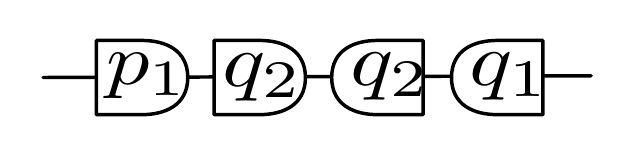}$}
=
\lower8pt\hbox{$\includegraphics[height=.8cm]{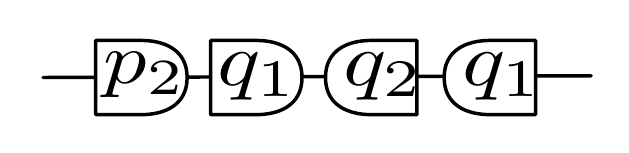}$}
=
\lower8pt\hbox{$\includegraphics[height=.8cm]{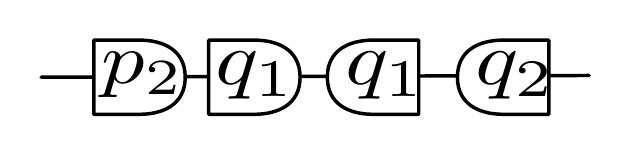}$}
=
\lower8pt\hbox{$\includegraphics[height=.8cm]{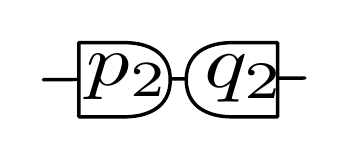}$}\text{.}
$}
\end{equation*}
Again, to prove that $\kappa$ is well defined, it is enough to inspect the axioms concerning the scalars in Figure \ref{fig:summaryModules}. For a scalar $p/q \in \rring{FR(R)}$, the two leftmost axioms hold in $ \Frob{\Theory{MODIS}_{\rring{R}}} $, since $q\neq 0$ and \eqref{eq:scalaris} entails that $\cgr[height=13pt]{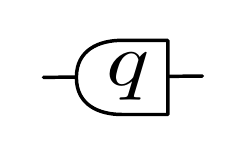}$ is a black comonoid homomorphism. For the same reason $\cgr[height=13pt]{qinv.pdf}$ is also a white monoid homomorphism and thus it is easy to see that also the fifth (from left to right) axioms in Figure \ref{fig:summaryModules} holds in $ \Frob{\Theory{MODIS}_{\rring{R}}}$.  The remaining three axioms are easily proved by using the definition of sum and multiplication in $\rring{FR(R)}$.

To conclude it is enough to prove that $\iota\circ \kappa = id$ and $\kappa \circ \iota =id$. For both, the proof proceeds by induction. The only interesting cases are the scalars: for $\kappa \circ \iota =id$, the proof is straightforward; for $\iota\circ \kappa = id$, the proof uses that \eqref{eq:inversefield} holds  in  $\Frob{\ModtheoryF{\rring{FR(R)}}}$.
\end{proof}

%\bibliographystyle{abbrv}
%\bibliography{catBib3}

\begin{thebibliography}{10}

\bibitem{BaezErbele-CategoriesInControl}
J.~Baez and J.~Erbele.
\newblock Categories in control.
\newblock {\em Theory and Application of Categories}, 30:836--881, 2015.

\bibitem{Bonchi2014b}
F.~Bonchi, P.~Sobocinski, and F.~Zanasi.
\newblock A categorical semantics of signal flow graphs.
\newblock In {\em CONCUR 2014}, volume 8704 of {\em LNCS}, pages 435--450,
  2014.

\bibitem{BialgAreFrob14}
F.~Bonchi, P.~Sobocinski, and F.~Zanasi.
\newblock Interacting bialgebras are frobenius.
\newblock In {\em Foundations of Software Science and Computation Structures -
  17th International Conference, {FOSSACS}, Proceedings}, pages 351--365, 2014.

\bibitem{interactinghopf}
F.~Bonchi, P.~Soboci\'nski, and F.~Zanasi.
\newblock Interacting {H}opf algebras.
\newblock {\em CoRR}, abs/1403.7048, 2014.

\bibitem{Bonchi2015}
F.~Bonchi, P.~Sobocinski, and F.~Zanasi.
\newblock Full abstraction for signal flow graphs.
\newblock In {\em Proceedings of the 42nd Annual {ACM} {SIGPLAN-SIGACT}
  Symposium on Principles of Programming Languages, {POPL}}, pages 515--526,
  2015.

\bibitem{bruni2003some}
R.~Bruni and F.~Gadducci.
\newblock Some algebraic laws for spans (and their connections with
  multirelations) 11research partly supported by the ec tmr network getgrats
  and by the italian murst project tosca.
\newblock {\em Electronic Notes in Theoretical Computer Science},
  44(3):175--193, 2003.

\bibitem{Bruni2006}
R.~Bruni, I.~Lanese, and U.~Montanari.
\newblock A basic algebra of stateless connectors.
\newblock {\em Theor Comput Sci}, 366:98--120, 2006.

\bibitem{Carboni1987}
A.~Carboni and R.~F.~C. Walters.
\newblock Cartesian bicategories {I}.
\newblock {\em Journal of Pure and Applied Algebra}, 49:11--32, 1987.

\bibitem{CoeckeDuncanZX2011}
B.~Coecke and R.~Duncan.
\newblock Interacting quantum observables: categorical algebra and
  diagrammatics.
\newblock {\em New Journal of Physics}, 13(4):043016, 2011.

\bibitem{coya2016corelations}
B.~Coya and B.~Fong.
\newblock Corelations are the prop for extraspecial commutative frobenius
  monoids.
\newblock {\em arXiv preprint arXiv:1601.02307}, 2016.

\bibitem{dovsen2013syntax}
K.~Do{\v{s}}en and Z.~Petri{\'c}.
\newblock Syntax for split preorders.
\newblock {\em Annals of pure and applied Logic}, 164(4):443--481, 2013.

\bibitem{freyd1990categories}
P.~J. Freyd and A.~Scedrov.
\newblock {\em Categories, allegories}, volume~39.
\newblock Elsevier, 1990.

\bibitem{hyland2007category}
M.~Hyland and J.~Power.
\newblock Lawvere theories and monads.
\newblock In {\em Plotkin Festschrift}, volume 172 of {\em ENTCS}, pages
  437--458, 2007.

\bibitem{jonsson1948representation}
B.~J{\'o}nsson and A.~Tarski.
\newblock Representation problems for relation algebras.
\newblock In {\em Bulletin of the American Mathematical Society}, volume~54,
  pages 80--80. AMER MATHEMATICAL SOC 201 CHARLES ST, PROVIDENCE, RI
  02940-2213, 1948.

\bibitem{lawvere1973metric}
F.~W. Lawvere.
\newblock Metric spaces, generalized logic, and closed categories.
\newblock {\em Rend. Sem. Mat. Fis. Milano}, 43(1):135--166, 1973.

\bibitem{LawvereOriginalPaper}
W.~F. Lawvere.
\newblock {\em {Functorial Semantics of Algebraic Theories}}.
\newblock PhD thesis, 2004.

\bibitem{MacLane1965}
S.~{Mac Lane}.
\newblock Categorical algebra.
\newblock {\em Bulletin of the American Mathematical Society}, 71:40--106,
  1965.

\bibitem{Selinger2009}
P.~Selinger.
\newblock A survey of graphical languages for monoidal categories.
\newblock {\em Springer Lecture Notes in Physics}, 13(813):289--355, 2011.

\bibitem{Sobocinski2013a}
P.~Soboci\'{n}ski.
\newblock Nets, relations and linking diagrams.
\newblock In {\em CALCO `13}, 2013.

\bibitem{zanasi2016algebra}
F.~Zanasi.
\newblock The algebra of partial equivalence relations.
\newblock {\em Electronic Notes in Theoretical Computer Science}, 325:313--333,
  2016.

\end{thebibliography}

\end{document}